\documentclass[10pt,final,journal]{IEEEtran}
\usepackage{cite}
\usepackage{amsmath,amssymb,amsfonts,bm,amsthm}
\usepackage{algorithmic}
\usepackage{graphicx}
\usepackage{subcaption}
\usepackage{textcomp}
\usepackage{xcolor}
\usepackage{mwe}
\newcommand{\E}{\mathrm{E}}

\DeclareMathOperator*{\argmax}{arg\,max}

\DeclareMathOperator{\Tr}{Tr}
\usepackage{mathtools}
\DeclarePairedDelimiter{\ceil}{\lceil}{\rceil}
\usepackage[ruled,vlined]{algorithm2e}
\DontPrintSemicolon
\usepackage{array}
\usepackage{booktabs}
\usepackage{enumerate}
\theoremstyle{definition}

 \theoremstyle{plain}

 \newtheorem{theorem}{Theorem}[section]

\newtheorem{lemma}[theorem]{Lemma}
\theoremstyle{remark}
\newtheorem{remark}{Remark}[theorem]
\newtheorem*{remark*}{Remark}

\def\BibTeX{{\rm B\kern-.05em{\sc i\kern-.025em b}\kern-.08em
    T\kern-.1667em\lower.7ex\hbox{E}\kern-.125emX}}

\newcommand{\overbar}[1]{\mkern 1.5mu\overline{\mkern-1.5mu#1\mkern-1.5mu}\mkern 1.5mu}

\begin{document}

\title{Variational Tracking and Redetection for Closely-spaced Objects in Heavy Clutter}

\author{Runze~Gan,~\IEEEmembership{Student Member,~IEEE,}
Qing~Li,~\IEEEmembership{Student Member,~IEEE,} 
Simon~J.~Godsill~\IEEEmembership{Fellow,~IEEE}
\thanks{This research is sponsored by the US Army Research Laboratory and the UK MOD University Defence Research Collaboration (UDRC) in Signal Processing under the SIGNeTS project. It is accomplished under Cooperative Agreement Number W911NF-20-2-0225. The views and conclusions contained in this document are of the authors and should not be interpreted as representing the official policies, either expressed or implied, of
the Army Research Laboratory, the MOD, the U.S. Government or the U.K. Government. The U.S. Government and U.K. Government are authorised to reproduce and distribute reprints for Government purposes notwithstanding any copyright notation herein.}
\thanks{ R. Gan, Q. Li, and S. J. Godsill are with the Engineering Department, University of Cambridge, Cambridge CB2 1PZ, U.K. e-mails: \{rg605, ql289, sjg30\}@cam.ac.uk}\vspace{-0.5em}}

\maketitle
\begingroup\renewcommand\thefootnote{\textsection}

\begin{abstract}
The non-homogeneous Poisson process (NHPP) is a widely used measurement model that allows for an object to generate multiple measurements over time. However, it can be difficult to track multiple objects efficiently and reliably under this NHPP model in scenarios with a high density of closely-spaced objects and heavy clutter. Therefore, based on the general coordinate ascent variational filtering framework, this paper presents a variational Bayes association-based NHPP tracker (VB-AbNHPP) that can efficiently track a fixed and known number of objects, perform
data association, and learn object and clutter rates with a parallelisable implementation. 
Moreover, a variational localisation strategy is proposed, which enables rapid rediscovery of missed objects from a large surveillance area under extremely heavy clutter. This strategy is integrated into the VB-AbNHPP tracker, resulting in a robust methodology that can automatically detect and recover from track loss. 
This tracker demonstrates improved tracking performance compared with existing trackers in challenging scenarios, in terms of both accuracy and efficiency.
\end{abstract}

\section{Introduction} \label{sec:intro}
In real-world tracking applications, scenarios are typically complicated and diverse, often presenting challenging situations such as high clutter density, unknown measurement rates, a large number of closely-spaced objects, and occlusion. 
When tracking objects in such adverse conditions, existing trackers may experience difficulties maintaining a high level of accuracy and efficiency simultaneously, ranging from classical methods such as the joint probabilistic data association (JPDA) filter and multiple hypothesis tracker (MHT)\cite{bar1995multitarget}, to the most recent techniques including random finite set (RFS) trackers\cite{granstrom2019poisson,granstrom2017likelihood} and message passing approaches\cite{meyer2018message,meyer2021scalable}.

One major bottleneck is the rapidly growing computational complexity of the data association with the number of objects and measurements. The non-homogeneous Poisson process (NHPP) measurement model  \cite{gilholm2005poisson}, which provides an exact association-free measurement likelihood, has recently drawn much attention. 
The original NHPP tracker \cite{gilholm2005poisson}, however, suffers from the `curse of dimensionality' due to its particle filter implementation. 
The same NHPP model was later employed in a JPDA framework \cite{yang2018linear} to simplify marginal association probabilities; however, it used crude approximations that severely impaired the tracking accuracy (see \cite{gan2022variational}).
Subsequently, an association-based NHPP (AbNHPP) measurement model was presented in \cite{li2022scalable,gan2022variational}, which reintroduces associations into the NHPP model to enable an efficient parallel sampling or a tractable structure. A sequential Markov chain Monte Carlo (SMCMC) implementation was designed in \cite{li2022scalable}, and for linear Gaussian models, a fast Rao-Blackwellised online Gibbs scheme (here referred to as the G-AbNHPP tracker) was developed in \cite{li2023adaptive} with an enhanced efficiency compared to \cite{li2022scalable}. Theoretically, these SMCMC methods can converge to optimal Bayesian filters with a large enough sample size, while in practice it can be computationally intensive when object and measurement numbers are large. 
Therefore, we proposed a high-performance AbNHPP tracker based on an efficient variational inference implementation, presented in the previous conference paper\cite{gan2022variational}. The devised VB-AbNHPP tracker is shown to achieve comparable tracking accuracy with G-AbNHPP tracker\cite{li2022scalable} while benefiting from much faster operation. 

However, in hostile environments such as heavy rain or stormy sea conditions, the clutter number in a single time step can be hundreds of times greater than the object's measurements number (see e.g., Fig. \ref{fig:allobs} and Fig. \ref{fig: highly heavy clutter}(c)). In these heavy clutter scenarios, all the above-mentioned NHPP/AbNHPP trackers are either computationally impractical\cite{li2022scalable,li2023adaptive,gilholm2005poisson} or struggle to maintain the tracking accuracy\cite{gilholm2005poisson,yang2018linear,gan2022variational}. 
Therefore, this paper proposes a novel variational localisation strategy that allows a fast redetection of missed objects from large surveillance area. Embedded with this redetection technique, our upgraded VB-AbNHPP tracker can automatically detect and recover from the track loss thus providing a robust tracking performance even in extreme cases of heavy clutter with parallelisable implementation and high efficiency.

\vspace{-1em}
\subsection{Related Work}
The RFS-based trackers, including the Probabilistic Hypothesis Density (PHD)\cite{mahler2003multitarget} and the Poisson Multiple Bernoulli Mixture (PMBM) filter\cite{granstrom2019poisson}, are amongst the most recent trackers, which provide a compact solution for joint detection and tracking.
However, these RFS-based methods often require heuristics and approximations to be feasible in real-world tracking tasks. For instance, the PHD filter can avoid the data association update step, yet it has no closed form solutions and relies on an approximated Gaussian mixture implementation for practical use. Another example is the PMBM filters, which have shown superiority over all other RFS-based trackers\cite{granstrom2019poisson}. 
Notwithstanding, all the existing PMBM filters inherit the heuristic pruning, gating and hypothesis management of the MHT framework to limit the exponential increase in the global hypotheses number\cite{granstrom2019poisson,granstrom2017likelihood}. Additionally, the most popular implementation in\cite{granstrom2019poisson} involves further approximation errors: first, it uses the k-best Murty's algorithm to truncate the number of global hypotheses; a pre-processing measurement clustering step is employed for cases that objects generate multiple measurements, e.g., under the NHPP measurement model. Although a sampling-based PMBM filter\cite{granstrom2017likelihood} was devised to reduce the measurement clustering error, the designed MCMC methods are not rigorous and can be computationally intensive. On top of that, it only keeps a truncated subset of data associations and may experience a sharp decline in tracking accuracy under the high data association uncertainty with a large number of objects and clutter. 

Alternatively, approximate inference methods\cite{bishop:2006:PRML}, such as variational inference\cite{turner2014complete,lau2016structured,gan2022variational} and (loopy) belief propagation\cite{meyer2018message,meyer2021scalable} have been actively investigated due to their promising tracking accuracy and computational efficiency. For example, the belief propagation was used to design a fast and scalable data association framework that can maintain a closed-form belief update under a point object measurement model, with guaranteed convergence despite of cycles in the graph \cite{meyer2018message}.
Although they often yield accurate tracking results in practical applications, under the NHPP measurement model where cycles also exist, the belief propagation method does not guarantee convergence. In addition, it requires a particle filter implementation \cite{meyer2021scalable}; consequently, it loses some computational advantage and proved to be slow in challenging scenes such as in \cite{li2023adaptive}. 
In contrast, coordinate ascent variational inference (CAVI)\cite{blei2017variational,bishop:2006:PRML}, also known as mean field variational inference, is another popular approximate inference method whose convergence can be guaranteed and easily monitored. Compared to MCMC sampling methods, CAVI can typically achieve a comparable performance but with a more efficient implementation, as demonstrated later in this paper. Several trackers have employed CAVI under a point object measurement model, e.g., \cite{turner2014complete,lau2016structured}. These methods apply loopy belief propagation in certain steps, hence introducing an additional layer of approximation compared with the proposed methods in this paper, especially when considered within the context of the NHPP observation model in which cycles are present.  In contrast our methods are purely based on the variational approximation and hence their convergence to a fixed point solution is guaranteed, although we do not provide theoretical results  on how close the vartiational approximation is to the true posterior distribution, as is standard in CAVI methods. Our current methods are for tracking a fixed and known number of objects under the NHPP measurement model, with convergence to a fixed point solution easily monitored.

%

Although track loss happens frequently in existing multi-object trackers, few papers have investigated a solution for retrieving the lost objects in fixed number object tracking. Several primitive track management strategies, including the M/N logic-based method and the sequential probability ratio test (SPRT) \cite{bar1995multitarget}, have been proposed for track initiation and termination, e.g., in JPDA and MHT. These initiation methods may be used to retrieve lost objects if considering relocating lost objects as detecting a new object birth.
More recently, various birth processes have been employed in many trackers \cite{granstrom2019poisson,granstrom2017likelihood,meyer2018message,meyer2021scalable,lau2016structured} to initiate the track and potentially redetect the missed objects.
A major issue that limits their ability to retrieve missed objects under heavy clutter is the low efficiency, since a large number of measurements would lead to massive potential objects, which then require considerable computational efforts to evaluate their existence probability. 
Although the potential object number can be reduced by pre-processing the measurements (e.g. clustering \cite{meyer2021scalable,granstrom2019poisson}), when handling object birth under heavy clutter such as Fig. \ref{fig:allobs} and Fig. \ref{fig: highly heavy clutter}(c), to our knowledge, the trackers in \cite{meyer2021scalable,granstrom2019poisson,granstrom2017likelihood} are still slow with too many potential objects and fail to reflect the true objects' positions. 

We note that many of the trackers reviewed above\cite{granstrom2019poisson,granstrom2017likelihood,meyer2018message,meyer2021scalable,lau2016structured,turner2014complete} can track a varying number of objects, often by employing birth and death processes for Bayesian object number estimation. This important aspect is a current topic of our research and will be reported in future publications, see also \cite{4526445,li2023scalable} for related versions using Bayesian sampling procedures. For the current context, assuming a known number of objects, we demonstrate previously unexploited potential of CAVI for demanding tracking scenarios. For example, the evidence lower bound from CAVI \cite{blei2017variational} can effectively indicate localisation quality, 
and the insights gained provide the basis for future work with unknown object numbers.

\subsection{Contributions}
In the previous conference paper \cite{gan2022variational}, we developed a standard VB-AbNHPP tracker (Algorithm 1 in \cite{gan2022variational}) that can achieve a superior tracking accuracy with a highly efficient implementation. In addition, we provided a succinct probabilistic formulation of the AbNHPP system and constructed a dynamic Bayesian network that reveals a clear dependence structure of the model. Further, we presented a unified framework for coordinate ascent variational filtering with or without static parameter learning. In contrast to other variational filtering strategies \cite{sarkka2013non,huang2017novel, vsmidl2006variational,vermaak2003variational} that are limited to different or specific dynamic systems, our framework is applicable to a general dynamic system.
This paper extends the previous work \cite{gan2022variational}, with the addition of several new aspects and extensive comparisons with existing multi-object trackers.

Our first contribution is an adaptive VB-AbNHPP tracker that can cope with unknown object and clutter rates. It thus demonstrates that our developed variational filtering framework can handle both known parameters (i.e., the standard VB-AbNHPP tracker in Algorithm 1 of \cite{gan2022variational}) and unknown static parameters cases. In particular, independent Gamma initial priors are used to model objects and clutter Poisson rates, under which tractable variational updates are maintained within our framework. Subsequently, this paper develops a tracker that can simultaneously perform tracking and Poisson rate learning tasks.
Additional parameters such as the measurement covariance can be learnt similarly if required. 

Another significant contribution is a variational object localisation strategy that can efficiently find the potential object locations in a large surveillance area under heavy clutter, such as shown in Fig. \ref{fig:allobs} and Fig. \ref{fig: highly heavy clutter}(c). With a known object rate and measurement covariance, the proposed strategy can localise the object without an informative positional prior, by using only measurements from a single time step. In particular, this variational localisation strategy proceeds by independently running multiple CAVIs, each featuring a specifically designed initialiasation that can guide the CAVI to find the most probable object location in a selected small region within the surveillance area. The efficiency of this strategy is then supported by parallelised CAVIs and/or a series of carefully selected local regions. 
Most significantly, this paper utilises the CAVI in a new way for the purposes of localisation. To our knowledge, such a localisation strategy is the first attempt to `control' the CAVI to converge to the desired local optimum, and ultimately to employ this concept to approach the global optimum of the considered problem. In contrast, most existing applications (e.g. \cite{sarkka2013non,vsmidl2006variational,vermaak2003variational,turner2014complete,lau2016structured}) simply perform the CAVI with a fixed number of iterations and directly adopt the obtained local optimum, whereas this regular application of CAVI cannot tackle the challenging localisation problems presented in this paper, since the target's posterior exhibits considerably multi-modal behaviour, owing to the heavy clutter. 

Our final contribution is a VB-AbNHPP tracker that includes a relocation strategy (VB-AbNHPP-RELO) for tracking a known number of objects, under the NHPP measurement model. 
Specifically, we propose a track loss detection procedure; once track loss is detected, the proposed variational localisation strategy can be used to relocate the objects. In this way, our VB-AbNHPP-RELO tracker can robustly track closely-spaced objects under extremely heavy clutter, allowing the missed objects to be detected and relocated in time automatically and efficiently. Moreover, it enjoys superior tracking accuracy owing to our carefully designed approximate inference paradigm. Compared to other trackers that are based on sampling or maintaining multiple hypotheses (e.g. \cite{li2023adaptive,meyer2021scalable,granstrom2019poisson}), the proposed tracker is faster, owing to its single Gaussian vector representation of the obtained object state posterior and efficient CAVI inference implementation; it can be further accelerated due to many parallelisable computational features in the algorithm. 
Results verify that the proposed VB-AbNHPP-RELO has improved tracking accuracy with faster execution time than other comparable trackers in cases of a large number of closely-spaced objects and heavy clutter.
\vspace{-0.7em}
\subsection{Paper Outline}
The remainder of paper is organised as follows. Section \ref{sec:PoissonModel} defines the association-based NHPP system and formulates the tracking problem. Section \ref{sec: cavf} introduces a general coordinate ascent variational filtering framework that allows approximate filtering with or without static parameter learning. Subsequently, we derive the VB-AbNHPP tracker with objects and clutter Poisson rates learning in Section \ref{sec: vbnhpp}. In Section \ref{sec: basic vatiational localisation}, we propose and demonstrate a novel variational localisation strategy, which can detect a single object in a large surveillance area under heavy clutter. Section \ref{sec: missed objects relocation} extends this technique to localise multiple missed objects, based on which and a proposed track loss detection strategy, we develop the VB-AbNHPP-RELO that can in timely fashion recover from track loss. Section \ref{sec: simulation} verifies the performance of both VB-AbNHPP-RELO (with a known rate) and VB-AbNHPP with rate estimation using simulated data. Finally, Section \ref{sec:conclusion} concludes the paper.

\vspace{-0.7em}
\section{Problem Formulation} \label{sec:PoissonModel}
Assume that there are a known number of $K$ objects. At time step $n$, their joint state is $X_n=[X_{n,1}^\top,X_{n,2}^\top,...,X_{n,K}^\top]^\top$, where each vector $X_{n,k}, k\in \{1,...,K\}$ denotes the kinematic state (e.g., position and velocity) for the $k$-th object.
Let $Y_n=[Y_{n,1},...,Y_{n,M_n}]$ denote measurements received at time step $n$, and $M_n$ is the total number of measurements. 

This paper is based on the NHPP measurement model proposed in \cite{gilholm2005poisson}. Denote the set of Poisson rates by $\Lambda=[\Lambda_0,\Lambda_1,...,\Lambda_K]$, where $\Lambda_0$ is the clutter rate and $\Lambda_k$ is the $k$-th object rate, $k=1,...,K$. Each object $k$ generates measurements by a NHPP with a Poisson rate $\Lambda_k$, and the total measurement process is also a NHPP from the superposition of the conditional independent NHPP measurement processes from $K$ objects and clutter. The total number of measurements follows a Poisson distribution with rate $\Lambda_{\text{sum}}=\sum_{k=0}^K\Lambda_k$.

The likelihood function deduced in \cite{gilholm2005poisson} is 
\vspace{-0.5em}
\begin{align}\label{likelihood poisson}
        h(Y_{n},M_n|{X}_{n},\Lambda) 
        &=\frac{e^{-\Lambda_{\text{sum}}}}{M_n !} \prod_{j=1}^{M_n}(\sum_{k=0}^{K}\Lambda_k  \ell(Y_{n,j}|X_{n,k})),
\end{align}
where each $Y_{n,j}$ may originate from any object or clutter. We assume the object originated measurement follows a linear Gaussian model while clutters are uniformly distributed in the observation area $V$:
\begin{equation} 
   \!\!\ell({Y}_{n,j}|{X}_{n,k})=\begin{cases} 
    \mathcal{N}({Y}_{n,j};H X_{n,k},R_{k}), & \!\!\text{$k\neq 0 $; (object)}\\
     \frac{1}{V}, & \!\!\text{$k= 0 $; (clutter)}
\end{cases}
\label{measurement model}
\end{equation}
where $X_{n,0}$ denotes the parameter/information of the clutter likelihood (e.g. in our case, the region of the uniform distribution), and in this paper, we assume it is always known. Note that $X_{n,0}$ is not included in the joint object state $X_n$.
$H$ is the observation matrix that extracts relevant measurement data (typically positional information) from an object state $X_{n,k}$. For point object $k$, $R_{k}$ is the sensor noise covariance, and for extended object, $R_{k}$ is the extended object's covariance/extent.

\subsection{Association-based NHPP Measurement Model}
Here we reformulate the NHPP measurement model by incorporating the association variables.
First, we define the association variable $\theta_{n}=[\theta_{n,1},...,\theta_{n,M_n}]$, with each component $\theta_{n,j}\in \{0, 1,...,K\}$; $\theta_{n,j}=0$ indicates that $Y_{n,j}$ is generated by clutter, and $\theta_{n,j}=k, k\in\{1,\ldots,K\}$ means that $Y_{n,j}$ is generated from the object $k$. Clearly, $\theta_{n}$ and $Y_{n}$ have the same length $M_n$. 
Subsequently, the joint distribution with the association variables $\theta_{n}$ is
\vspace{-0.1em}
\begin{equation}\label{likelihood joint} 
        p(Y_{n},M_n,\theta_{n}|{X}_{n},\Lambda)=p(Y_{n}|\theta_{n},X_{n})p(\theta_n|M_n,\Lambda)p(M_n|\Lambda),
        \vspace{-0.2em}
\end{equation}
where $p(M_n|\Lambda)$ is a Poisson distribution
\vspace{-0.3em}
\begin{align}
    p(M_n|\Lambda)
    &=\frac{\text{exp}(-\Lambda_{\text{sum}})(\Lambda_{\text{sum}})^{M_n}}{M_n!},
    \label{eq:measurement number pdf}
\end{align}
\vspace{-0.3em}
and all these $M_n$ measurements are conditionally independent
\begin{align}
    \label{eq:obs prior}
    p(Y_{n}|\theta_{n},X_{n})&=\prod_{j=1}^{M_n}\ell(Y_{n,j}|X_{n,\theta_{n,j}}),
\end{align}
where $M_n$ is known from $\theta_n$, and the function $\ell$ is the same as in \eqref{likelihood poisson} and defined in \eqref{measurement model}. For association $\theta_n$ we define
\vspace{-0.3em}
\begin{align} 
 \label{eq:assoc prior}
    &p(\theta_n|M_n,\Lambda)=\prod_{j=1}^{M_n}p(\theta_{n,j}|\Lambda),
\end{align}
where each association component is categorical distributed
\begin{align}
    \label{eq:single assoc prior}
    &p(\theta_{n,j}|\Lambda)=\frac{\sum_{k=0}^K\Lambda_k\delta[\theta_{n,j}=k]}{\Lambda_{\text{sum}}}.
\end{align}

We can see that this measurement model formulation is theoretically equivalent to the NHPP model in \cite{gilholm2005poisson}. To demonstrate it, we use the definition from \eqref{likelihood joint} to \eqref{eq:single assoc prior}; by marginalising the association $\theta_{n}$ out from $p(Y_{n},M_n,\theta_{n}|{X}_{n},\Lambda) $, we can find out that it is equal to the likelihood function in \eqref{likelihood poisson}:
\vspace{-0.2em}
\begin{equation}
    \sum_{\theta_n} p(Y_{n},M_n,\theta_{n}|{X}_{n},\Lambda)  = h(Y_{n},M_n|{X}_{n},\Lambda).\vspace{-0.2em}
\end{equation}
Therefore, these two formulations are equivalent.
\vspace{-0.3em}
\subsection{Dynamic  Bayesian  Network  Modelling}
\begin{figure}
    \centering
    \includegraphics[width=7cm]{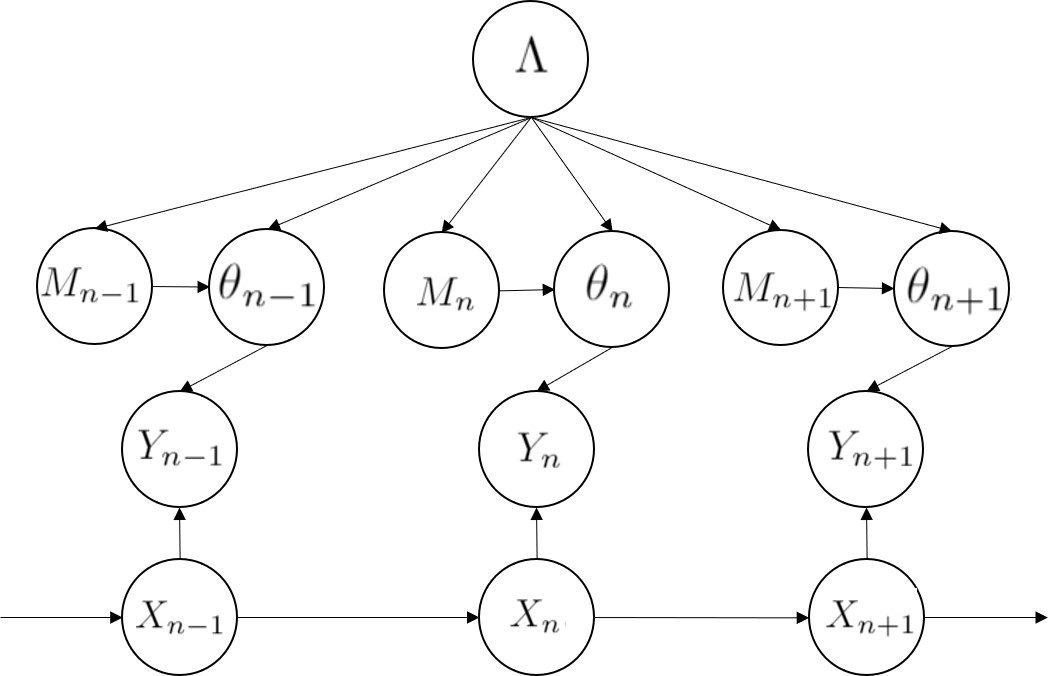}
    \caption{Dynamic Bayesian network representing the joint distribution over associations $\theta_{n}$, Poisson rates $\Lambda$, object states $X_n$, measurement number $M_n$, and measurement set $Y_{n}$ }
    \label{fig:my_label}
\end{figure}
The association-based NHPP measurement model can be combined with any dynamic models for multi-object state $X_n$ to formulate a complete dynamic Bayesian network. In this paper, we consider a standard independent linear Gaussian transition for each object state $X_{n,k}$, i.e.,
\vspace{-0.5em}
\begin{align} \notag
    p(X_n|X&_{n-1})=\prod_{k=1}^K p(X_{n,k}|X_{n-1,k})\\[-0.5em] \label{eq: dynamic transition}
    =&\prod_{k=1}^K\mathcal{N}(X_{n,k};F_{n,k}X_{n-1,k}+B_{n,k},Q_{n,k}).
\end{align}
Corresponding to the multi-object tracking problem under the association-based NHPP measurement model, the target distribution from time step $0$ to $N$ can be factorised as follows 
\vspace{-1.5em}
\begin{align}  \notag 
    p(X_{0:N},Y_{1:N}, &\theta_{1:N}, M_{1:N}|\Lambda)=p(X_0)\prod_{n=1}^Np(X_n|X_{n-1})\\ \label{eq:joint facto}
    &\times p(Y_n|\theta_n,X_n)p(\theta_n|M_n,\Lambda)p(M_n|\Lambda) \vspace{-0.3em}
\end{align}
where $p(X_n|X_{n-1})$ is given in \eqref{eq: dynamic transition}, and the other distributions are given in \eqref{eq:measurement number pdf}-\eqref{eq:single assoc prior}.  

This joint distribution can be represented by a dynamic Bayesian network (DBN) shown in Fig. \ref{fig:my_label}, based on the factorisation in \eqref{eq:joint facto}. Each conditional distribution on the right hand side of \eqref{eq:joint facto} can be depicted by the directed arrows, which point from the parent nodes to the child nodes (e.g., the arrow is from $\Lambda$ to $M_n$ for the conditional distribution $p(M_n|\Lambda)$). 

Conditional on known parameters $K,R_{1:K}$ and the transition $p(X_n|X_{n-1})$, the objective of filtering is to sequentially estimate the posterior $p(X_n,\theta_n,\Lambda|Y_{1:n})$, which is equivalent to $p(X_n,\theta_n,\Lambda|Y_{1:n},M_{1:n})$ conditional on $M_{1:n}$ since cardinality $M_n$ is inherently known once the measurements $Y_n$ are received. The exact filtering strategy for the association-based NHPP tracker can then be expressed as follows
\begin{align} \notag
    p(X_n,\theta_n,\Lambda|Y_{1:n})\propto & p(Y_n|\theta_n,X_n) p(\theta_n|M_n,\Lambda)p(M_n|\Lambda) \\\label{eq:exact filter}
    \times \int p(X_{n-1}&,\Lambda|Y_{1:n-1})p(X_n|X_{n-1})dX_{n-1},
\end{align}
where parameters $K,R_{1:K}$ are known and thus are implicitly conditioned and omitted from all relevant densities for convenience. The explicit evaluation of \eqref{eq:exact filter} is intractable since it requires enumerating all possible configurations of $\theta_{1:n}$. 
In this paper, the VB-AbNHPP tracker approximates the online filtering recursion \eqref{eq:exact filter} with the coordinate ascent variational Bayes technique.

\section{Coordinate ascent variational filtering with online parameter learning} \label{sec: cavf}
In this section, we will introduce the coordinate ascent variational filtering framework in a general setting before we apply it to the tracking problem formulated in Section \ref{sec:PoissonModel}. 
We consider a dynamic system with several static parameters, a sequence of latent states, and measurements $\mathcal{Y}_{1:n}$ from time step $1$ to $n$. 
We denote $\mathcal{Z}_n=\mathcal{X}_n\bigcup\Xi$ 
as the set of all latent variables in the system at time step $n$ that we wish to infer, where $\mathcal{X}_n$ is the set of all unknown sequential latent states at time step $n$, and $\Xi$ is the set of all unknown static parameter(s) of the system. We assume that elements in all defined sets  are distinguishable.

Note that $\mathcal{Z}_1,...,\mathcal{Z}_n$ by our definition are not mutually disjoint unless all system static parameters are known; in this special case, $\Xi=\emptyset$ and thus $\mathcal{Z}_n=\mathcal{X}_n$.
Finally, we assume that the exact optimal filtering with online parameter estimation for this system can be recursively expressed by the prediction step and the update step, which yield the following predictive prior $p_{n|n-1}(\mathcal{Z}_n)$ and posterior $p(\mathcal{Z}_n|\mathcal{Y}_{1:n})$, respectively, i.e.
\vspace{-0.5em}

\begin{align} \label{eq:general exact filter with para learn}
    \begin{aligned}
        p_{n|n-1}(\mathcal{Z}_n)=&\int f(\mathcal{X}_n|\mathcal{Z}_{n-1})p(\mathcal{Z}_{n-1}|\mathcal{Y}_{1:n-1})d\mathcal{X}_{n-1},\\
        p(\mathcal{Z}_n|\mathcal{Y}_{1:n})\propto& g(\mathcal{Y}_n,\mathcal{Z}_n)p_{n|n-1}(\mathcal{Z}_n),
    \end{aligned}
\end{align}
where $f$ is a known conditional probability density of $\mathcal{X}_n$ conditional on $\mathcal{Z}_{n-1}$. $g$ is a known function that depends on $\mathcal{Y}_n$ and $\mathcal{Z}_n$, typically proportional to the observation likelihood for the data.
Note that $f$ and $g$ may also depend on other known parameters.
$p$ is defined as the exact probability law of the considered dynamic system; here $p$ is the same probability law as defined in Section \ref{sec:PoissonModel} in the case of a NHPP system described in Section \ref{sec:PoissonModel}.

A typical example of the considered system is a general state space model where the unknown latent state $\mathcal{X}_{1:n}$ follows a first-order Markovian transition and $\Xi$ refers to the unknown parameter(s) of the transition and/or measurement function. Moreover, this system also applies to our AbNHPP framework formulated in Section \ref{sec:PoissonModel}. Specifically, $\mathcal{Y}_n$ refers to the measurements $Y_n$, $\mathcal{X}_n$ refers to the set $\{X_n,\theta_n$\}, and $\Xi$ refers to the parameter set $\{\Lambda,R_{1:K}$\} or simply $\emptyset$ depending on whether those parameters are required to be inferred.

Variational filtering can be employed when the exact filtering recursion in \eqref{eq:general exact filter with para learn} is intractable, which recursively approximates $p(\mathcal{Z}_{n}|\mathcal{Y}_{1:n})$ in \eqref{eq:general exact filter with para learn} with a converged variational distribution $q^*_n(\mathcal{Z}_{n})$ that is chosen by minimising a certain KL divergence. By convention, we divide our variational filtering framework into a prediction step and an update step. In the following subsections, we will first describe the objective and rationale of our approximate filtering strategy in Section \ref{sec:approximate filtring objective}, then clarify the details of the variational update step in Section \ref{sec:ca update}, and finally present the way to perform a reliable prediction step in Section \ref{sec:predic prior approx}. The framework's applicability will be discussed in Section \ref{sec: further disc}.

\vspace{-0.5em}
\subsection{Approximate Filtering Objective and Probability Law} \label{sec:approximate filtring objective}

Due to the unavailability of an exact form $p(\mathcal{Z}_{n-1}|\mathcal{Y}_{1:n-1})$, the exact predictive prior $p_{n|n-1}(\mathcal{Z}_n)$ in \eqref{eq:general exact filter with para learn} is intractable. Therefore, the prediction step in our filter aims to approximate the exact predictive prior $p_{n|n-1}(\mathcal{Z}_n)$ with a tractable distribution $\hat{p}_{n|n-1}(\mathcal{Z}_n)$, whose construction will be discussed in Section \ref{sec:predic prior approx}. The approximate joint probability law for the update step at the time step $n$ can be defined as $\hat{p}_{n}$, which satisfies the following factorisation:
\begin{align} \label{eq: law for approximate filtering}
    \hat{p}_{n}(\mathcal{Z}_{n},\mathcal{Y}_n)\propto g(\mathcal{Y}_n,\mathcal{Z}_n)\hat{p}_{n|n-1}(\mathcal{Z}_{n}),
\end{align}
from which the filtering conditional is obtained as
\begin{align}  \label{eq:general phatjoint}
    \hat{p}_n(\mathcal{Z}_{n}|\mathcal{Y}_n)
    &\propto   g(\mathcal{Y}_n,\mathcal{Z}_n)\hat{p}_{n|n-1}(\mathcal{Z}_{n}).
\end{align}

By comparing \eqref{eq:general phatjoint} with \eqref{eq:general exact filter with para learn}, it can be seen that the posterior under the approximate filtering probability law, i.e., $\hat{p}_n(\mathcal{Z}_{n}|\mathcal{Y}_n)$, is identical to the exact posterior $p(\mathcal{Z}_n|\mathcal{Y}_{1:n})$ if $\hat{p}_{n|n-1}(\mathcal{Z}_{n})$ equals $p_{n|n-1}(\mathcal{Z}_n)$. Therefore, $\hat{p}_{n}(\mathcal{Z}_{n}|\mathcal{Y}_n)$ is regarded as the target distribution of the approximate filtering, since it is expected to provide a close approximation to the exact posterior $p(\mathcal{Z}_{n}|\mathcal{Y}_{1:n})$ in \eqref{eq:general exact filter with para learn} if the prediction step produces an accurate $\hat{p}_{n|n-1}(\mathcal{Z}_{n})$. Subsequently, the objective of the update step in our approximate filtering is to infer the target distribution $\hat{p}_n(\mathcal{Z}_{n}|\mathcal{Y}_n)$ via \eqref{eq:general phatjoint}.  

\vspace{-0.7em}
\subsection{Update Step with Coordinate Ascent Variational Inference} \label{sec:ca update}
In the update step, we approximate the target distribution $\hat{p}_n(\mathcal{Z}_{n}|\mathcal{Y}_n)$ in \eqref{eq:general phatjoint} by a converged variational distribution $q^*_n(\mathcal{Z}_{n})$ from the standard coordinate ascent variational inference framework \cite{bishop:2006:PRML,blei2017variational}.
To this end, we first posit a (mean-field) family of variational distributions $q_n(\mathcal{Z}_{n})$. Each member of this family must satisfy the following factorisation $q_n(\mathcal{Z}_{n})=\prod_{i=1}^\nu q_n^i(z_n^i)$, where each $z_n^i$ ($i=1,2,...,\nu$) is a predefined factorised variable and is disjointly partitioned from $\mathcal{Z}_{n}$, i.e. $\{z_{n}^{1},z_{n}^{2},...,z_{n}^{\nu}\}=\mathcal{Z}_{n}$. Then our variational distribution $q^*_n(\mathcal{Z}_{n})$ is chosen from the posited family that maximises the following evidence lower bound (ELBO)\footnote{
The ELBO $\mathcal{F}_n(q_n)$ in this article may not be a lower bound of the log-evidence $\log\hat{p}_n(\mathcal{Y}_n)$ traditionally described in literature, e.g. \cite{blei2017variational}. Since our $g(\mathcal{Y}_n,\mathcal{Z}_n)$ is not limited to a normalised likelihood, the ELBO $\mathcal{F}_n(q_n)$ in \eqref{eq:general ELBO} is a lower bound of $\log\hat{p}_n(\mathcal{Y}_n)$ plus a logarithm of the normalisation constant from \eqref{eq: law for approximate filtering}.

} $\mathcal{F}_n(q_n)$:
\begin{align}  
     \label{eq:general ELBO}
    \mathcal{F}_n(q_n)&=\E_{q_n(\mathcal{Z}_n)}\log \frac{g(\mathcal{Y}_n,\mathcal{Z}_n)\hat{p}_{n|n-1}(\mathcal{Z}_{n})}{q_n(\mathcal{Z}_n)}.
\end{align}
The rationale behind this optimisation is to minimise the KL divergence KL$(q_n(\mathcal{Z}_n)||\hat{p}_n(\mathcal{Z}_{n}|\mathcal{Y}_n))$, since according to \eqref{eq:general phatjoint}, the negative of this KL divergence can be expressed as the ELBO $\mathcal{F}_n(q_n)$ in \eqref{eq:general ELBO} plus a constant independent of $q_n(\mathcal{Z}_n)$.

The optimisation of $\mathcal{F}_n(q_n)$ in \eqref{eq:general ELBO} with respect to $q_n$ can be done by the following coordinate ascent algorithm. We start by setting $q_n^{i}(z_n^i)$ to an initialised distribution $q_{n}^{(0)}(z_{n}^{i})$ for all $i=1,2,...,\nu$; then we iteratively update $q_n^{i}$ for each $i=1,2,...,\nu$ according to \eqref{eq:general ca update} while keeping $q_n^{i-}(z_n^{i-})$ fixed, where $q_n^{i-}(z_n^{i-})$ is defined as the joint variational distribution of all variables in $\mathcal{Z}_n$ except $z_n^i$, i.e. $q_n^{i-}(z_n^{i-})=\prod_{j\neq i} q_n^{j}(z_n^{j})$.
\begin{align} \label{eq:general ca update}
    q_n^{i}(z_{n}^{i})\propto\text{exp}\left(\E_{q_n^{i-}(z_n^{i-})}\log g(\mathcal{Y}_n,\mathcal{Z}_n)\hat{p}_{n|n-1}(\mathcal{Z}_{n})\right).
\end{align}
The $q_n^{i}$ in \eqref{eq:general ca update} is the optimal distribution that achieves the highest ELBO $\mathcal{F}_n(q_n)$ when $q_n^{i-}$ is fixed. Each update via \eqref{eq:general ca update} guarantees an increment of $\mathcal{F}_n(q_n)$ so that the algorithm eventually finds a local optimum. Such an optimisation procedure is known as CAVI.
More details about CAVI, including the derivation of \eqref{eq:general ca update}, can be found in \cite{bishop:2006:PRML,blei2017variational}.

The convergence of the CAVI can be assessed by monitoring the ELBO $\mathcal{F}_n(q_n)$ in \eqref{eq:general ELBO} for each iteration of updates. Since each update via \eqref{eq:general ca update} guarantees an increment of $\mathcal{F}_n(q_n)$, when the increment of $\mathcal{F}_n(q_n)$ is smaller than a certain threshold, we assume CAVI has converged, and the latest updated $q_n(\mathcal{Z}_n)$, that is, the converged variational distribution, is chosen as our approximate filtering result $q^*_n(\mathcal{Z}_{n})$.
\vspace{-0.7em}
\subsection{Prediction Step with Approximate Filtering Prior} \label{sec:predic prior approx}
Recall that the objective of our prediction step is to find a $\hat{p}_{n|n-1}(\mathcal{Z}_n)$ that approximates the exact predictive prior $p_{n|n-1}(\mathcal{Z}_n)$ in \eqref{eq:general exact filter with para learn}, whose intractability arises from the lack of an exact filtering prior $p(\mathcal{Z}_{n-1}|\mathcal{Y}_{1:n-1})$. 
To obtain a tractable predictive prior, we replace this intractable $p(\mathcal{Z}_{n-1}|\mathcal{Y}_{1:n-1})$ in the integrand of \eqref{eq:general exact filter with para learn} with some tractable approximate filtering prior $r(\mathcal{Z}_{n-1})$. A natural choice for such an approximate filtering prior is $q^*_{n-1}(\mathcal{Z}_{n-1})$,
which is the converged variational distribution obtained from the approximate filtering at the time step $t_{n-1}$, and subsequently our approximate predictive prior can be written as \eqref{eq:predictive prior approx pure filter}. However, in some cases, a more sophisticated construction of the approximate filtering prior is required to ensure a reliable static parameter estimator. Below, we discuss this issue by considering two different scenarios depending on whether the online estimation of the static parameters is required.

\subsubsection{$\Xi=\emptyset$}
In this case, we assume that  all static parameters in the system are known. Thus, we have $\mathcal{Z}_{n-1}=\mathcal{X}_{n-1}$, and our algorithm focuses solely on performing the approximate filtering task.  To accomplish this, we suggest using the standard approximate filtering prior $r(\mathcal{Z}_{n-1})=q^*_{n-1}(\mathcal{Z}_{n-1})$ to evaluating the $\hat{p}_{n|n-1}(\mathcal{Z}_n)$ for the prediction step:
\begin{align} \label{eq:predictive prior approx pure filter}
    \hat{p}_{n|n-1}(\mathcal{Z}_n)=&\int f(\mathcal{X}_n|\mathcal{Z}_{n-1})q^*_{n-1}(\mathcal{Z}_{n-1})d\mathcal{X}_{n-1}.\\ \notag
    =&\int f(\mathcal{X}_n|\mathcal{X}_{n-1})q^*_{n-1}(\mathcal{X}_{n-1})d\mathcal{X}_{n-1}.
\end{align}
In other words, the converged variational distribution for the last approximate filtering step is employed as the prior for our current prediction step. Such an empirical approximate filtering prior is commonly used in many approximate filtering methods such as the extended Kalman filter, and other Gaussian approximate filters \cite{sarkka2013bayesian}.
\subsubsection{$\Xi\neq\emptyset$} \label{sec:prior choice for para learn}
Our algorithm in this setting performs the approximate filtering along with online Bayesian parameter learning for the static parameters $\Xi$. Specifically, the parameter posterior $p(\Xi|\mathcal{Y}_{1:n})$ is approximated by $q_n^*(\Xi)=\int q_n^*(\Xi,\mathcal{X}_{n}) d\mathcal{X}_n$.
It has been observed that applying the standard approximate filtering prior in \eqref{eq:predictive prior approx pure filter} to this setting (i.e. set $\hat{p}_{n|n-1}(\mathcal{Z}_{n})=\int f(\mathcal{X}_n|\Xi,\mathcal{X}_{n-1})q_{n-1}^*(\Xi,\mathcal{X}_{n-1})d\mathcal{X}_{n-1}$ for all time steps) may cause the variance of $q_n^*(\Xi)$ to be very small before its mean converges to the ground truth of $\Xi$; such an overconfident $q_n^*(\Xi)$ renders it very difficult to correct our estimation of $\Xi$ with future data. To mitigate this issue, it is empirically helpful to slightly `flat' $q_n^*(\Xi)$ so that it can be less confident. Therefore, we construct our approximate filtering prior $r(\mathcal{Z}_{n-1})$ as follows,
\begin{align} \label{eq:filter prior approx with para}
    r(\mathcal{Z}_{n-1})&\propto q_{n-1}^*(\Xi)^{\gamma_{n-1}} q_{n-1}^*(\mathcal{X}_{n-1}|\Xi),
\end{align}
where $q_{n-1}^*(\mathcal{X}_{n-1}|\Xi)=q_{n-1}^*(\mathcal{X}_{n-1},\Xi)/q_{n-1}^*(\Xi)$, and henceforth the approximate predictive prior $\hat{p}_{n|n-1}(\mathcal{Z}_n)$ is 
\begin{align} \label{eq:pred prior approx with para}
        \hat{p}_{n|n-1}(\mathcal{Z}_n)&=\int f(\mathcal{X}_n|\mathcal{Z}_{n-1})r(\mathcal{Z}_{n-1})d\mathcal{X}_{n-1}\\ \notag 
        \propto q_{n-1}^*&(\Xi)^{\gamma_{n-1}} \int f(\mathcal{X}_n|\Xi,\mathcal{X}_{n-1})q_{n-1}^*(\mathcal{X}_{n-1}|\Xi) d\mathcal{X}_{n-1}.
\end{align}
$\gamma_{n-1}\in(0,1]$ is a forgetting factor that can `smooth' the density $q_{n-1}^*(\Xi)$ when used in the approximate filtering prior $r(\mathcal{Z}_{n-1})$. A higher value of $\gamma_n$ implies that the new parameter prior $r(\Xi)$ closely resembles the previously obtained posterior $q_{n-1}^*(\Xi)$. Specifically, $\gamma_{n-1}=1$ recovers the standard approximate filtering prior $q^*_{n-1}(\mathcal{Z}_{n-1})$ as used in \eqref{eq:predictive prior approx pure filter}, and in the limiting case $\gamma_{n-1}=0$, the approximate filtering prior $r(\mathcal{Z}_{n-1})$ completely ignores the $q_{n-1}^*(\Xi)$ and assumes a uniformly distributed $\Xi$.

Subsequently, setting $\{\gamma\}_n$ to a  value slightly less than 1 introduces a degree of uncertainty into the prior, making it suitable for estimating slowly varying or static parameters, though this estimation would maintain a non-zero level of uncertainty over time. For static parameter estimation, we recommend an increasing sequence $\{\gamma\}_n\in (0,1]$ that approaches $1$ as $n$ increases. This approach will intuitively lead to a variational approximation with smaller and smaller uncertainty as $n$ increases, while allowing effective convergence to an appropriate parameter region for smaller $n$.

\subsection{Further Discussions on the Applicability} \label{sec: further disc}

The introduced coordinate ascent variational filtering with parameter learning framework is devised for a general dynamic system that satisfies the optimal filtering recursion in \eqref{eq:general exact filter with para learn}. However, not all such systems can lead to closed-form approximated distributions through the introduced framework. The key to yielding a tractable approximate distribution is to make sure that the choice of the factorised variable $z_n^i (i=1,2,...,\nu)$ and the $\hat{p}_{n|n-1}(\mathcal{Z}_n)$ produced by our prediction step should lead to an analytically tractable updated variational distribution $q_n^i$ in \eqref{eq:general ca update}. This may require the predictive prior $\hat{p}_{n|n-1}(\mathcal{Z}_n)$ to take some particular parametric form that does not agree with the two constructions of $\hat{p}_n(\mathcal{Z}_n)$ we suggested in Section \ref{sec:predic prior approx} (i.e. \eqref{eq:predictive prior approx pure filter} and \eqref{eq:pred prior approx with para}). In this case, we could adopt a $\hat{p}_{n|n-1}(\mathcal{Z}_n)$ in the desired parametric form, but ensure that it matches the moments of the predictive distribution suggested in Section \ref{sec:predic prior approx} to preserve some accuracy. Alternatively, any intractable coordinate ascent update in \eqref{eq:general ca update} may be approximated by the Monte Carlo methods. In particular, we can sample from the intractable variational distributions to approximate other updated variational distributions in a closed-form expression. Such a strategy is commonly seen in the Monte Carlo based CAVI, e.g. \cite{ye2020monte,Gan2022}.

\section{Variational Bayes AbNHPP Tracker}\label{sec: vbnhpp}
The approximate inference for the association based NHPP system in Section \ref{sec:PoissonModel} can now be carried out within the coordinate ascent variational filtering framework in Section \ref{sec: cavf}. In particular, our VB-AbNHPP tracker with known static system parameters (i.e. $\Xi=\emptyset$) was presented in \cite{gan2022variational} for tracking tasks with known $\Lambda$ and $R_{1:K}$. In this section, we demonstrate that our VB-AbNHPP tracker can also learn online the unknown parameters within the variational filtering framework presented in Section \ref{sec: cavf}. We will consider the tracking tasks with the unknown Poisson rate $\Lambda$. The measurement covariance/extent $R_{1:K}$ can be learnt in a similar fashion if unknown, but for now we assume it is known for simplicity.

From now on, we assume $K,R_{1:K}$ in the association based NHPP system in Section \ref{sec:PoissonModel} are all known, and unless otherwise stated, these parameters are always implicitly conditioned. Subsequently, the variables of the general dynamic system in Section \ref{sec: cavf} are now $\Xi=\{\Lambda\}$, $\mathcal{Y}_n=Y_n$, $\mathcal{X}_{n}=\{X_n,\theta_n\}$ and $\mathcal{Z}_{n}=\{X_n,\theta_n,\Lambda\}$. Comparing the optimal filter recursion \eqref{eq:general exact filter with para learn} to the exact optimal filter recursion for the association based NHPP system in \eqref{eq:exact filter}, the densities for $f,g$ become
\begin{align} \notag
    f(\mathcal{X}_n|\mathcal{Z}_{n-1})=&p(\theta_n|M_n,\Lambda)p(X_n|X_{n-1}),\\
    g(\mathcal{Y}_n,\mathcal{Z}_n)=&p(Y_n|\theta_n,X_n) p(M_n|\Lambda).
\end{align}
Note that $f$ and $g$ are known when the approximate filtering is carried out at time step $n$, where $M_n$ is inherently obtained from the latest received observations $Y_n$. Since $\Xi\neq\emptyset$, we construct our approximate predictive prior $\hat{p}_{n|n-1}(\mathcal{Z}_{n})$ according to \eqref{eq:pred prior approx with para} in Section \ref{sec:prior choice for para learn}, i.e.
\begin{align} \notag
    \hat{p}_{n|n-1}(X_n,\theta_n,&\Lambda)\propto q_{n-1}^*(\Lambda)^{\gamma_{n-1}} p(\theta_n|M_n,\Lambda) \\ \label{eq: full form approx prior}
    \times&\int p(X_n|X_{n-1})q_{n-1}^*(X_{n-1}|\Lambda) dX_{n-1},
\end{align}
where the forgetting factor $\gamma_{n-1}\in(0,1]$ was introduced in Section \ref{sec:prior choice for para learn}. Now we assume the factorisation for our mean-field family is
$q(X_n,\theta_n,\Lambda)=q(X_n)q(\theta_n)q(\Lambda)$. This factorisation will later result in tractable coordinate ascent updates, as we will see later. It also suggests that $q_{n-1}^*(X_{n-1}|\Lambda)=q_{n-1}^*(X_{n-1})$. Subsequently, the approximate predictive prior in \eqref{eq: full form approx prior} can be rewritten in a factorised form as follows for the convenience of later derivation, 
\begin{align} 
\notag
    \hat{p}_{n|n-1} (X_n,\theta_n,\Lambda)\!=&\hat{p}_{n|n-1}(X_n)\hat{p}_{n|n-1}(\Lambda)\hat{p}_{n|n-1}(\theta_n|\Lambda),\\  \notag
    \hat{p}_{n|n-1}(X_n)=\int &p(X_n|X_{n-1})q_{n-1}^*(X_{n-1}) dX_{n-1},\\ \label{eq: approx prior factorize and predictive prior}
    \hat{p}_{n|n-1}(\theta_n|\Lambda)=&p(\theta_n|M_n,\Lambda),\\ \notag
    \hat{p}_{n|n-1}(\Lambda)\propto&q_{n-1}^*(\Lambda)^{\gamma_{n-1}}.
\end{align}
According to Section \ref{sec:approximate filtring objective} and \eqref{eq: law for approximate filtering}, the approximate filtering law $\hat{p}_n$ of the update step for $Y_n,X_n,\theta_n,\Lambda$ is defined by,
\begin{align} \notag
    \hat{p}_n(X_n,\theta_n,\Lambda,Y_n)&\propto  p(Y_n|\theta_n,X_n) p(M_n|\Lambda)\\ \label{eq: prob law}
    \times&p(\theta_n|M_n,\Lambda)\hat{p}_{n|n-1}(X_n)\hat{p}_{n|n-1}(\Lambda).
\end{align}
The prediction step described in Section \ref{sec:predic prior approx} is now used to evaluate the approximate filtering priors in \eqref{eq: approx prior factorize and predictive prior}. Since $\hat{p}_{n|n-1}(\theta_n|\Lambda)$ in \eqref{eq: approx prior factorize and predictive prior} can be directly obtained from our model in \eqref{eq:assoc prior}, the prediction step only requires evaluating $\hat{p}_{n|n-1}(X_n)$ and $\hat{p}_{n|n-1}(\Lambda)$ in \eqref{eq: approx prior factorize and predictive prior}. Their explicit forms will be given in \eqref{eq:state predictive prior computation} and \eqref{eq:rate predictive prior computation} after a discussion of the conjugate prior. 
\vspace{-0.7em}
\subsection{Coordinate Ascent Update}
Recall that our mean-field family satisfies $q(X_n,\theta_n,\Lambda)=q(X_n)q(\theta_n)q(\Lambda)$. Based on Section \ref{sec:ca update}, the update step of our tracker aims to minimise the KL divergence $\text{KL}(q_n(X_n)q_n(\theta_n)q(\Lambda)||\hat{p}_n(X_n,\theta_n,\Lambda|Y_n))$, or equivalently, to maximise the ELBO in \eqref{eq:general ELBO} as follows
\begin{align} \notag
    \mathcal{F}(q_n&)=\E_{q_n(X_n)q_n(\theta_n)q_n(\Lambda)}\log\frac{\hat{p}_{n|n-1}(X_n,\theta_n,\Lambda)}{q_n(X_n)q_n(\theta_n)q_n(\Lambda)}\\ \label{eq: with para ELBO}
    &+\E_{q_n(X_n)q_n(\theta_n)q_n(\Lambda)}\log p(Y_n|\theta_n,X_n) p(M_n|\Lambda),
\end{align}
where $\hat{p}_{n|n-1}(X_n,\theta_n,\Lambda)$ is specified in \eqref{eq: approx prior factorize and predictive prior}. To this end, it iteratively updates $q_n(X_n)$, $q_n(\Lambda)$ and $q_n(\theta_n)$ according to \eqref{eq:general ca update} until convergence. We now derive these updates.

\subsubsection{update for $q_n(X_n)$} \label{sec:ca update for Xn}
According to \eqref{eq:general ca update}, \eqref{eq:obs prior} and \eqref{measurement model}, and the fact that $q_n(\theta_n)=\prod_{j=1}^{M_n}q_n(\theta_{n,j})$ (this will later be shown in \eqref{eq:vari theta independent}), it yields the following update for $X_n$
\begin{align} \notag
    q_n(X_n)\propto& \hat{p}_{n|n-1}(X_n)\text{exp}\left(\E_{q_n(\theta_n)}\log p(Y_n|\theta_n,X_n) \right)\\ \label{eq: X update}
    \propto& \hat{p}_{n|n-1}(X_n)\prod_{k=1}^{K}\mathcal{N}\left(\overbar{Y}_n^k;HX_{n,k},\overbar{R}^k_n\right),
\end{align}
where
 \vspace{-0.5em}
\begin{align}
\label{eq:pseudomeas R}
    \overbar{R}_n^k=&\frac{R_k}{\sum_{j=1}^{M_n}q_n(\theta_{n,j}=k)},\\ \label{eq:pseudomeas Y}
    \overbar{Y}_n^k=&\frac{\sum_{j=1}^{M_n}Y_{n,j}q_n(\theta_{n,j}=k)}{\sum_{j=1}^{M_n}q_n(\theta_{n,j}=k)}.
\end{align}
The derivation of the required expectation in \eqref{eq: X update}-\eqref{eq:pseudomeas Y},  i.e. $\E_{q_n(\theta_n)}\log p(Y_n|\theta_n,X_n)$, is presented as \eqref{eq: elbo first line third part inner}-\eqref{eq: elbo first line third part inner final} in Appendix \ref{apx: elbo derivation}. Such an update can be considered as updating the (predictive) prior $\hat{p}_{n|n-1}(X_n)$ in \eqref{eq: approx prior factorize and predictive prior} with $K$ pseudo-measurements $\overbar{Y}_{n}^k,k=1,2,...,K$  (defined in \eqref{eq:pseudomeas Y}), each generated independently from each object with a measurement covariance of \eqref{eq:pseudomeas R}. The conjugate prior for such an update is Gaussian. In fact, with the independent linear Gaussian transition $p(X_n|X_{n-1})$ in \eqref{eq: dynamic transition} and an independent initial Gaussian prior $p(X_0)=\prod_{k=1}^K p(X_{0,k})$, the updated variational distribution can always maintain an independent Gaussian form for each object (i.e. $q_n(X_n)=\prod_{k=1}^K q_n(X_{n,k})$).

Specifically, if we denote the converged variational distribution for the $k$-th object at time step $n-1$ as $q^*_{n-1}(X_{n-1,k})= \mathcal{N}(X_{n-1,k};\mu^{k*}_{n-1|n-1},\Sigma^{k*}_{n-1|n-1})$, then its predictive prior is
\begin{align}
    \begin{aligned} \label{eq:state predictive prior computation}
    \hat{p}_{n|n-1}(X_{n,k})=&\mathcal{N}(X_{n,k};\mu^{k*}_{n|n-1},\Sigma^{k*}_{n|n-1}),\\
    \mu^{k*}_{n|n-1}=&F_{n,k}\mu^{k*}_{n-1|n-1}+B_{n,k},\\
    \Sigma^{k*}_{n|n-1}=&F_{n,k}\Sigma^{k*}_{n-1|n-1}F_{n,k}^\top+Q_{n,k},
    \end{aligned}
\end{align}
where $F_{n,k},B_{n,k},Q_{n,k}$ are given in \eqref{eq: dynamic transition}. The variational distribution $q_n(X_{n,k})$ can now be updated by the Kalman filter:
\begin{align}
\begin{aligned} \label{eq:KF update}
    q_n(X_{n,k})&=\mathcal{N}(X_{n,k};\mu^{k}_{n|n},\Sigma^{k}_{n|n}),\\
    T^k_{n}&=\overbar{Y}_{n}^k-H\mu_{n|n-1}^{k*},\\
    S^k_{n}&=H\Sigma^{k*}_{n|n-1}H^\top+\overline{R}_n^k,\\
    G&=\Sigma^{k*}_{n|n-1}H^\top S^{k^{-1}}_{n},\\
    \mu^{k}_{n|n}&=\mu^{k*}_{n|n-1}+GT^k_{n},\\
    \Sigma^{k}_{n|n}& =(I-GH)\Sigma^{k*}_{n|n-1}.
\end{aligned}
\end{align}
Such an update can be independently carried out for all objects.
\subsubsection{update for $q_n(\Lambda)$}
According to \eqref{eq:general ca update}, \eqref{eq:measurement number pdf} and \eqref{eq:assoc prior}, the variational distribution $q_n(\Lambda)$ is updated as follows,
\begin{align} \notag
    q_n(\Lambda)\propto &\hat{p}_{n|n-1}(\Lambda)p(M_n|\Lambda)\text{exp}\left(\E_{q_n(\theta_n)}\log p(\theta_n|M_n,\Lambda)\right)\\ \notag
    =&\hat{p}_{n|n-1}(\Lambda)\frac{\text{exp}(-\Lambda_{\text{sum}})(\Lambda_{\text{sum}})^{M_n}}{M_n!}\times\frac{1}{(\Lambda_{\text{sum}})^{M_n}}\\ \notag
    & \ \times \text{exp}\biggl(\E_{q_n(\theta_n)}\sum_{j=1}^{M_n}\sum_{k=0}^K \delta[\theta_{n,j}=k] \log\Lambda_k\biggr)\\ \label{eq: rate update}
    \propto&\hat{p}_{n|n-1}(\Lambda)\prod_{k=0}^K \text{exp}(-\Lambda_{k}) \Lambda_k^{\sum_{j=1}^{M_n} q_n(\theta_{n,j}=k)}.
\end{align}
Independent gamma turns out to be the conjugate prior for the likelihood function in \eqref{eq: rate update}. Specifically, if we assume an independent initial prior $p(\Lambda)=\prod_{k=0}^K p(\Lambda_k)$ where each $p(\Lambda_k)$ is a gamma distribution, then with prediction in \eqref{eq: approx prior factorize and predictive prior} and update in \eqref{eq: rate update}, $q_n(\Lambda)$ is always in the independent Gamma form, i.e. $q_n(\Lambda)=\prod_{k=0}^Kq_n(\Lambda_k)$ where each $q_n(\Lambda_k)$ is a gamma distribution. Denote $q_{n-1}^*(\Lambda_k)=\mathcal{G}(\Lambda_k;\eta^{k*}_{n-1|n-1},\rho^{k*}_{n-1|n-1})$, where $\mathcal{G}(\eta,\rho)$ is the Gamma distribution with shape parameter $\eta$ and scale parameter $\rho$. According to $\hat{p}_{n|n-1}(\Lambda)$ in \eqref{eq: approx prior factorize and predictive prior}, the predictive prior for each $\Lambda_k$ is
\begin{align} \notag
    \hat{p}_{n|n-1}(\Lambda_k)=&\mathcal{G}(\Lambda_k;\eta^{k*}_{n|n-1},\rho^{k*}_{n|n-1}),\\ \label{eq:rate predictive prior computation}
    \eta^{k*}_{n|n-1}=&\eta^{k*}_{n-1|n-1}\gamma_{n-1}-\gamma_{n-1}+1,\\ \notag
    \rho^{k*}_{n|n-1}=&\rho^{k*}_{n-1|n-1}/\gamma_{n-1},
\end{align}
then the update for $q_n(\Lambda_k)$ is 
\begin{align} \notag
    q_n(\Lambda_k)=&\mathcal{G}(\Lambda_k;\eta^{k}_{n|n},\rho^{k}_{n|n}),\\ \label{eq:independent rate update}
    \eta^{k}_{n|n}=&\eta^{k*}_{n|n-1}+\sum_{j=1}^{M_n} q_n(\theta_{n,j}=k),&  \\ \notag
    \rho^{k}_{n|n}=&\rho^{k*}_{n|n-1}/(\rho^{k*}_{n|n-1}+1). 
\end{align}
Consequently, each $q_n(\Lambda_k)$ can be updated independently.
\subsubsection{update for $q_n(\theta_n)$}
The variational distribution $q_n(\theta_n)$ can be updated according to \eqref{eq:general ca update}, \eqref{eq:obs prior} and \eqref{eq:assoc prior}, i.e.
\begin{align} \notag
    q_n(\theta_n)\! &\propto \text{exp}\left(\E_{q_n(\Lambda)q_n(X_n\!)}\!\log p(\theta_n|M_n,\Lambda)\log p(Y_n|\theta_n,X_n) \right)\\ \notag
    &=\prod_{j=1}^{M_n}\text{exp}\left(\E_{q_n(\Lambda)q_n(X_n)}\log p(\theta_{n,j}|\Lambda) \ell(Y_{n,j}|X_{n,\theta_{n,j}})\right)\\ \label{eq:vari theta independent}
    &\propto\prod_{j=1}^{M_n}q_n(\theta_{n,j}),
\end{align}
with each $q_n(\theta_{n,j})$ being 
\begin{align} \notag
    q_n(\theta_{n,j})\propto & \text{exp}\left(\E_{q_n(\Lambda)q_n(X_n)}\log p(\theta_{n,j}|\Lambda) \ell(Y_{n,j}|X_{n,\theta_{n,j}})\right)\\ \label{eq:update theta}
    \propto&\frac{\overbar{\Lambda}_0}{V}\delta[\theta_{n,j}=0]+\sum_{k=1}^K\overbar{\Lambda}_k l_k\delta[\theta_{n,j}=k],\\ \notag
    \overbar{\Lambda}_k=&\exp(\E_{q_n(\Lambda)}\log\Lambda_k)=\exp(\psi(\eta^{k}_{n|n}))\rho^{k}_{n|n}, \\  \notag
    l_k=\mathcal{N}(&Y_{n,j};H\mu_{n|n}^k,R_k)\text{exp}(-0.5\text{Tr}(R_k^{-1}H\Sigma_{n|n}^k H^\top)),
\end{align}
where $\overbar{\Lambda}_k$ is for $k=0,1,...,K$, $l_k$ is for $k=1,...,K$, and $\psi(\cdot)$ is the digamma function. The $\mu_{n|n}^k,\Sigma_{n|n}^k$ are given in \eqref{eq:KF update}, and \eqref{eq:update theta} results from substituting \eqref{eq:single assoc prior}, (\ref{measurement model}), \eqref{eq:KF update} and \eqref{eq:independent rate update}. Detailed derivations of \eqref{eq:update theta} are provided in Appendix \ref{apx:association update deriv}.
It can now be seen that the variational distributions for each association variable are independent, and each of them is a categorical distribution specified (with a proportional constant) in \eqref{eq:update theta}. Thus, the updates for the association $\theta_n$ can be independently carried out for each $\theta_{n,j}$, Similar those for $X_n,\Lambda$.
\vspace{-0.8em}
\subsection{Initialisation}
As a deterministic algorithm, the local optimum (i.e. $q_n^*(X_n,\theta_n,\Lambda)$) found by the CAVI is sensitive to the initial variational distribution. A good initialisation can prevent the algorithm from being trapped in a bad local minimum and may also lead to faster convergence. Therefore, a good initial variational distribution is important for our tracker to perform an accurate and efficient multi-object tracking task. 

A simple choice of initialisation is to employ the predictive distribution of target state $\hat{p}_{n|n-1}(X_n)$ to initialise $q_n(X_n)$. An alternative option utilises the prior distribution $p(\theta_n|M_n,\Lambda)$ in \eqref{eq:assoc prior} to initialise $q_n(\theta_n)$. However, in our experiments, these straightforward initialisations often caused the CAVI to converge to undesirable local minima, resulting in track loss. 

Therefore, this section presents an enhanced initialisation strategy for $q_n(\theta_n)$. Specifically, the setup for the initial association distribution $q_n^{(0)}(\theta_n)$ is detailed. After this initialisation, CAVI is performed by iteratively updating $q(\Lambda)$ and $q_n(X_n)$, which require the initial association distribution $q_n^{(0)}(\theta_n)$, introduced below, for the calculation of the update form.

The initial association distribution $q_n^{(0)}(\theta_n)$ is defined as follows. First, independent initial variational distributions are assumed for each $\theta_{n,j}$ so as to be consistent with the updated form in \eqref{eq:vari theta independent}, that is,
\begin{align}  \label{eq:init independent theta}
q_n^{(0)}(\theta_n)=&\prod_{j=1}^{M_n}q_n^{(0)}(\theta_{n,j})
\end{align}
Recall that the objective of the CAVI in our setup is to minimise the KL divergence between $q_n(X_n)q_n(\theta_n)q(\Lambda)$ and the target distribution $\hat{p}_n(X_n,\theta_n,\Lambda|Y_n)$. Therefore, the best (but intractable) initial distribution for $q_n(\theta_{n,j})$ may be the marginal target distribution $\hat{p}_n(\theta_{n,j}|Y_n)$.
A suboptimal choice would be $\hat{p}_n(\theta_{n,j}|Y_{n,j})$, which, under the same approximate filtering law $\hat{p}_n$, only incorporates the information of the corresponding measurement $Y_{n,j}$ rather than all measurements $Y_n$. However, this suboptimal distribution is still intractable except for the special case that the Gamma prior $\hat{p}_n(\Lambda_k)$ shares the same scale parameter $\rho_{n|n-1}^{k*}$ in \eqref{eq:rate predictive prior computation} for all $k=1,2,...,K$. As a result, we further approximate the intractable $\hat{p}_n(\theta_{n,j}|Y_{n,j})$ with $\hat{p}_n(\theta_{n,j}|Y_{n,j},\Lambda=\hat{\Lambda})$, which can be regarded as an empirical Bayes approximation. 
We expect this intialisation to become quite accurate as $n$ increases and our uncertainty about $\Lambda$ decreases, but in any case the variational method will still be valid as the initialisation is essentially arbitrary.  

Therefore, we suggest using $\hat{p}_n(\theta_{n,j}|Y_{n,j},\Lambda=\hat{\Lambda})$ --- an empirical Bayes approximation to the suboptimal initial distribution $\hat{p}_n(\theta_{n,j}|Y_{n,j})$ --- as the individual initial distribution $q_n^{(0)}(\theta_{n,j})$ in \eqref{eq:init independent theta}. This can be expressed as follows: 
\begin{align} \notag
    q_n^{(0)}(\theta_{n,j})=&\hat{p}_n(\theta_{n,j}|Y_{n,j},\Lambda=\hat{\Lambda})\propto \hat{p}_n(\theta_{n,j},Y_{n,j}|\Lambda=\hat{\Lambda})
    \\ \notag
    =&p(\theta_{n,j}|\Lambda=\hat{\Lambda})\int \ell({Y}_{n,j}|X_{n,\theta_{n,j}})\hat{p}_{n|n-1}(X_n)dX_n\\ \label{eq:init individul theta}
    =&\frac{\hat{\Lambda}_0}{V}\delta[\theta_{n,j}=0]+\sum_{k=1}^K\hat{\Lambda}_k l_k^{0}\delta[\theta_{n,j}=k],\\   \notag
    l_k^{0}=&\mathcal{N}(Y_{n,j};H\mu^{k*}_{n|n-1},H\Sigma^{k*}_{n|n-1}H^\top+R_k),\\   \notag
    \hat{\Lambda}=&\E_{q_{n-1}^*(\Lambda)}\Lambda=[\hat{\Lambda}_0,\hat{\Lambda}_1,...,\hat{\Lambda}_K],\\   \notag
    \hat{\Lambda}_k=&\E_{q_{n-1}^*(\Lambda)}\Lambda_k=\eta^{k*}_{n-1|n-1}\rho^{k*}_{n-1|n-1},
\end{align}
where $l_k^{0}$ is for $k=1,2,...,K$, and $\hat{\Lambda}_k$ ($k=0,1,...,K$) is the estimate of $\Lambda_k$ at the time step $n-1$. $\hat{p}_{n|n-1}(X_n)$ is the predictive prior and $\mu^{k*}_{n|n-1},\Sigma^{k*}_{n|n-1}$ are given in \eqref{eq:state predictive prior computation}. Each $q_n^{(0)}(\theta_{n,j})$ is set to be $\hat{p}_n(\theta_{n,j}|Y_{n,j},\Lambda=\hat{\Lambda})$, where $\hat{p}_n$ is the approximate filtering probability law defined in Section \ref{sec:approximate filtring objective} and \eqref{eq: prob law}. The evaluation of $\hat{p}_n(\theta_{n,j}|Y_{n,j},\Lambda=\hat{\Lambda})$ in \eqref{eq:init individul theta} is a direct result of \eqref{eq: prob law}, and its detailed derivation is shown in Appendix \ref{apx:initialisation law}. This initial variational distribution $q_n^{(0)}(\theta_{n,j})$ has a similar form as the updated $q_n(\theta_n)$ in \eqref{eq:vari theta independent} and \eqref{eq:update theta}, and can also be independently evaluated for each $\theta_{n,j}$.

\subsection{ELBO Computation}
It is useful to compute the ELBO in \eqref{eq: with para ELBO} for each recursion of updates. The reasons are: 1) It provides a convenient way to monitor the convergence. 2) It is practically useful to check the implementation of the algorithm. 3) It reflects the quality of the found variational distribution (this property will be further illustrated and applied in Section \ref{sec: basic vatiational localisation} and \ref{sec: missed objects relocation}). Recall that the ELBO is guaranteed to increase at each recursion of updates, and the objective of our tracker is to maximise the ELBO. By monitoring the increment of ELBO, we can assume that the algorithm converges once the increment of ELBO at the latest recursion falls below a certain threshold. Moreover, if the increment of ELBO resulting from any update is negative, the implementation of the algorithm must be incorrect.

The straightforward computation of the ELBO in \eqref{eq: with para ELBO} involves inversions of non-diagonal matrices $\Sigma^{k}_{n|n}$ and $\Sigma^{k*}_{n|n-1}$, which can be computationally prohibitive for a large number of objects. Here we adopt a relatively efficient way to evaluate the ELBO based on a more easy-to-compute marginal likelihood. Such a technique was first used to compute the ELBO in Section 5.3.7 in \cite{beal2003}. Specifically, if the $q(X_n)$ is the latest updated variational distribution (i.e., $q(\theta_n)$ has not been updated with the current $q(X_n)$), the ELBO in \eqref{eq: with para ELBO} can be exactly expressed as follows,
\begin{align} \notag
    \mathcal{F}(q_n&)=\sum_{j=1}^{M_n}\sum_{k=0}^K q(\theta_{n,j}=k)\left(\psi(\eta_{n|n}^k)+\log\frac{\rho_{n|n}^k}{q(\theta_{n,j}=k)} \right)\\ \notag
    &-\frac{1}{2}\sum_{j=1}^{M_n}\sum_{k=1}^Kq_n(\theta_{n,j}=k)\left(Y_{n,j}^\top R_{k}^{-1}Y_{n,j}+\log\det R_k\right)\\ \notag
    &+\frac{1}{2}\sum_{k=1}^K\left({\overbar{Y}_n^k}^\top{\overbar{R}_n^k}^{-1}\overbar{Y}_n^k-{T_n^k}^\top{S_n^k}^{-1}T_n^k+\log \frac{\det \overbar{R}_n^k}{\det S_n^k} \right)\\ \notag
    &+\left[\frac{D}{2}\log 2\pi+\log\frac{1}{V}\right]\sum_{j=1}^{M_n}q_n(\theta_{n,j}=0)\\ \notag
    &-\sum_{k=0}^K\eta_{n|n}^k\rho_{n|n}^k- \text{KL}(q_n(\Lambda)||\hat{p}_{n|n-1}(\Lambda))\\
    \label{eq: efficient ELBO}
    &-\frac{1}{2}DM_n\log2\pi-\log (M_n!),
\end{align} 
\begin{align} \notag
    \text{KL}(q_n&(\Lambda)||\hat{p}_{n|n-1}(\Lambda))=\sum_{k=0}^K \Bigl[-\eta_{n|n-1}^{k*}\log \rho_{n|n}^{k}-\log \Gamma(\eta_{n|n}^{k})\\ \notag
    & +(\eta_{n|n}^{k}-\eta_{n|n-1}^{k*}) \psi(\eta_{n|n}^{k})+\eta_{n|n}^{k}(\rho_{n|n}^{k}/\rho_{n|n-1}^{k*}-1) \Bigr]\\  \label{eq: kl for ELBO}
    &+\sum_{k=0}^K\left(\log\Gamma(\eta_{n|n-1}^{k*})+\eta_{n|n-1}^{k*}\log\rho_{n|n-1}^{k*}\right),
\end{align}
where $\overbar{Y}_n^k,\overbar{R}_n^k$ are given in \eqref{eq:pseudomeas R}; \eqref{eq:pseudomeas Y}, and
$T_n^k,S_n^k$ are given in the Kalman filter update \eqref{eq:KF update}, and $\eta_{n|n}^{k},\rho_{n|n}^{k},\eta_{n|n-1}^{k*},\rho_{n|n-1}^{k*}$ are given in \eqref{eq: rate update}\eqref{eq:rate predictive prior computation}. Note that all these terms have already been computed for updating the latest $q_n(X_n)$ and $q_n(\lambda)$.
$\Gamma(\cdot)$ is the gamma function. $D$ is the dimension of the vector of a single measurement $Y_{n,j}$, and $V,R_k$ are defined in (\ref{measurement model}). The laborious derivation for \eqref{eq: efficient ELBO} is presented in Appendix \ref{apx: elbo derivation}. We emphasise again that the evaluation in \eqref{eq: efficient ELBO} only equals the exact ELBO in \eqref{eq: with para ELBO} when $q_n(X_n)$ is the latest updated variational distribution.

The last lines in \eqref{eq: efficient ELBO} and \eqref{eq: kl for ELBO} are constants through each iteration of CAVI, and hence can be omitted when monitoring the increments of $\mathcal{F}(q_n)$. In a common tracking scenario where the dynamic transition in \eqref{eq: dynamic transition} and the positional measurement function are both independent for different coordinates,  $R_k,\overbar{R}_n^k,S_n^k$ are all diagonal matrices.  Subsequently, all matrix inversions in \eqref{eq: efficient ELBO} can be computed easily.

\subsection{Algorithm}
Finally, the approximate filtering algorithm at time step $n$ of our VB-AbNHPP tracker is summarised in Algorithm \ref{Algo:tracker}. In brief, this algorithm first 1) evaluates the predictive distribution for $X_n,\Lambda$; then 2) sets the initialisation variational distribution for $\theta_n$; and 3) carries out the CAVI until convergence. 
For each CAVI iteration in Algorithm \ref{Algo:tracker}, the primary computational demand arises from updating $q_n(X_n)$ and $q_n(\theta_n)$. The cost of updating $q_n(X_n)$ is essentially computing pseudo measurements \eqref{eq:pseudomeas R}-\eqref{eq:pseudomeas Y} and applying the Kalman filter \eqref{eq:KF update} for each of the $K$ objects, all of which can be parallelised across $K$. The cost for updating $q_n(\theta_n)$ is primarily calculating $(K+1)M$ likelihoods, specifically the $\overbar{\Lambda}_kl_k$ in \eqref{eq:update theta}, and these $(K+1)M$ calculations can also be parallelised.
Moreover, the ELBO computation can also be further accelerated by exploring the parallel computing of the summands in \eqref{eq: efficient ELBO}.

The ELBO evaluation step \eqref{eq: efficient ELBO} in Algorithm \ref{Algo:tracker} can be neglected if the computational power is limited. In this case, we can either monitor the change in statistics of the variational distribution (e.g. the mean of $q(X_n)$) to check for convergence, or we can directly set the algorithm to perform a predefined number of iterations.

\begin{algorithm} 
\SetAlgoLined
\algsetup{linenosize=\tiny}
  \scriptsize
\textbf{Require}: $q^*_{n-1}(X_{n-1}), q^*_{n-1}(\Lambda), Y_n,M_n$, maximum iteration limit $I$, tolerance threshold $\epsilon>0$.\\
\textbf{Output}: $q^*_{n}(X_{n}),q^*_{n}(\theta_{n}),q^*_{n}(\Lambda)$ .\\
\textit{Prediction for $X_n$}: Evaluate $\mu^{k*}_{n|n-1},\Sigma^{k*}_{n|n-1}$ for $\hat{p}_{n|n-1}(X_n)$ via \eqref{eq:state predictive prior computation}.\\
\textit{Prediction for $\Lambda$}: Evaluate $\eta^{k*}_{n|n-1},\rho^{k*}_{n|n-1}$ for $\hat{p}_{n|n-1}(\Lambda)$ via \eqref{eq:rate predictive prior computation}.\\
\For {$j=1,2,...,M_n$}
{ Evaluate $q_n^{(0)}(\theta_{n,j})$ via \eqref{eq:init individul theta}, and initialise $q_n(\theta_{n,j})\gets q_n^{(0)}(\theta_{n,j})$.
}

   \For {$i=1,2,...,I$}  
   {
   \For {$k=1,2,...,K$}
   { 
   Evaluate $\eta^{k}_{n|n},\rho^{k}_{n|n}$ according to \eqref{eq:independent rate update} for updating $q_n(\Lambda)$.\\
   }
   \For {$k=1,2,...,K$}
   { 
   Evaluate $\overbar{R}_n^k,\overbar{Y}_n^k$ according to \eqref{eq:pseudomeas R}\eqref{eq:pseudomeas Y}.\\
   Compute $\mu_{n|n}^k,\Sigma_{n|n}^k,T_n^k,S_n^k$ via \eqref{eq:KF update} for updating $q_n(X_{n,k})$.
   }
   Evaluate the ELBO $\mathcal{F}_{n}^{(i)}$ according to \eqref{eq: efficient ELBO}.\\
    \If{$\mathcal{F}_{n}^{(i)}-\mathcal{F}_{n}^{(i-1)}<\epsilon \land i\geq 2$}
      {
      \textbf{break} 
      }
    \For {$j=1,2,...,M_n$}
    { Update $q_n(\theta_{n,j})$ according to \eqref{eq:update theta}.
    }
   }
   Set $q^*_{n}(X_{n})=\prod_{k=1}^Kq_n(X_{n,k})$, $q^*_{n}(\Lambda)=\prod_{k=1}^K q_n(\Lambda_{k})$ and $q^*_{n}(\theta_{n})=\prod_{j=1}^{M_n}q_n(\theta_{n,j}).$
 \caption{VB-AbNHPP tracker with rate estimation at time step $n$}
 \label{Algo:tracker}
\end{algorithm}

\vspace{-1.2em}
\section{Variational object localisation under heavy clutter with non-informative prior} \label{sec: basic vatiational localisation}
This section presents a novel technique for efficiently locating a object amidst a high level of clutter using the variational Bayes framework. Different from tracking scenarios where the previous time step's tracking result can provide an informative prior of the object's position, the localisation strategy discussed here does not require such a strong informative prior. As a result, this technique can be useful for relocating objects once they lose track, or for initialising a tracking algorithm where the object positions are hardly known. Moreover, it has the potential to be developed into a strategy for estimating the number of objects. This section places an emphasis on clarifying the technique's rationale and thus only considers the localisation of a single object. We will extend the relocation technique to handle multiple missed objects in Section \ref{sec: missed objects relocation}, and integrate it into the complete VB-AbNHPP tracking algorithm.

We assume that the object to be localised follows the NHPP model, and that both the Poisson rates $\Lambda$ and the measurement covariance are known to us (e.g., have been estimated in advance using the proposed method). 
Our strategy is designed for challenging scenarios where the clutter number in the survey area can be hundreds of times greater than the object's measurement number. We aim to efficiently locate the object only using measurements received at a single time step, and the object can be anywhere in the survey area. Currently, our strategy can handle object Poisson rates as low as 3. In more challenging scenarios where the object Poisson rate is lower, the localisation may be achieved by using measurements from multiple time steps, and this case will be discussed in future.

\begin{figure}
    \centering
    \includegraphics[width=7cm]{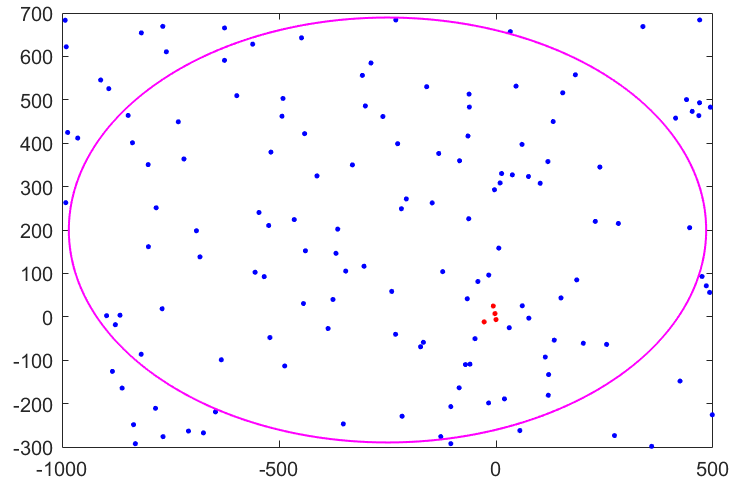}
    \caption{Measurements received at time step $n$; ground-truth object position is $(0,0)$; purple ellipse is the $95\%$ uncertainty ellipse of the object position Gaussian prior; four red dots are object measurements, and the blue dots are clutters.}
    \label{fig:allobs}
      \vspace{-1em}
\end{figure}
\vspace{-0.8em}
\subsection{Problem Setting}
To clarify the main concept of our procedure, let us consider a typical task of locating a single object under heavy clutter at time step $n$. This procedure can be performed at the initialisation ($n=0$) or at any time step ($n\neq 0$) when the track is lost. In this example, the object number $K=1$, the object state $X_n=X_{n,1}$, and the received measurements $Y_n$ are shown in Fig. \ref{fig:allobs}.

We assume that only position measurements of the object state $X_{n,1}$ can be directly observed, with the measurement matrix $H$ and measurement noise covariance matrix $R_{n,1}=100\text{I}_2$, where $\text{I}_2$ is a $2D$ identity matrix. The object measurement number is assumed to be Poisson distributed with a Poisson rate of $\Lambda_1=4$, and these object-oriented measurements are buried in the uniformly distributed clutter with a heavy clutter density of $\Lambda_0/V=10^{-4}$. The exact object position, i.e. $HX_{n,1}$, is $(0,0)$. However, we have no other information regarding this position except for a rather flat Gaussian prior $p(X_{n,1})$ and the $95\%$ uncertainty ellipse of position is shown in Fig. \ref{fig:allobs}. We denote this ellipse as the survey area, and we will only search for the object within this area.
Note that, strictly speaking, $p(X_{n,1})$ should be computed by using the given initial prior $p(X_0)$ and transition $p(X_n|X_{n-1})$. In this section, we directly assign a highly uncertain prior to $p(X_{n,1})$, which is specifically designed for missed objects.

In such scenarios, locating a object can be be difficult due to the dense clutter,
which may lead to a multimodal posterior and further confuse the algorithm from finding the true object location. Specifically, when clutter measurements are occasionally densely displayed in some small regions, these regions become deceptive candidates for the object’s true location, leading to several competitive modes in the posterior of the object's position. Taking Fig. \ref{fig:allobs} as an example, it may be difficult to 
determine whether the object is located at $(0,0)$ or $(10,300)$, as both of these two locations have many measurements around them.

\vspace{-0.7em}
\subsection{Variational Localisation Strategy} \label{sec: single object vari loca strate}
Now we formulate the object localisation task within the variational inference framework, where we aim to approximate the exact posterior $p(X_{n,1},\theta_n|Y_n,\Lambda)$ with a variational distribution $q(X_{n,1})q(\theta_n)$. Here, the CAVI features similar updates as in \eqref{eq:KF update} and \eqref{eq:update theta}: the update for $q(X_{n,1})$ is the same as \eqref{eq:KF update} except that $\mu^{k*}_{n|n-1}$ and $\Sigma^{k*}_{n|n-1}$ are replaced by the mean and covariance of the prior $p(X_{n,1})$, and the update for $q(\theta_n)$ is the same as \eqref{eq:update theta} except for replacing $\overbar{\Lambda}$ with the exact known $\Lambda$. The standard CAVI with a single initialisation, when applied to the considered task, is prone to getting trapped in local optima and the converged variational distribution only accommodates a single mode of the posterior distribution. To overcome it, the main concept of our localisation technique is to identify multiple competitive modes in the exact posterior $p(X_{n,1},\theta_n|Y_n,\Lambda)$ by implementing multiple runs of the CAVI in parallel with different initialisations. We then select the most likely mode by evaluating the ELBO calculated for each mode.

Specifically, each run of CAVI starts from initialising the association distribution $q(\theta_n)$ with the $q^{(0)}(\theta_n)$,
\vspace{-0.5em}
\begin{align}   \label{eq: Loca demo initial}
\begin{aligned}
q^{(0)}(\theta_n)=&\prod_{j=1}^{M_n}q^{(0)}(\theta_{n,j})\\
    q^{(0)}(\theta_{n,j})
    \propto&\frac{\Lambda_0}{V}\delta[\theta_{n,j}=0]+\Lambda_1 l_1^{0}\delta[\theta_{n,j}=1],\\  
    l_1^{0}=&\mathcal{N}(Y_{n,j};m_s,C+R_{n,1}), 
\end{aligned}
\end{align}
where $N$ is the total number of initialisations, and $m_s$ ($s=1,2,....,N$) and $C$ are manually selected for each initialisation and have the same dimensions as $Y_{n,j}$ and $R_{n,1}$, respectively. 

Such an initial distribution mimics the form of the initialisation in \eqref{eq:init individul theta}, except that $H\mu^{k*}_{n|n-1}$ and $H\Sigma^{k*}_{n|n-1}H^\top$ (i.e., the mean and the covariance of the predictive prior $\hat{p}_{n|n-1}(X_{n,k})$) are replaced by $m_s$ and $C$, respectively.
In a nutshell, the localisation procedure for the considered task is three-step:
\begin{enumerate}[Step 1:]
\item Choose a $C$, and assign a series of values for $m_s$ ($s=1,2,...,N$) such that the union of all the $95\%$ error ellipses of $\mathcal{N}(m_s,C)$ can cover the $95\%$ error ellipse of the prior $p(X_{n,1})$. See Fig. \ref{fig:all init} for example.
\item Run CAVI with the initialisation in \eqref{eq: Loca demo initial} for each pair of $m_s$ and $C$, then record each $q^*(X_{n,1})$ and its final ELBO. The CAVIs can be run in parallel.
\item Find $q^*(X_{n,1})$ with the highest ELBO, which is the most likely one to capture the true object's position. 
\end{enumerate}
In the following, we present two rationales that this localisation method is inherently based upon.
\begin{remark} \label{remark: initial}
For each $s$, the CAVI is expected to find the most probable object location (the location with the greatest density of measurements) within the $95\%$ confidence ellipse of $\mathcal{N}(m_s,C)$, with the initialisation in \eqref{eq: Loca demo initial} and a properly chosen $C$. See discussions below and in Section \ref{sec: object loca demo}. 
  \end{remark}
 \begin{remark} \label{remark: ELBO}
 The ELBO, which equals the negative KL divergence up to an additive
constant, reflects the quality of the found variational distribution $q^*(X_{n,1})q^*(\theta_n)$. The higher the ELBO, the better the variational distribution we have found. 
\end{remark}

Under these two remarks of the proposed localisation strategy, each CAVI for each initialisation will explore a specific local area (e.g., the green circle in Fig. \ref{fig:all init}) to find the `best' object position in the local area. As the union of these local areas encompasses the survey area, it is anticipated that the object can be successfully localised by a converged variational distribution $q^*(X_{n,1})q^*(\theta_n)$ with the highest ELBO.

Remark \ref{remark: ELBO} is well studied (e.g. \cite{blei2017variational}) while Remark \ref{remark: initial} is only verified empirically for the considered localisation case in heavy clutter. Particularly, we observe that, typically, Remark \ref{remark: initial} holds true only for a small enough $C$. If $C$ is too large, the CAVI either converges to $q^*(X_{n,1})$ that covers the dense measurements nearest to $m_s$ (i.e., the centre of the local area), or traps in the local optimum with $q^*(X_{n,1})=p(X_{n,1})$. Note that however, we still would like $C$ to be as large as possible in practical implementation, since this leads to fewer initialisations and hence less computational power required to explore the entire survey area. More details about the empirical properties we observed about the Remark \ref{remark: initial}, and an informal justification of their rationales is provided in Appendix \ref{apx:remark for initialisation}.

\begin{figure}
    \centering
    \includegraphics[width=7cm]{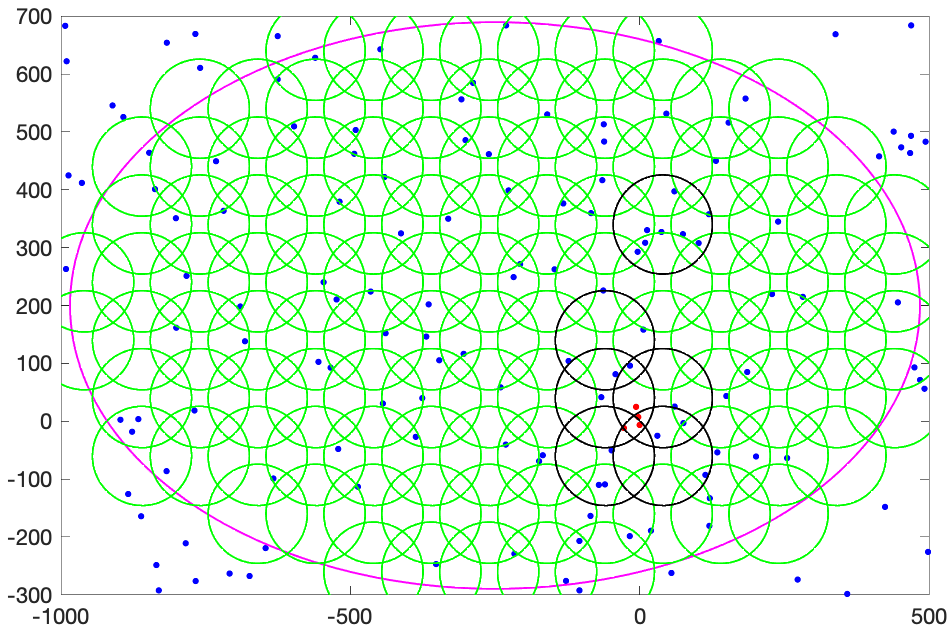}
    \caption{
    An example of all $95\%$ error ellipses (green/black circles) of $\mathcal{N}(m_s,C)$ used for the initialisation in \eqref{eq: Loca demo initial} (total number $N=117$). The effects of six of these initialisations (black circles) are further demonstrated in Fig. \ref{fig:loca demo}.}
    \label{fig:all init}
      \vspace{-1em}
\end{figure}

\vspace{-1em}
\subsection{Demonstration} \label{sec: object loca demo}
We now demonstrate the effectiveness of our localisation strategy with an example task. The initialisation setting is shown in Fig. \ref{fig:all init} where each green circle denotes one initialisation, and all initialisations have the same constant $C=35^2\text{I}_2$. To show how the algorithm works, we further analyse six initialisations among them, highlighted in black circles in Fig. \ref{fig:all init}, and present the results of the CAVI with these initialisations in Fig. \ref{fig:loca demo}. In each subfigure in Fig. \ref{fig:loca demo}, the green circle denotes the initialisation $\mathcal{N}(m_s,C)$; black/grey circles are the object positional variational distribution $q(HX_{n,1})$ evaluated at all iterations, and the color of the circles gradually darkens from grey to black along with the sequence of iterations. Specifically, the lightest grey circle denotes the first iteration's variational distribution and the black circle denotes the converged variational distribution. Note that in Fig. \ref{fig:loca demo}, we only present the position information $q(HX_{n,1})$ by extracting the mean and covariance of the position from $q(X_{n,1})$. All circles/ellipses in Fig. \ref{fig:loca demo} are plotted with a 95\% confidence interval. For each subfigure, we compute the ELBO (with an additive constant) by using its converged variational distribution (shown as the black circle). Likewise, we independently carry out the CAVI for all initialisations in Fig. \ref{fig:all init} until convergence, and compute an ELBO for each initialisation. It turns out that the highest ELBO is $5.671$, and the corresponding iterative update results and the initialisation are shown in Fig. \ref{fig:loca demo}(e). We can see from Fig. \ref{fig:loca demo}(e) that the converged variational distribution successfully captures the exact object position, which demonstrates that our localisation strategy is effective for this task.

This demonstration example can also facilitate verifying two remarks and the efficacy of the initialisation in \eqref{eq: Loca demo initial}. Let us take a closer look at Fig. \ref{fig:loca demo}. 
It can be noticed that, whether or not it converges to the true position, the converged variational distribution $q^*(HX_{n,1})$ is able to capture the most probable location with the greatest density of measurements in the area covered by the green circle, which verifies Remark \ref{remark: initial}.
In particular, subfigures (a-b) and (e) all locate the object around its exact position $(0,0)$, and correspondingly achieve the highest three ELBO values, in which case it provides a triple guarantee for the success of our localisation strategy. A noteworthy fact is that the aforementioned deceptive potential object location $(10,300)$ is found by the initialisation in the subfigure (c) with a relatively high ELBO of $4.31$, verifying that the found variational distribution is a competitive mode in the exact posterior. Still, this ELBO is not as high as the ELBO achieved by those variational distributions that capture the exact object location with an ELBO of $5.67$. This demonstrates that the ELBO is a reliable metric that can tell the slight difference between multiple competitive modes in the exact posterior, which verifies Remark \ref{remark: ELBO}. 

\begin{figure}
    \centering
    \includegraphics[width=8.9cm]{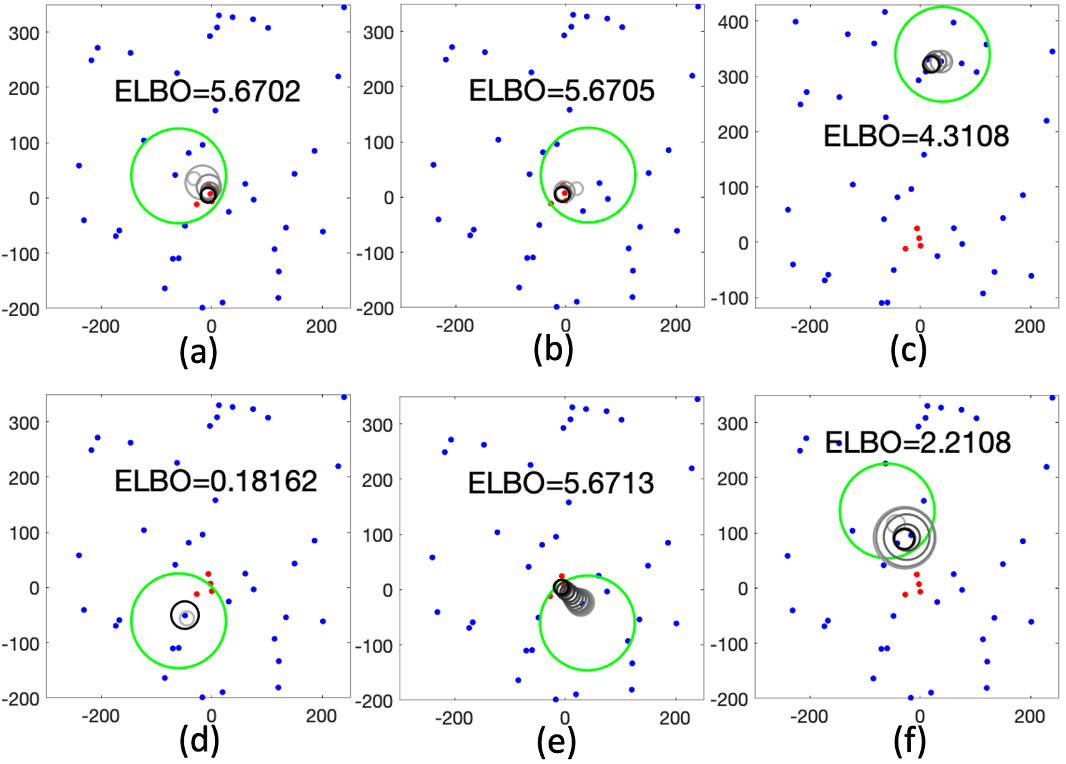}
    \caption{The results of CAVI with six initialisations in Fig. \ref{fig:all init}. The green circle depicts the $\mathcal{N}(m_s,C)$ in \eqref{eq: Loca demo initial}. The grey/black circles represent the iterative updated $q(HX_1)$, whose color gradually darkens with the sequence of iterative updates in CAVI. All circles/ellipses are with 95\% confidence interval.}
    \label{fig:loca demo}
      \vspace{-1em}
\end{figure}

\begin{figure}
    \centering
    \includegraphics[width=8.9cm]{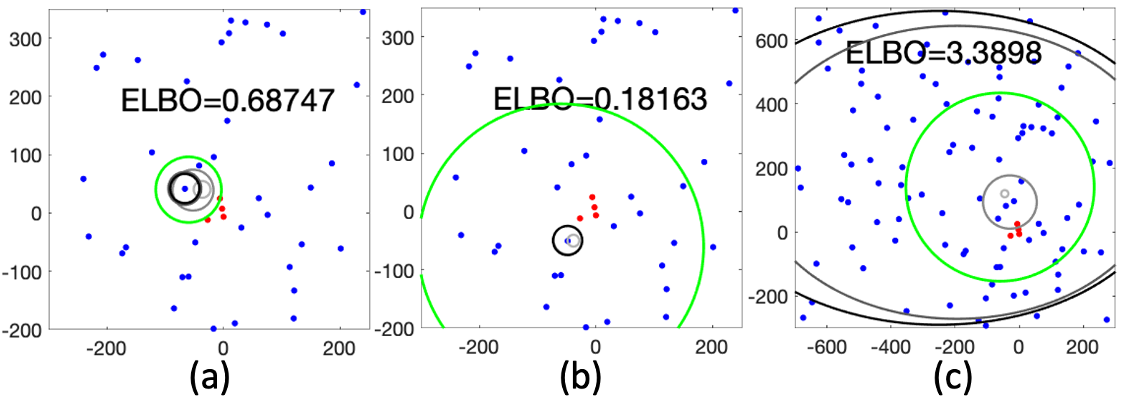}
    \caption{The results of CAVI with three initialisations that employ different $m_s$ and different $C$.}
    \label{fig:locvar}
    \vspace{-1.2em}
\end{figure}

Finally, it should be noted that the selection of $C$ in \eqref{eq: Loca demo initial} is important to our localisation strategy and may need to be tuned for each realisation of the parameter set $\Lambda$.
Recall that for each $s=1,2,...,N$, if $C$ is too large, the final $q^*(X_{n,1})$ would either concentrate around the measurement nearest to $m_s$, or be trapped in the local optimum with $q^*(X_{n,1})=p(X_{n,1})$.
For example, the subfigures (b) and (c) in Fig. \ref{fig:locvar} employ the same $m_s$ as in Fig. \ref{fig:loca demo}(d) and Fig. \ref{fig:loca demo}(f), respectively, while both having a larger $C$. We can see that both cases in Fig. \ref{fig:locvar} fail to find the most probable object location in the green ellipse. Specifically, even though the green ellipse in Fig. \ref{fig:locvar}(b) covers the exact object location, the final $q^*(X_{n,1})$ still finds the local optimum that is closer to $m_s$. As to the $q^*(X_{n,1})$ in Fig. \ref{fig:locvar}(c), it fails to capture the local optimum in Fig. \ref{fig:loca demo}(f), and instead converges to the prior $p(X_{n,1})$.
In contrast, choosing a smaller $C$ leads to a higher computational burden as it requires more initialisations to cover the high confidence region of the prior, and thus more runs of CAVI to explore the entire survey area. Meanwhile, it may be inefficient as each CAVI will only be able to explore a smaller local region. For example, the initialisation in Fig. \ref{fig:locvar}(a) has the same $m_s$ as in Fig. \ref{fig:loca demo}(a) but with a lower $C$. The final $q^*(X_{n,1})$ captures the most probable object location in the small green ellipse. However, it misses the exact location which would have been covered by a larger $C$ and then captured by the CAVI as in Fig. \ref{fig:loca demo}(a).

\vspace{-1em}
\subsection{Acceleration of the localisation strategy} \label{sec: acce algo}

The computational complexity of the presented localisation strategy is primarily determined by the number of runs of CAVI, which is equivalent to the number of initialisations $N$. In addition to parallelising CAVI runs, enhancing the efficiency of the localisation procedure can be achieved by judiciously reducing $N$. This reduction can be significant by leveraging our knowledge of the object rate $\Lambda_1$ ($4$ in the current problem setting). Recall that the concept of localisation strategy may be understood by first dividing the entire survey area into $N$ local regions, then employing CAVI to explore each local region. However, some of these local regions are not worth exploring if they include only a few measurements, which result in a total count too low to adhere to our Poisson distribution assumption with rate $\Lambda_1$. Hence we can ignore such initialisations to accelerate the algorithm. 

For example, the regions that include $0$ measurement are obviously not worth exploring; and even if we did explore them, without any measurement, there wouldn't be any reasonable localisation result. In particular, we can perform the CAVI only for the local regions that include at least $M^{init}$ measurements. The choice of $M^{init}$ will be discussed in Appendix \ref{apx:choice of Mreloc and Minit} for our full tracking algorithm. Subsequently, the number of eligible initialisations (or green circles) in Fig. \ref{fig:all init} is $N=117$ if $M^{init}=0$. This number reduces to $104$,$70$ and $36$ if $M^{init}$ is set to $1$,$2$, and $3$ respectively. For this specific task, we can still successfully localise the object even if we set $M^{init}=6$, which only requires $6$ initialisations in total.

\section{VB-AbNHPP Tracker with Missed Objects Relocation} \label{sec: missed objects relocation}

We now extend the localisation strategy for a single object in Section \ref{sec: single object vari loca strate} to enable the relocation of multiple lost track objects. This strategy will be employed in the developed VB-AbNHPP tracker, and the procedure of the full VB-AbNHPP tracker with relocation will be given in subsection \ref{sec:full tracking with relocation}. Recall that the standard VB-AbNHPP tracker aims to track $K$ objects with labels $1,2,...,K$. Now define $\mathcal{L}_n$ as the set of labels for objects whose tracks have been lost after a standard tracking procedure at time step $n$. The objective of our relocation strategy at time step $n$ is to relocate these objects before we implement the next tracking procedure at $n+1$. 

Specifically, to relocate multiple objects in $\mathcal{L}_n$, we first assign a non-informative Gaussian prior $\tilde{p}_n(X_{n,h})$, $h\in\mathcal{L}_n$ to each of them. As in Section \ref{sec: basic vatiational localisation}, we assume that the Poisson rate $\Lambda$ is known or has been estimated. Then the target distribution for the relocation task at time step $n$ is 
\vspace{-0.3em}
\begin{align} \notag
    \tilde{p}_n(X_n,\theta_n|Y_n)\propto  \prod_{h\in\mathcal{L}_n}\tilde{p}_n(&X_{n,h})\hat{p}_{n|n-1}(X_{n,k\not\in\mathcal{L}_n})\\[-1em] \label{eq: multi relocate target distribution}
    \times& p(Y_n|\theta_n,X_n)p(\theta_n|M_n,\Lambda), \vspace{-0.3em}
\end{align}
where $X_{n,k\not\in\mathcal{L}_n}$ denotes the state of all objects that are not in $\mathcal{L}_n$, and their predictive prior $\hat{p}_{n|n-1}(X_{n,k\not\in\mathcal{L}_n})$ has been defined in \eqref{eq: approx prior factorize and predictive prior} and computed in the standard tracker via \eqref{eq:state predictive prior computation}. This target distribution $\tilde{p}_n(X_n,\theta_n|Y_n)$ is defined similarly to $\hat{p}_n(X_n,\theta_n|Y_n,\Lambda)$, where $\hat{p}_n$ is defined in \eqref{eq: prob law}, except that the predictive priors for missed objects' states are replaced by a non-informative prior $\tilde{p}_n(X_{n,h})$, $h\in\mathcal{L}_n$.

Subsequently, the relocation task can be formulated within the variational Bayes framework as follows: We aim to minimise the KL divergence between the variational distribution $q_n(\theta_n)q_n(X_n)$ and the target distribution $\tilde{p}_n(X_n,\theta_n|Y_n)$ in \eqref{eq: multi relocate target distribution}. Different from the tracking tasks in Section \ref{sec: vbnhpp} where $q_n(X_n)$ are updated independently for each object, here for the efficiency of the relocation task, we are only interested in updating the $q_n(X_{n,h})$ for missed objects $h\in\mathcal{L}_n$. Moreover, we will localise each object in $\mathcal{L}_n$ one at a time. This is to avoid the exponentially increasing number of (multi-dimensional) local regions that are required to cover the entire multi-dimensional survey area if multiple objects are localised at the same time. For example, if for each single object, $100$ local regions are required to overlap the single-dimensional survey area, then it only requires $100\times3$ initialisations of CAVI in total to locate three objects one at a time; however, it needs $100^3$ initialisations (each in a three-dimensional space) to cover the whole survey area in three-dimensional space if three objects are localised together. 

It should be noted that when we localise the object $h$ ($h\in\mathcal{L}_n$), other objects (including other objects in $\mathcal{L}_n$ and all objects that are tracked properly) have fixed variational distributions $q^*_n(X_{n,h-})$, where $X_{n,h-}$ denotes all the object states in $X_n$ except $X_{n,h}$. This converged variational distribution $q^*_n(X_{n})$ is first obtained from the standard tracking procedure at time step $n$, and should be updated timely as $q^{*new}_n(X_{n,h})q^*_n(X_{n,h-})$ once a lost track object $h$ in $\mathcal{L}_n$ has been successfully relocated by a converged variational distribution $q^{*new}_n(X_{n,h})$. Subsequently, for each $h\in\mathcal{L}_n$, the CAVI is implemented to maximise the following ELBO by iteratively updating $q_n(\theta_n)$ and $q_n(X_{n,h})$,
\vspace{-0.1em}
\begin{align}  \notag
    &\mathcal{F}_{n,h}(q_n(\theta_n),q_n(X_{n,h}))=\E_{q_n(\theta_n)q_n(X_{n,h})q^*_n(X_{n,h-})}\\ \notag
    &\!\!\log\!\frac{p(Y_n|\theta_n,\!X_n)p(\theta_n|M_n,\!\Lambda)\hat{p}_{n|n-1}(\!X_{n,k\not\in\mathcal{L}_n}\!)\!\prod_{i\in\mathcal{L}_n}\!\tilde{p}_n(X_{n,i})}{q_n(\theta_n)q_n(X_{n,h})q^*_n(X_{n,h-})}\\ \notag
    &\!\!= -\text{KL}(q_n(\theta_n)||p(\theta_n|M_n,\Lambda))-\text{KL}(q_n(X_{n,h})||\tilde{p}_n(X_{n,h}))\\ \label{eq: reloc ELBO}
    &\ \ +\E_{q_n(\theta_n)q_n(X_{n,h})q^*_n(X_{n,h-})}\log p(Y_n|\theta_n,X_n)+c, \vspace{-0.2em}
\end{align}
where $c$ is a constant that does not depend on $q_n(\theta_n)$ or $q_n(X_{n,h})$. The specific form of \eqref{eq: reloc ELBO} is derived in Appendix \ref{apx: reloc elbo derivation}. This optimisation is still equivalent to minimising the KL$(q_n(\theta_n)q_n(X_n)||\tilde{p}_n(X_n,\theta_n|Y_n))$, only now $q_n(X_n)$ is set as $q_n(X_{n,h})q^*_n(X_{n,h-})$ and $q^*_n(X_{n,h-})$ is fixed when localising the object $h$. 

By fixing other objects' variational distributions, the variational localisation of the object $h$ naturally takes into account the impact of other successfully localised objects. Specifically, the measurements that are covered by the fixed variational distributions $q^*_n(X_{n,h-})$ naturally have a considerable probability to associate with the successfully tracked objects in $X_{n,h-}$. In order to achieve a high ELBO, typically $q_n(X_{h})$ tends to encompass some dense measurements that are away from other objects, instead of relocating the object $h$ around other properly tracked objects to compete for the association probability of the measurements nearby. 
Therefore, multiple objects can be relocated to different potential locations, rather than occupying the same place with limited measurements around them.

The optimisation of the ELBO $\mathcal{F}_{n,h}$ in \eqref{eq: reloc ELBO} requires similar variational updates as in Section \ref{sec: single object vari loca strate}. The update for $q_n(X_{n,h})$ is the same as in \eqref{eq:KF update} except replacing the $\mu^{h*}_{n|n-1}$ and $\Sigma^{h*}_{n|n-1}$ with the mean and covariance in $\tilde{p}_n(X_{n,h})$. Recall that $q^*_n(X_{n,k})=\mathcal{N}(\mu_{n|n}^{k*},\Sigma_{n|n}^{k*})$ and $q_n(X_{n,h})=\mathcal{N}(\mu_{n|n}^{h},\Sigma_{n|n}^{h})$, the update of $q_n(\theta_n)$ is also similar to \eqref{eq:update theta}:
\vspace{-0.4em}
\begin{align} 
\begin{aligned} \label{eq: reloc theta update}
    q_n(\theta_n)&=\prod_{j=1}^{M_n}q_n(\theta_{n,j}),\\[-0.5em]
    q_n(\theta_{n,j})&\propto  \frac{\Lambda_0}{V}\delta[\theta_{n,j}=0]+\sum_{k=1}^K\Lambda_k l_k\delta[\theta_{n,j}=k],\\[-0.3em] 
    l_k&=\begin{cases}  \frac{\mathcal{N}(Y_{n,j};H\mu_{n|n}^{k*},R_k)}{\text{exp}(0.5\text{Tr}(R_k^{-1}H\Sigma_{n|n}^{k*} H^\top))}, &k\neq h\\
    \frac{\mathcal{N}(Y_{n,j};H\mu_{n|n}^{h},R_h)}{\text{exp}(0.5\text{Tr}(R_h^{-1}H\Sigma_{n|n}^{h} H^\top))}, &k= h
    \end{cases}
\end{aligned}
\end{align}

Moreover, to relocate the lost track object $h$ under heavy clutter, it is essential to implement the same variational localisation strategy with multiple initialisations as presented in Section \ref{sec: single object vari loca strate}. That is, we need to first determine an initialisation covariance $C$, then assign a series values of $m_s$ ($s=1,2,...,N$) such that the union of high (e.g. $95\%$) confidence ellipses of $\mathcal{N}(m_s,C)$ can cover the high confidence region of our prior $\tilde{p}_n(X_{n,h})$. Then, multiple CAVIs can be carried out in parallel to search for the most likely object position in a local region $\mathcal{N}(m_s,C)$ for each initialisation:
\vspace{-0.5em}
\begin{align}  \notag 
q_n^{(0)}(\theta_n)&=\prod_{j=1}^{M}q_n^{(0)}(\theta_{n,j})\\[-0.5em] \label{eq: reLoca initial}
    q_n^{(0)}(\theta_{n,j})
    &\propto\frac{\Lambda_0}{V}\delta[\theta_{n,j}=0]+\sum_{k=1}^K\Lambda_k l_k^0\delta[\theta_{n,j}=k],\\[-0.2em]  \notag
    l_k^0=&\begin{cases}  \mathcal{N}(Y_{n,j};H\mu^{k*}_{n|n-1},H\Sigma^{k*}_{n|n-1}H^\top+R_k), &k\neq h\\
    \mathcal{N}(Y_{n,j};m_s,C+R_h), &k= h
    \end{cases} \vspace{-1.5em}
\end{align}
where $s=1,2,....,N$; $\mu^{k*}_{n|n-1}$ and $\Sigma^{k*}_{n|n-1}$ are the mean and covariance of $\hat{p}_n(X_{n,k})$, which have been computed by \eqref{eq:state predictive prior computation} in the standard tracking procedure. Similar to the initialisation \eqref{eq: Loca demo initial} in Section \ref{sec: single object vari loca strate}, \eqref{eq: reLoca initial} mimics the form of \eqref{eq:init individul theta} with a modification to the term $l_h^0$. As discussed in Section \ref{sec: acce algo}, not all of these initialisations are necessary to implement. We will only consider the initiasations in which the $95\%$ confidence ellipses of $\mathcal{N}(m_s,C)$ include at least $M_k^{init}$ measurements, where $M_k^{init}$ is an eligible initialisation threshold. The choice of $M_k^{init}$ will be discussed in Appendix \ref{apx:choice of Mreloc and Minit}.

\newlength{\textfloatsepsave} 
\setlength{\textfloatsepsave}{\textfloatsep} 
\setlength{\textfloatsep}{0.5pt}
\begin{algorithm}  
\SetAlgoLined
\algsetup{linenosize=\tiny}
  \scriptsize
\textbf{Require}: $Y_n,M_n,I,\epsilon,\hat{p}_{n|n-1}(X_n),q^*_{n}(X_{n})$ from the standard tracker, the exact or estimated $\Lambda$, missed target set $\mathcal{L}_n$, relocation initialisation thresholds $M_{h\in\mathcal{L}_n}^{reloc}$, and eligible initialisation thresholds $M_{h\in\mathcal{L}_n}^{init}$.  \\
\textbf{Output}: Successfully tracked/relocated set $\mathcal{K}_n$, refined $q^*_{n}(X_{n})$ and $q^*_{n}(\theta_{n})$.\\
Initialise $\mathcal{K}_n=\{k\in\{1,2,...,K\}: k\notin\mathcal{L}_n\}$.\;
\ForEach { $h\in\mathcal{L}_n$} 
{
Assign an uninformative Gaussian prior $\tilde{p}_n(X_{n,h})$.\;
Determine $m_s(s=1,2,...,N)$ and $C$ for initialisations \eqref{eq: reLoca initial} such that the union of $95\%$ \textup{confidence ellipse of} $\mathcal{N}(m_s,C)$ covers $\tilde{p}_n(X_{n,h})$.\;
\For {$s=1,2,...,N$}
{
\If{ 
 \textup{there are less than} $M_h^
{init}$ \textup{measurements in the} $95\%$ \textup{confidence ellipse of} $\mathcal{N}(m_s,C)$
}
{
Set $\mathcal{F}_{n,h}^{s}=-\infty$.\;
\textbf{continue} 
}
Evaluated $q_n^{(0)}(\theta_{n})$ via \eqref{eq: reLoca initial}, and initialise $q_n(\theta_{n})= q_n^{(0)}(\theta_{n})$.\;

\For {$i=1,2,...,I$}  
   {
   Evaluate $\overbar{R}_n^h,\overbar{Y}_n^h$ according to \eqref{eq:pseudomeas R}\eqref{eq:pseudomeas Y}.\;
   Update $q_n(X_{n,h})$ via \eqref{eq:KF update} where  $\mu^{k*}_{n|n-1},\Sigma^{k*}_{n|n-1}$ are replaced by the mean and covariance of $\tilde{p}_n(X_{n,h})$.\;
   
   Evaluate the ELBO $\mathcal{F}_{n,h}^{(i)}$ according to \eqref{eq: reloc ELBO}.\;
    \If{$\mathcal{F}_{n,h}^{(i)}-\mathcal{F}_{n,h}^{(i-1)}<\epsilon \land i\geq 2$}
      {
      \textbf{break} 
      }
    Update $q_n(\theta_{n})$ according to \eqref{eq: reloc theta update}.\;
   }
  Set $\mathcal{F}_{n,h}^{s}=\mathcal{F}_{n,h}^{(i)}$, and $q^{s}_{n}(X_{n,h})=q_n(X_{n,h})$.\;
  }
  Find the best index $w=\argmax_s \mathcal{F}_{n,h}^{s}$.\;
  \eIf{$\sum_{j=0}^{M_n} q^w_{n}(\theta_{n,j}=h)\geq M_h^{reloc}$}
  {
  Update $q_n^*(X_n)\gets  q_n^w(X_{n,h})q_n^*(X_{n,h-})$, i.e. update $\mu_{n|n}^{h*}$, $\Sigma_{n|n}^{h*}$ as the mean and covariance of $q_n^w(X_{n,h})$.\; 
  Update $\mathcal{K}_n\gets\mathcal{K}_n\bigcup\{h\}$.\;
  }
  {
  Update $q_n^*(X_n)\gets  \tilde{p}_n(X_{n,h})q_n^*(X_{n,h-})$, i.e. update $\mu_{n|n}^{h*}$,  $\Sigma_{n|n}^{h*}$ as the mean and covariance of $\tilde{p}_n(X_{n,h})$\;
  }
 }  
 Update $q_n^*(\theta_n)$ with the refined $q_n^*(X_n)$, i.e. using \eqref{eq: reloc theta update} where $\mu_{n|n}^{h},\Sigma_{n|n}^{h}$ are replaced by $\mu_{n|n}^{h*},\Sigma_{n|n}^{h*}$.
 \caption{Relocation strategy at time step $n$}
 \label{Algo:relocation}
\end{algorithm}

Finally, the relocation strategy of lost track objects for VB-AbNHPP tracker can be summarised in Algorithm \ref{Algo:relocation}. In brief, for each missed object, it first 1) determines the uninformative Gaussian prior and settings for $N$ initialisations; then 2) carries out multiple CAVI with all eligible initialisations as discussed above, and finds the converged variational distribution with the highest ELBO; 3) if the found variational distribution is convincing enough, use it as the obtained posterior of this missed object; otherwise give up the relocation for this object. For each missed object, the primary cost of the relocation strategy in Algorithm \ref{Algo:relocation} stems from performing $N_{\text{elig}}$ full iterations of CAVI, where $N_{\text{elig}}$ is the number of eligible initialisations containing at least $M_h^{\text{init}}$ measurements within the 95\% confidence ellipse. These $N_{\text{elig}}$ CAVIs can proceed in parallel. The major expense for each CAVI is the iterative update of $q_n(\theta_n)$, while the cost of updating $q_n(X_{n,h})$ for an individual object is negligible. The iterative update of $q_n(\theta_n)$ primarily requires computing $(K+1)M_n$ likelihoods (i.e. the $\Lambda_kl_k$ in \eqref{eq: reloc theta update}), all of which are parallelisable.

A final remark on our relocation strategy is that Algorithm \ref{Algo:relocation} does not always relocate all objects in $\mathcal{L}_n$ at every time step $n$. Rather, it only relocates the lost track object when it is confident about its position. This is because even though the obtained variational distribution that achieves the highest ELBO is expected to find the most likely object's location, it may still not capture the true position. If we relocate a object to the wrong position, the algorithm may take many time steps to realise this relocation is wrong and the object may be lost for too long, which may further deteriorate the tracking performance. Therefore, Algorithm \ref{Algo:relocation} will first determine whether the obtained variational distribution is convincing enough by using an effective relocation criterion. If this criterion is satisfied, the algorithm refines $q_n^*(X_n)$ with the obtained variational distribution; if not, the prior $\hat{p}(X_{n,h})$ will be employed to update the $q_n^*(X_n)$. Such an effective relocation criterion, which has already been specified in Algorithm \ref{Algo:relocation}, will be explained in Appendix \ref{apx:choice of Mreloc and Minit}.

In the following subsections, we will first describe the track loss detection strategy, and then explain the rationale of the effective relocation criterion in Algorithm \ref{Algo:relocation}. Finally, we will apply all these techniques to develop the full VB-AbNHPP-RELO tracker in Section \ref{sec:full tracking with relocation}.
\subsection{Track Loss Detection} \label{sec: lost track detect}
Our relocation strategy is only triggered when a track loss is detected. To this end, the (estimated) number of measurements produced by a object is used to monitor whether the object is tracked properly. The exact number of measurements $M_{n,k}$ generated at time step $n$ by the object $k$ can be defined with the association $\theta_n$, i.e. $M_{n,k}=\sum_{j=0}^{M_n} I(\theta_{n,j}=k)$. A convenient point estimate of this object measurement number can be computed as a byproduct of the VB-AbNHPP using the converged variational distribution $q^*_{n}(\theta_n)$, i.e.
\vspace{-0.5em}
\begin{align}  \label{eq: estimaed measurement number}
    \hat{M}_{n,k}=\E_{q^*_{n}(\theta_n)}M_{n,k}=\sum_{j=0}^{M_n} q^*_{n}(\theta_{n,j}=k).\vspace{-0.5em}
\end{align}

If the algorithm fails to track the object $k$, the estimated object position may have fallen in an area with no real measurements but only sporadic clutter; in this case,  $\hat{M}_{n,k}$ is expected to be very small as there may be few measurements associated with object $k$. Therefore, a criterion for detecting a track loss event is that $\hat{M}_{n,k}$ is too low. 

Specifically, we will record $\hat{M}_{t,k}$ for each time step $t$ and each object $k$, and if the total (estimated) number of measurements generated during $\tau_k$ ($\tau_k\geq1$) time steps before the current time step $n$ falls into a certain positive threshold $M^{los}_k$, in other words, if the following event $E$ happens, 
\vspace{-0.5em}
\begin{align} \label{eq:lost track event}
    \text{event} \ E:\sum_{t=n-\tau_k+1}^n \hat{M}_{t,k}\leq M_k^{los},\vspace{-0.5em}
\end{align}
we conclude that the track of the object $k$ has been lost. The details about the choice of the parameters $\tau_k$ and $\hat{M}_{n,k}$, including an automatic parameter selection strategy, are presented in Appendix \ref{apx:choice of tau and M}.

\setlength{\textfloatsep}{\textfloatsepsave}
\vspace{-1em}
\subsection{Effective Relocation Criterion}
Similar to the track lost detection in Section \ref{sec: lost track detect}, the (estimated) number of measurements produced by a relocated object is used to determine whether the obtained variational distribution is convincing enough. Specifically, the estimated number of measurements generated from the object $h\in\mathcal{L}_n$ based on the obtained variational distribution $q_n^w(X_{n,h},\theta_n)$ is $\sum_{j=0}^{M_n} q^w_{n}(\theta_{n,j}=h)$. The effective relocation criterion in Algorithm \ref{Algo:relocation} is satisfied when this estimated measurement number is greater than the relocation threshold $M_{h}^{reloc}$, i.e.
\vspace{-0.8em}
\begin{equation} \label{eq: effective relocation criterion}
    \sum_{j=0}^{M_n} q^w_{n}(\theta_{n,j}=h)\geq M_h^{reloc} \vspace{-0.5em}
\end{equation}
The reason is that a high estimated measurement number is more likely to have been generated by the object itself, than by uniformly distributed clutter that happened to cluster in a small area by chance. As a result, a high estimated measurement number makes it more likely that the obtained variational distribution $q_n^w(X_{n,h})$ has successfully relocated the object $h$.
The selection of the relocation threshold $M_{h}^{reloc}$ will be discussed in Appendix \ref{apx:choice of Mreloc and Minit}.
\vspace{-1em}
\subsection{Full VB-AbNHPP Tracker with Relocation Strategy} \label{sec:full tracking with relocation}
Finally, we are in a position to present the complete VB-AbNHPP tracker with relocation (VB-AbNHPP-RELO) algorithm.
For simplicity, we only present the VB-AbNHPP-RELO tracker for tasks where the Poisson rate $\Lambda$ and measurement covariance $R_k$ $(k=1,2,...K)$ are known or have been accurately estimated. This allows our algorithm to robustly track closely-spaced objects under extremely heavy clutter. In this setup, the track loss detection parameters $\tau_k,M_k^{los}$, relocation threshold $M_{k}^{reloc}$, and eligible initialisation threshold $M_{k}^{init}$ ($k=1,2,...,K$) can be easily and automatically selected according to Appendix \ref{apx:parameter selection}. These parameters, along with initialisation covariance $C$ and convergence monitoring parameters $I,\epsilon$, can all be determined before performing tracking tasks and are not required to be updated in subsequent time steps. For convenience, all these parameters will not be listed as inputs in our single time step VB-AbNHPP-RELO procedure summarised in Algorithm \ref{Algo:full tracking with relocation}.
Note that in Algorithm \ref{Algo:full tracking with relocation}, the standard VB-AbNHPP tracker with known $\Lambda$ and $R_k$ is incorporated, which has been outlined in Algorithm 1 in \cite{gan2022variational}. In a nutshell, Algorithm \ref{Algo:full tracking with relocation} first 1) performs the standard VB-AbNHPP tracker for the current time step, then 2) updates the missed object set $\mathcal{L}_n$ by track loss detection, and finally 3) implements the relocation only for missed objects in $\mathcal{L}_n$.

For objects relocated successfully at the current time step, their estimated measurement numbers at previous time steps are in principle unavailable to us. However, the track loss detection in \eqref{eq:lost track event} requires the estimated measurement number from previous $\tau_k$ time steps. Thus, in order to ensure the track loss detection can be directly carried out in the next time step, Algorithm \ref{Algo:full tracking with relocation} automatically sets the estimated measurement number at previous time steps as object's Poisson rates. This empirical setup allows for the timely detection of track loss for any properly tracked object at any time step.

\begin{algorithm}  
\SetAlgoLined
\algsetup{linenosize=\tiny}
  \scriptsize
\textbf{Require}: $q^*_{n-1}(X_{n-1}),Y_n,M_n$, successfully tracked/relocated set $\mathcal{K}_{n-1}$.
If $\tau_k\geq2$, also require $\hat{M}_{t,k}$ ($k\in\mathcal{K}_{n-1}$ and $t$ from $n\!-\!\tau_k\!+\!1$ to $n\!-\!1$).\\

\textbf{Output}: $q^*_{n}(X_{n}),\mathcal{K}_{n}$. If $\tau_k\geq2$, also output $\hat{M}_{t,k}$ for all $k\in\mathcal{K}_{n}$ and $t$ from $n\!-\!\tau_k\!+\!2$ to $n$.\\

Run the standard tracker with known $\Lambda$, i.e. Algorithm 1 in \cite{gan2022variational}, to obtain $\hat{p}_{n|n-1}(X_n),q^*_{n}(X_{n})$, and $q^*_{n}(\theta_{n})$.\;
Initialise $\mathcal{L}_n=\{k\in\{1,2,...,K\}:k\notin\mathcal{K}_{n-1}\}$.\;
\ForEach {$k\in \mathcal{K}_{n-1}$}
{ Evaluate $\hat{M}_{n,k}$ according to \eqref{eq: estimaed measurement number}.\;
\If{$\sum_{t=n-\tau_k+1}^n \hat{M}_{t,k}\leq M_k^{los}$}
{
Update $\mathcal{L}_n\gets\mathcal{L}_n\bigcup\{k\}$.
}
}
\If{$\mathcal{L}_n\neq\emptyset$}
{
Run the relocation procedure (i.e. Algorithm \ref{Algo:relocation}) to obtain $\mathcal{K}_n$ and the refined $q^*_{n}(X_{n})$ and $q^*_{n}(\theta_{n})$.\;
\ForEach{$k\in \mathcal{K}_{n}$}
{
Update $\hat{M}_{n,k}$ via \eqref{eq: estimaed measurement number} with the refined $q^*_{n}(\theta_{n})$.\;
\If{$k\in \mathcal{L}_n\land\tau_k\geq3$}
{ Set $\hat{M}_{t,k}=\Lambda_k$ for $t$ from $n\!-\!\tau_k\!+\!2$ to $ n\!-\!1$. }

}
}
 \caption{Tracking with relocation at time step $n$}
 \label{Algo:full tracking with relocation}
\end{algorithm}

It is yet  be discussed the choice of the uncertain Gaussian prior $\tilde{p}_n(X_{n,h})$ for the relocation of the missed object $h\in\mathcal{L}_n$ in Algorithm \ref{Algo:full tracking with relocation} and Algorithm \ref{Algo:relocation}. Typically the mean of $\tilde{p}_n(X_{n,h})$ should be set to the object $h$'s latest state estimate when it is properly tracked. Moreover, it is practically efficient to set a relatively small covariance of $\tilde{p}_n(X_{n,h})$ for the object $h$ whose track has just been lost. For the object $h$ that has been missed for a longer time and failed to be captured within a small survey area, a large covariance of $\tilde{p}_n(X_{n,h})$ should be set so that more time will be spent searching for it in a large survey area. Finally, it is useful to incorporate the prior information of the sensor surveillance region in $\tilde{p}_n(X_{n,h})$. For example, $\tilde{p}_n(X_{n,h})$ should not have a high probability density over areas not covered by the sensor surveillance region or where the object is definitely not present.
\vspace{-0.5em}
\section{Simulation} \label{sec: simulation}
This section focuses on the evaluation of the proposed full VB-AbNHPP-RELO tracker in Algorithm \ref{Algo:full tracking with relocation} on multi-object tracking tasks in heavy clutter.
Particularly, subsection \ref{sec: result moderately clutter} considers a tracking scenario under moderately heavy clutter where objects exhibit relatively more random movements, and subsection \ref{sec: result highly heavy clutter} presents a more challenging scenario with highly heavy clutter where objects first intersect in a small region and then disperse. Extensive experimental results are reported in both scenarios. Finally, in subsection \ref{sec: result rate estimation}, we briefly demonstrate that the object and clutter's Poisson rates required in the VB-AbNHPP-RELO algorithm can be accurately estimated using the proposed VB-AbNHPP tracker with unknown rates in Algorithm \ref{Algo:tracker}. Additional results for the standard VB-AbNHPP tracker without the relocation strategy in a relatively low clutter tracking scene can be found in \cite{gan2022variational}.

In subsections \ref{sec: result moderately clutter} and \ref{sec: result highly heavy clutter}, we compare the proposed VB-AbNHPP-RELO with our original VB-AbNHPP tracker \cite{gan2022variational} and other competing methods (specified later) to demonstrate its efficacy on the fixed number multi-object tracking tasks. Datasets are generated with various trajectories, object number, and measurement rates to simulate diverse tracking scenes with heavy clutter. Examples of single-time-step measurements in a challenging dense clutter scene are given in Fig. \ref{fig:moderatelyclutter}(c) and Fig. \ref{fig: highly heavy clutter}(c), where objects' measurements (red dots) are difficult to distinguish from the clutter (blue dots). Such an ambiguity may lead to multiple modes in the exact posterior of object states, whereas only one of them  may capture the true objects' positions. An algorithm that fails to keep necessary modes can easily lose track of the objects, which explains why tracking in heavy clutter situations is challenging.

The methods evaluated in subsections \ref{sec: result moderately clutter} and \ref{sec: result highly heavy clutter} can be divided into two classes: The first class is methods based on sampling or maintaining multiple hypotheses, where several modes may be kept in the estimate per time step at the cost of computational time and memory. This class of methods includes Rao-Blackwellised Gibbs-AbNHPP (G-AbNHPP) (Algorithm 2, \cite{li2023adaptive}) and the popular PMBM filter \cite{granstrom2019poisson}. In addition, besides the standard PMBM filter with a birth model (termed PMBM-B), we also implement an altered version of the PMBM filter with no birth and death models (PMBM-NB) as it matches our modeling assumptions and we found it can lead to better performance and faster implementation in the considered tasks. However, note that the original recycling step is kept in PMBM-NB, such that a birth process will still be carried out in a small region.
Another tracking algorithm that belongs to this class is the SPA-based tracker \cite{meyer2021scalable,meyer2020scalable}. Since the target number is known, here we compare with the SPA-NB tracker with no target birth that corresponds to the fixed target version in \cite{meyer2020scalable}. 
The other class of evaluated algorithms in our experiments only keep one mode in the estimate per time step so that the tracking can be performed in the most efficient manner.
These algorithms include the ET-JPDA filter in \cite{yang2018linear}, our original VB-AbNHPP tracker and the proposed VB-AbNHPP-RELO, of which only the VB-AbNHPP-RELO has the capability to relocate the missed objects.

Some common settings and parameters for the simulations and tested algorithms are listed below:

\paragraph{Modelling assumptions in synthetic dataset} \label{sec:model para}
We assume that objects move in a 2D surveillance area with each $X_{n,k}=[x^1_{n,k}, \Dot{x}^1_{n,k},x^2_{n,k}, \Dot{x}^2_{n,k}]^T$, where $x^d_{n,k}$ and $\Dot{x}^d_{n,k}$ ($d=1,2$) indicate the $k$-th object's position and velocity in the $d$-th dimension, respectively. The following models and parameters are employed in all synthetic tracks generation and inference algorithms presented below. We utilise the constant velocity model with a transition density in \eqref{eq: dynamic transition}, where we have $F_{n,k}=diag(F^1_{n,k},F^2_{n,k})$, $Q_{n,k}=diag(Q^1_{n,k},Q^2_{n,k})$, $B_{n,k}=0$, with $F_{n,k}^d,Q_{n,k}^d$ ($d=1,2$) being specified in \eqref{eq:model para}.
\vspace{-0.4em}
\begin{align} \label{eq:model para}
    F_{n,k}^d=\begin{bmatrix} 1 & \tau \\ 0& 1
    \end{bmatrix}, 
    Q_{n,k}^d=25\begin{bmatrix} \tau^3/3 & \tau^2/2 \\ \tau^2/2& \tau
    \end{bmatrix},
    H^d=\begin{bmatrix} 1 & 0 
    \end{bmatrix}. \vspace{-2em}
\end{align}
The total time steps are 50, and the time interval between measurements is $\tau=1$s. The object is simulated as an extended object with an elliptical extent, and its covariance in \eqref{measurement model} is set to $R_k=100\text{I}_2$, $k=1,...,K$. The measurement matrix in \eqref{measurement model} is designed to generate the 2D positional measurements, i.e. $H=diag(H^1,H^2)$, where $H^d$ ($d=1,2$) is specified in \eqref{eq:model para}.

\paragraph{Metric}
The metric utilised for measuring the tracking performance is the optimal sub pattern assignment (OSPA) \cite{schuhmacher2008consistent}, where the order is set to $p=1$ and the cut-off distance is $\text{c}=50$. Meanwhile, to evaluate the computational complexity, we monitor the CPU time (System: Intel(R) Core(TM) i9-9980 CPU@2.40GHz, 32 GB RAM) required at a single time step and averaged over all time steps. The overall CPU time presented in the table is the averaged value across all runs. 

\paragraph{General parameter settings}
We set $I=100,\epsilon=0.01$ for both the standard VB-AbNHPP tracker and the VB-AbNHPP-RELO tracker. In our experiments, our methods usually converge in less than 10 iterations under this setup. 

For the VB-AbNHPP-RELO tracker, we search for the recently missed objects (for $h\in\mathcal{L}_n$ but not in $\mathcal{L}_{n-1}$) within a $490$ radius, and all other missed objects (for $h \in \mathcal{L}_n\bigcap\mathcal{L}_{n-1}$) within a $1713$ radius, both centred on the last tracked position. Such a circular area is modelled by a $95\%$ confidence region of a $2D$ Gaussian with covariance $200\text{I}_2$ and $700\text{I}_2$ respectively. The Gaussian prior $\tilde{p}_n(X_{n,h})$ assigned for the missed object $h$ is designed to cover the intersection of this circular region and the sensor surveillance area, that is, if this circular region is completely within the surveillance area, the $\tilde{p}_n(X_{n,h})$ has the corresponding mean and covariance $200\text{I}_2$ or $700\text{I}_2$; otherwise it is adjusted to cover the intersection of these two regions.
The velocities in $\tilde{p}_n(X_{n,h})$ are set with mean $0$ and covariance $1600$ for all dimensions. Eligible initialisation threshold $M_{k}^{init}$ is set to be equal to the relocation threshold $M_{k}^{reloc}$, and $M_{k}^{reloc}$ is set according to Appendix \ref{apx:choice of Mreloc and Minit} with $P_{thres}^{reloc}=0.5$. For the moderately heavy clutter scenes in subsection \ref{sec: result moderately clutter}, the track loss detection parameters $\tau_k,M_k^{los}$ are set according to Appendix \ref{apx:choice of tau and M} with $P_{thres}^{los}=$7e-4, and the initialisation covariance $C$ is $35^2\text{I}_2$; for the highly heavy clutter scenarios in subsection \ref{sec: result highly heavy clutter}, we set $\tau_k,M_k^{los}$ according to Appendix \ref{apx:choice of tau and M} with $P_{thres}^{los}=$5e-4 to achieve a slightly less sensitive track loss detection, and $C=20^2\text{I}_2$ to maintain a reasonable relocation effectiveness with larger computational costs. See sections \ref{sec: single object vari loca strate} and \ref{sec: object loca demo} for the rationale for selecting $C$.

For the G-AbNHNPP algorithm, the number of particles is set to 500 and the burn-in time is 100 iterations. We choose to compare with the fast Rao-Blackwellisation scheme of the G-AbNHPP since it outperforms other schemes in \cite{li2023adaptive}. A higher tracking accuracy may be achieved with a larger particle size, but this comes with a longer computational time.

Since PMBM filters were originally devised for tracking a varying number of objects, we have heavily modified them for tracking a fixed number of $K$ objects for a fair comparison in our experiment. In particular, for both PMBM-B and PMBM-NB, if their multi-Bernoulli global hypothesis with the highest weight includes more than $K$ Bernoulli components, then the OSPA is evaluated with only the $K$ components with the highest existence probability. Under this setup, we find that their cardinality estimates in our experiment are rarely wrong, and hence their OSPA merely reflects the localisation error in most cases. Moreover, both PMBM filters are set with the same accurate initialisation of object states, Poisson rate $\Lambda$, and object extent $R$ as in other competing methods. Specifically, at time step $1$, both PMBM filters start with a single multi-Bernoulli hypothesis with $K$ components, each associated with an existence probability $1$, and a ground truth object state, object extent, and Poisson rate (i.e., their priors all have an extremely small variance). 
The transitions of rates and shapes are removed such that the rates and shapes are nearly constant over time, allowing the PMBM filters to use ground truth information of $\Lambda$ and $R$ at every time step, like the other compared methods.

Similarly, the SPA-NB tracker in \cite{meyer2020scalable} is initialised with the ground truth target number $K$ and target state prior. The survival probability is set to $1$. The number of particles is set to 500 with 3 iterations. However, due to its computational time being much beyond the real-time processing requirement (e.g., for the 10-target case in Section \ref{sec: result moderately clutter}, it already took over 69s per time step with only 500 samples and 3 iterations), we only include comparison results for cases up to 10 targets with moderate clutter in Section \ref{sec: result moderately clutter}. 

Other parameters of the PMBM-B and the PMBM-NB filters are listed as follows: For all scenarios, both the object’s survival probability and the detection probability are set to 1. The ellipsoidal gate size in probability is 0.999. The number of global hypothesis is capped at 100. The DBSCAN algorithm is run with distance thresholds between 5 and 50, with a maximum number of 20 assignments. For PMBM-NB, the Poisson birth model is neglected by setting the intensity as a single component with weight 0.01 and Gaussian density with mean $[0, 0, 0, 0]^T$ and covariance $ diag(0.01, 0.25, 0.01, 0.25)$, covering a negligible surveillance area. Note that there will still be Poisson intensity for undetected objects because of the recycling step in the PMBM algorithm. For PMBM-B, its Poisson intensity is set as a single component with weight 0.1 and Gaussian density with mean $[0, 0, 0, 0]^T$ and covariance $ diag(40000, 16, 40000, 16)$, covering the circular region where the objects most likely appear/lose track. These parameters are manually tuned for good tracking performance and efficient implementation. For example, the grid for DBSCAN clustering is carefully selected such that a finer grid no longer improves the tracking accuracy but only introduces extra computational time in our tracking examples.

\subsection{Tracking Multiple Objects under Moderately Heavy Clutter} \label{sec: result moderately clutter}

\begin{figure*}
\centerline{\includegraphics[width=17cm]{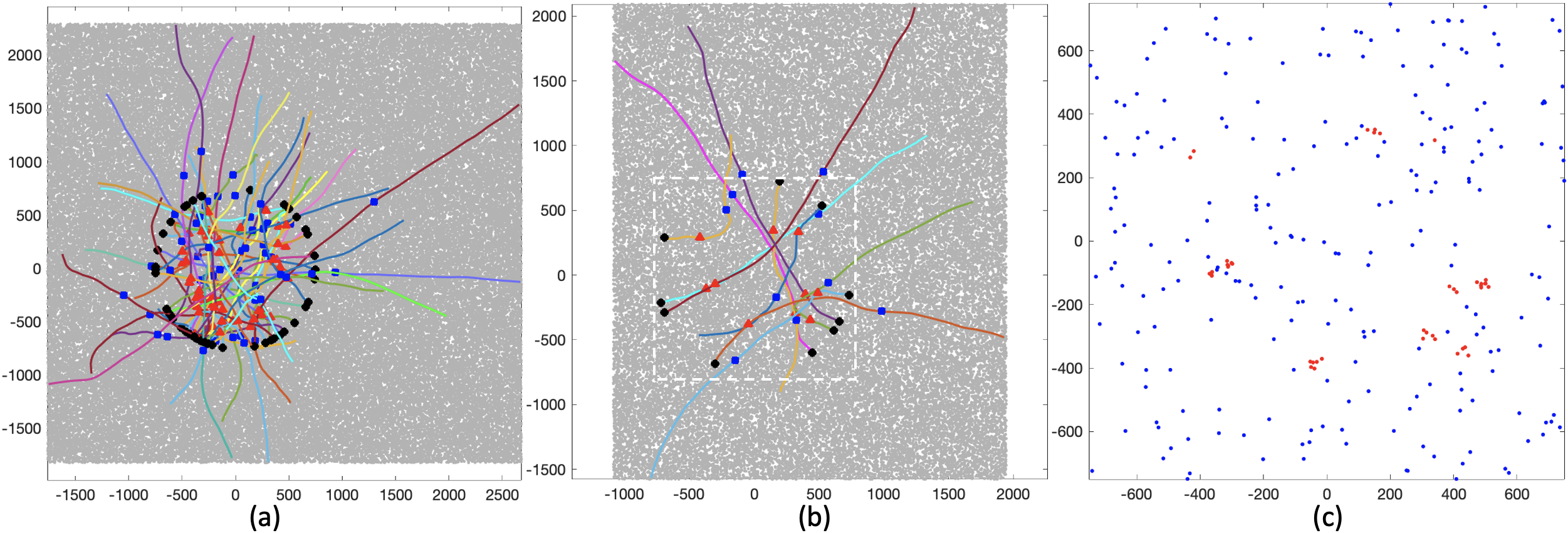}}
\vspace{-0.5em}
\caption{Example synthetic trajectories and measurements in the tested dataset in Section \ref{sec: result moderately clutter}. Object number $K$ in figure (a) and (b) are $50$ and $10$ respectively. The black circles, red triangles and blue squares mark the objects' positions at the time step $1$, $11$ and $31$. The dense grey dots in (a) and (b) are all measurements received from all time steps. Figure (c) shows the measurements received at the time step $11$ from the same dataset as figure (b). Red dots are the measurements generated by the $10$ objects (whose position is in the red triangles in figure (b)), and blue dots are the clutter received at the time step $11$. Figure (c)'s range $[-750,750]\times[-750,750]$ corresponds to the white dashed square depicted in figure (b).}
\label{fig:moderatelyclutter}
\end{figure*}

\begin{table*}[ht]
\centering
\caption{Tracking performance comparisons for closely spaced randomly moving objects under moderately heavy clutter.}
\begin{tabular}{*8c}
\toprule
 &  & \multicolumn{5}{c}{OSPA (mean $\pm 1$ standard deviation) $|$ CPU time (s)}  \\
\cmidrule(lr){2-8}
$K$ &  VB-AbNHPP-RELO  & VB-AbNHPP & G-AbNHPP   & ET-JPDA & PMBM-NB & PMBM-B & SPA-NB  \\
\midrule
5 & 5.70$\pm$0.41 $|$ 3e-3  & 6.12$\pm$1.83 $|$ 5e-5 &\textbf{5.65}$\pm$0.43 $|$ 0.40  & 7.24$\pm$2.31 $|$ 0.02 & 6.19$\pm$1.32 $|$ 0.44  & 7.29$\pm$2.26 $|$ 4.36 & 5.96$\pm$1.03 $|$ 15.87 
  \\
 10& \textbf{5.72}$\pm$0.35 $|$ 0.01& 6.24$\pm$1.34 $|$ 1e-4  & 5.79$\pm$0.84 $|$ 0.88 & 21.42$\pm$4.13 $|$ 0.08 & 6.65$\pm$1.23 $|$ 1.86   & 8.09$\pm$2.11 $|$ 5.63 & 6.19$\pm$1.13 $|$ 69.40   \\
20& \textbf{5.95}$\pm$0.41 $|$ 0.05& 7.05$\pm$1.34 $|$ 3e-4 & 6.16$\pm$0.95 $|$ 1.98  & 22.47$\pm$2.73 $|$ 0.22 & 8.67$\pm$2.29 $|$ 2.43 & 12.17$\pm$2.94 $|$ 11.57 & \textemdash     \\
 30 & \textbf{6.06}$\pm$0.33 $|$ 0.14  & 7.35$\pm$1.29 $|$ 5e-4 & 6.78$\pm$1.24 $|$ 2.76  & 24.09$\pm$2.26 $|$ 0.39 & 10.20$\pm$1.80 $|$ 3.24 & 15.63$\pm$2.74 $|$ 16.63 & \textemdash    \\
 50  & \textbf{6.45}$\pm$0.34 $|$ 0.87& 8.64$\pm$1.13 $|$ 1e-3 & 7.95$\pm$1.46 $|$ 4.12 & 25.06$\pm$1.82 $|$ 0.86 & 14.09$\pm$1.71 $|$ 6.14& 19.97$\pm$2.32 $|$ 40.98 & \textemdash  \\
\bottomrule
\end{tabular}
\label{rmse}
\vspace{-1.5em}
\end{table*}
This section evaluates the proposed VB-AbNHPP trackers for tracking randomly moving objects that are closely spaced in moderately heavy clutter. Specifically, five scenarios are considered with the object number increasing from $5, 10, 20, 30$ to $50$. For each scenario, 100 datasets are generated, each with a completely different set of synthetic objects' trajectories and measurements. Each synthetic trajectory in each dataset was randomly initiated on a circle with a radius of $750$ from the origin $(0,0)$, and it has an initial velocity of $30$ pointing toward the origin. This simulates a scenario where objects move closer together before gradually spreading out. Two example tracks from $50$ objects and $10$ objects are shown in Fig. \ref{fig:moderatelyclutter}(a) and (b), where black circles, red triangles and blue squares mark the objects' positions at time step $1$, $11$ and $31$ respectively. It can be seen that during the first 30 time steps, objects are closely spaced, and track coalescence occurs frequently with its frequency and severity increasing with object number. After the $30$-th time step, the objects gradually move away from each other and tracking becomes less difficult.


For all datasets of this moderately heavy clutter scene, we set the Poisson rate $\Lambda_k=5$ for all objects ($k=1,2,...,K$). The clutter density (i.e. $\Lambda_0/V$) is $10^{-4}$ per unit area, and correspondingly, the average $\Lambda_0$ for each scenario with object number $K$ being $5,10,20,30$ and $50$ are $775, 1175, 1761, 1967$ and $2521$, respectively. With such a high clutter, the measurements received from all time steps (showed as grey dots in Fig. \ref{fig:moderatelyclutter}(a) and (b)) are visually overlapped. The measurements received from a single time step are also plotted in Fig. \ref{fig:moderatelyclutter}(c) where the red dots are measurements generated by 10 objects that are located at the red triangles in Fig. \ref{fig:moderatelyclutter}(b), and the blue dots are the clutter received in this single time step. Such dense clutter causes ambiguity in objects' true positions, highlighting another challenge of our experiment, in addition to the frequent coalescences that occur when object numbers are high.

The tracking results for all methods are presented in Table \ref{rmse}. For each $K=5,10,...,50$, the mean OSPA and CPU time are first averaged over all 50 time steps, and then averaged over all 100 datasets. The standard deviation of these 100 average OSPA is also shown in Table \ref{rmse}. We can observe that our method has a promising performance in both time efficiency and tracking accuracy. In terms of efficiency, our standard VB-AbNHPP tracker is the fastest algorithm, and its efficiency advantage becomes more evident as the object number increases. The proposed VB-AbNHPP-RELO tracker naturally requires extra computational time than the standard VB-AbNHPP tracker, but it is still more efficient than all other compared methods, and the advantages over the SPA-NB tracker, PMBM filter and G-AbNHPP are particularly significant. In terms of tracking accuracy, the proposed VB-AbNHPP-RELO outperforms all other compared methods in all considered scenarios, except for being comparable to G-AbNHPP with similar average OSPA values.
We can also see that the SPA-NB method took a much longer time compared with all other trackers in both cases of $K=5,10$. Since its computational time grows very fast and beyond real-time processing ability, we only present results up to 10 targets.
Based on the experimental results, the remaining tested algorithms, except SPA-NB, can be ranked in terms of their average OSPA from lowest to highest as follows: G-AbNHPP, our standard VB-AbNHPP tracker, PMBM-NB, PMBM-B, and ET-JPDA.

\setlength{\textfloatsep}{0.2pt}
\begin{figure}[htp!]
\centerline{\includegraphics[width=8cm]{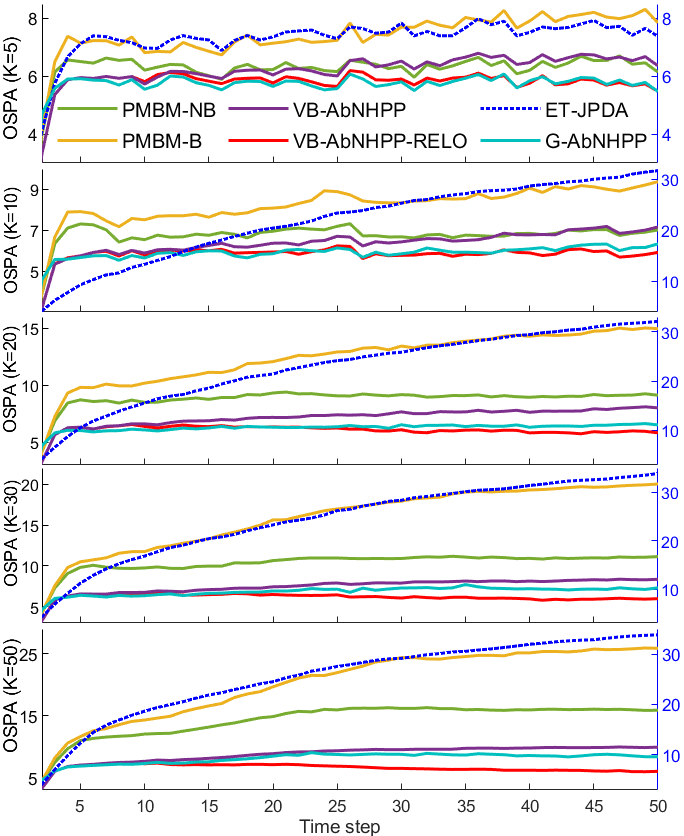}}
\caption{Mean OSPA metric over 50 time steps for tracking scene in Section \ref{sec: result moderately clutter}. Blue dashed line is associated with the right y-axis, and all other lines are with the left y-axis. }
\label{fig:ospa}
\end{figure}

\begin{figure}[htp!]
\centerline{\includegraphics[width=8cm]{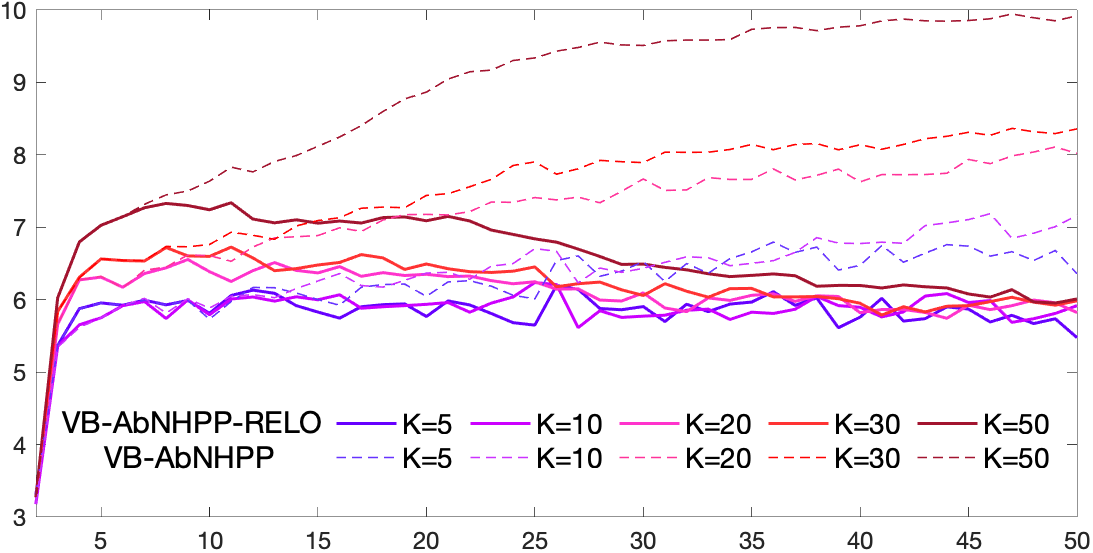}}
\caption{Mean OSPA of VB-AbNHPP-RELO and standard VB-AbNHPP over 50 time steps for cases in Section \ref{sec: result moderately clutter}. }
\label{fig:ospa moderately clutter}
\end{figure}

To illustrate detailed tracking performance over time, the mean OSPA over 50 time steps for each scenario is plotted in Fig. \ref{fig:ospa}. In the beginning stage of each tracking scene, the G-AbNHPP has the lowest mean OSPA, followed by our variational Bayes trackers. This advantage of G-AbNHPP lasts shorter as tracking scenarios become more challenging, i.e., when $K$ grows. As time goes on, the proposed VB-AbNHPP-RELO outperforms G-AbNHPP and always achieves the lowest OSPA in the last few time steps, owing to the proposed effective relocation strategy. Such a phenomenon about G-AbNHPP and VB-AbNHPP-RELO is reasonable, and we analyse it below.


Recall that all tested methods are given an identical and accurate initialisation of object states, and the G-AbNHPP can retain different modes in the posterior with different samples at the cost of computational time and memory. In the beginning of each tracking scene, a limited number of samples in G-AbNHPP are likely to retain all necessary modes and thus all objects can be tracked successfully. In contrast, our efficient VB-AbNHPP-RELO always keeps a single mode at each time step, making it more likely to lose track. Although the proposed relocation strategy can retrieve the missed objects, the track loss prior to the successful relocation inevitably leads to a deterioration of the mean OSPA. Therefore, the G-AbNHPP typically outperforms the VB-AbNHPP-RELO in the beginning of a tracking task. 
However, in each new time step, there is always a chance that the limited number of samples in G-AbNHPP may miss some important mode that can potentially lead to track loss. Subsequently, as time goes on, G-AbNHPP is more likely to lose track and has an irreparably high OSPA due to its inability to relocate missed objects. For more challenging tracking tasks where there are more modes in the posterior, G-AbNHPP are less likely to retain all the necessary modes using the same sample size, resulting in more frequent track losses. Accordingly, the OSPA difference (both in Table \ref{rmse} and in the last few time steps in Fig. \ref{fig:ospa}) between G-AbNHPP and VB-AbNHPP-RELO becomes more significant as $K$ increases, since the latter maintains a stable tracking performance by constantly detecting and relocating the missed objects across all tracking scenes. 

Although outperformed by our more efficient VB-AbNHPP-RELO in challenging tracking tasks (e.g. $K\geq10$), the G-AbNHPP is still advantageous over all other tested methods. 
Other tested methods that retain multiple modes are PMBM-NB and PMBM-B, where different hypotheses of measurement associations are recorded at each time step. 

In particular, PMBM-NB does better than PMBM-B in all considered scenarios. However, they do not perform very well: apart from outperforming the ET-JPDA,  both PMBM filters have a higher overall OSPA in Table \ref{rmse} than all other methods, including our standard VB-AbNHPP that is also over thousands of times faster.

This trend is also reflected in Fig. \ref{fig:ospa}, where both PMBM filters always have a higher OSPA than the G-AbNHPP and our variational Bayes trackers, except that the PMBM-NB slightly outperforms our standard VB-AbNHPP tracker in the last few time steps in the cases $K=5$ and $K=10$. One factor that could lead to the high OSPA of PMBM filters is the poor performance of the clustering algorithm in this heavy clutter, which further results in inaccurate association hypotheses.

To present a clearer picture of the efficacy of the proposed missed object detection and relocation strategy, the mean OSPAs of the standard VB-AbNHPP tracker and the VB-AbNHPP-RELO are shown in Fig. \ref{fig:ospa moderately clutter}. In the beginning of each tracking scene, the mean OSPAs of two trackers are overlapped, meaning that no track loss is detected in this period and both trackers perform exactly the same. As time goes on, while the OSPA of the standard VB-AbNHPP continues to grow, the OSPA of the VB-AbNHPP-RELO starts to decrease over the rest of time, validating clearly the superiority of the proposed detection and relocation strategy. In particular, the decline in the OSPA of the VB-AbNHPP-RELO is because 1) tracking tasks become easier in the later stage of our simulation as objects become progressively more separated; and 2) missed objects are successfully detected and relocated. Furthermore, we observe in Fig. \ref{fig:ospa moderately clutter} that although the OSPA values of the VB-AbNHPP-RELO initially peak at different levels due to varying tracking scenarios, they subsequently decrease and finally stabilises at a similar value over different $K$s. Notably, this value corresponds with the mean OSPA of our least challenging $K=5$ case, suggesting minimal or no track losses. In contrast, the OSPA for the standard VB-AbNHPP and other methods in Fig. \ref{fig:ospa} consistently increases to higher values as $K$ increases. This demonstrates the robustness of our VB-AbNHPP-RELO tracker: regardless of varying levels of track loss that might temporarily raise the OSPA due to coalescence of objects, our relocation strategy can effectively retrieve (nearly) all missed objects. Ultimately, it achieves robust tracking, stabilising at a lower OSPA value once objects are adequately separated.

Finally, we note from Fig. \ref{fig:ospa} and \ref{fig:ospa moderately clutter} that the proposed VB-AbNHPP-RELO is the only tested method whose OSPA decreases after the first $10$ time steps. This is again due to the significant advantage of our effective relocation method over other existing methods. On the contrary, the OSPAs of all other methods either grow or remain the same even when the coalescence is less severe, meaning that missed objects are seldom retrieved and more objects may be lost due to the high clutter. In particular, we may find that in such a heavy clutter tracking scene, the birth process in PMBM-B is unable to retrieve the missed objects as effectively as the proposed relocation strategy. This may be because: 1) a distinct mismatch between the birth process and our model assumptions, and 2) the birth process cannot cover the whole surveillance area due to the significant computational time.

\setlength{\textfloatsep}{\textfloatsepsave}

\vspace{-0.85em}
\subsection{Analysis of Two Example Coalescence Scenarios under Extremely Heavy Clutter}
\label{sec: result highly heavy clutter}
\vspace{-0.2em}
\begin{figure*}
\centerline{\includegraphics[width=17.5cm]{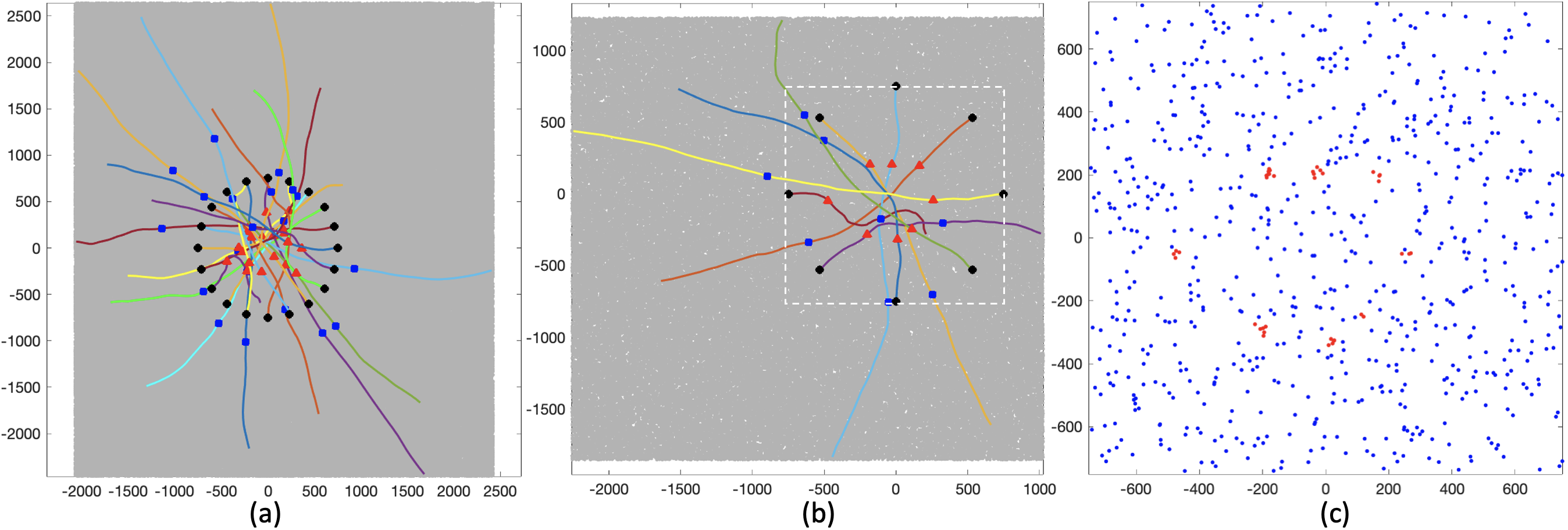}}
\caption{Trajectories and example measurements in the two considered coalescence cases in Section \ref{sec: result highly heavy clutter}. Object number $K$ in figure (a) and (b) are $20$ and $8$ respectively. The black circles, red triangles and blue squares mark the objects' positions at time steps $1$, $11$ and $31$. The dense grey dots in (a) and (b) are all measurements received from all time steps. Figure (c) shows the measurements received at the time step $11$ from the same dataset as figure (b). Red dots are the measurements generated by the $10$ objects (whose position is in the red triangles in figure (b)), and blue dots are the clutter received at the time step $11$. Figure (c)'s range $[-750,750]\times[-750,750]$ corresponds to the white dashed square depicted in figure (b).}
\label{fig: highly heavy clutter}
\vspace{-1em}
\end{figure*}

In this section, we consider two specific coalescence scenarios where multiple moving objects intersect around the origin point at a specific time. Objects' trajectories of these two cases, which feature $K=8$ and $K=20$ objects, are shown in Fig. \ref{fig: highly heavy clutter}(b) and (a), respectively. In both scenarios, the initial positions of objects are equally spaced on a circle of radius 750 from the origin, and each object's initial velocity is 50 pointing towards the origin. Compared to the relatively random object's movement in Section \ref{sec: result moderately clutter}, here the objects' trajectories are more straight, and all objects now intersect in a smaller region between time steps $11$ and $25$, see Fig. \ref{fig: highly heavy clutter}. 

Here, two coalescence scenarios are designed under extremely heavy clutter to verify the advantages of our variational trackers in more challenging settings. In particular, the clutter density is $3\times 10^{-4}$ per unit area, which leads to a Poisson rate of $3038$ for $K=8$, and $6916$ for $K=20$. Poisson rates are $6$ for all objects. For each scenario, 100 synthetic measurement sets are generated under ground truth trajectories in Fig. \ref{fig: highly heavy clutter}, where grey dots denote one example measurement set for all time steps. Fig. \ref{fig: highly heavy clutter}(c) shows the measurements received at a single time step, where the true measurements (red dots) are largely buried in heavy clutter (blue dots), which can hardly be distinguished with human eyes. Note that Fig. \ref{fig: highly heavy clutter}(c) is on the same scale as Fig. \ref{fig:moderatelyclutter}(c), and thus we can easily visualise the increase in clutter density in Fig. \ref{fig: highly heavy clutter}(c), indicating an increased difficulty in performing tracking tasks.


\begin{table}[ht]
\centering
\caption{Tracking performance comparisons for intersecting objects under highly heavy clutter.}
\begin{tabular}{*3c}
\toprule
&\multicolumn{2}{c}{\kern-2em OSPA (mean $\pm1$ standard deviation) $|$ CPU time (s)}  \\
\cmidrule(lr){2-3}
Method  & 8 objects & 20 objects    \\
\midrule
ET-JPDA & 22.03$\pm$2.48 $|$ 0.16 & 35.40$\pm$1.17 $|$ 0.96\\
PMBM-NB & 7.53$\pm$2.14 $|$ 2.15 & 9.29$\pm$1.94 $|$ 6.68 \\
PMBM-B & 8.26$\pm$2.38 $|$ 10.17 & 13.09$\pm$3.21 $|$ 21.53 \\
G-AbNHPP & 6.10$\pm$1.42 $|$ 1.31 & 6.48$\pm$1.28 $|$ 4.07 \\
VB-AbNHPP  & 6.36$\pm$1.78 $|$ 2e-4 & 7.62$\pm$1.89 $|$ 1e-3 \\
VB-AbNHPP RELO &  \textbf{5.63}$\pm$0.55 $|$ 0.02 & \textbf{6.16}$\pm$0.64 $|$ 0.44  \\
\bottomrule
\end{tabular}
\label{rmse3}
\vspace{-1em}
\end{table}

\begin{figure*}[t!]
\centerline{\includegraphics[width=16cm]{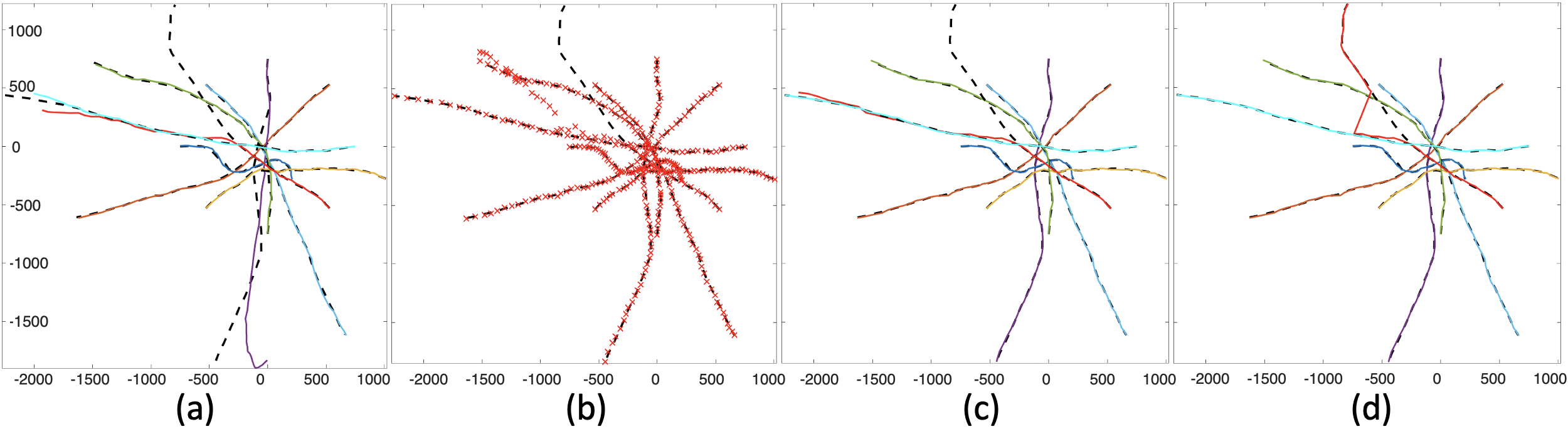}}
\caption{Estimated trajectories (colored lines and red crosses) from (a) ET-JPDA, (b) PMBM-B, (c) VB-AbNHPP and (d) VB-AbNHPP RELO for a particular measurement set of $K=8$ objects.
The black dashed lines are the ground truth. }
\label{fig:4methodest}
\vspace{-1em}
\end{figure*}

The overall tracking performance is presented in Table \ref{rmse3}. It can be observed that the performance of all tested methods exhibits a pattern consistent with that in Table \ref{rmse} in Section \ref{sec: result moderately clutter}. For both two coalescence scenarios, the proposed VB-AbNHPP-RELO has the lowest mean OSPA and the second fastest implementation. Our standard VB-AbNHPP tracker is the fastest algorithm, but its tracking accuracy is surpassed by the G-AbNHPP and VB-AbNHPP-RELO. The ranking of all tested methods in terms of the mean OSPA (from low to high) for both scenarios is: VB-AbNHPP- RELO, G-AbNHPP, standard VB-AbNHPP, PMBM-NB, PMBM-B, and ET-JPDA. 

We select a specific measurement set from $100$ datasets for $K=8$, where there are `misleading' measurements that can trick trackers into steering the estimated trajectory of a particular object in the wrong direction. The estimated trajectories of four methods (ET-JPDA, PMBM-B, VB-AbNHPP and VB-AbNHPP RELO) are depicted in Fig. \ref{fig:4methodest}. For brevity, the trajectories of PMBM-NB and G-AbNHPP are not shown, but they are similar to Fig. \ref{fig:4methodest}(b) and (c) respectively. From Fig. \ref{fig:4methodest} we can see that, once the object in red line intersects with the object in cyan line, all methods incorrectly track the red object, and the estimated tracks of the red and cyan objects coincide. However, the proposed VB-AbNHPP-RELO can detect track loss in a few time steps and successfully relocate it to its true position, whereas all other methods continue to track the red object falsely ever since its coalescence. 
In contrast, even though the track loss occur within the specified birth region in this example, the birth process in PMBM-B fails to localise the missed object. Consequently, VB-AbNHPP-RELO is the only method that correctly tracks all objects at the end of the task, demonstrating the superiority of our effective missed object detection and relocation strategy. 

Finally, the mean OSPAs at 50 time steps for all methods are shown in Fig. \ref{fig:ospa heavy} to reflect the tracking performance over time. All methods demonstrate similar patterns as in Fig. \ref{fig:ospa}, and the analysis and comparisons of the general performance of these methods can be found in Section \ref{sec: result moderately clutter}. Here we highlight two advantages of the proposed VB-AbNHPP-RELO: its robust tracking performance in handling the coalescence, and its lowest mean OSPA in the last few time steps. In particular, all other methods have a more evident increase in OSPA when the coalescence occurs between time steps $11$ and $25$, whilst the OSPA of VB-AbNHPP-RELO does not change greatly as missed objects are detected and relocated timely. For PMBM-NB and PMBM-B algorithms, the high OSPA at the moment of coalescence may be due to the limitation of the embedded clustering algorithm. We observe a drop (between time steps $17$ and $19$) from this temporarily high OSPA in PMBM filters when $K=8$, but the drop does not happen in the more challenging case of $K=20$. In the last few time steps in Fig. \ref{fig: highly heavy clutter}, the proposed VB-AbNHPP-RELO always has the lowest OSPA, demonstrating again the superiority of the proposed relocation strategy in the long-term tracking tasks.

\begin{figure}[tp!]
\centerline{\includegraphics[width=7cm]{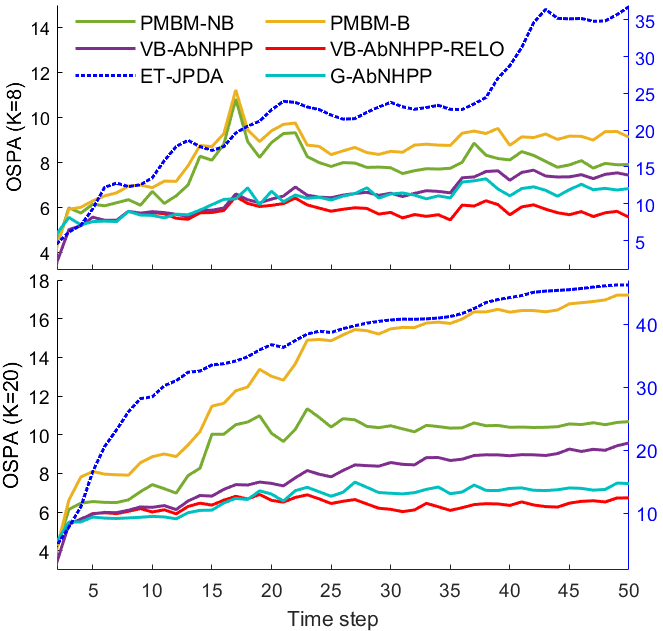}}
\caption{Mean OSPA metric over 50 time steps for tracking scenes in Section \ref{sec: result highly heavy clutter}. Blue dashed line is associated with the right y-axis, and all other lines are with the left y-axis. }
\label{fig:ospa heavy}
\end{figure}

\vspace{-1em}
\subsection{Tracking with rate estimation}
\label{sec: result rate estimation}
Recall that the proposed relocation strategy in Section \ref{sec: basic vatiational localisation} and the VB-AbNHPP-RELO in Section \ref{sec:full tracking with relocation} require a known object extent and object/clutter Poisson rates, such that missed objects can be effectively relocated with merely the measurements from a single time step. In this subsection, we give a simple demonstration of rate estimation with the VB-AbNHPP tracker in Algorithm \ref{Algo:tracker}; the object extent can be estimated in a similar fashion if needed. We consider a tracking and rate estimation task where $10$ objects move in different directions from initial positions around the origin point. A synthetic dataset is generated for $200$ time steps with the system models in Section \ref{sec: simulation}(a), and 10 objects' Poisson rates are randomly generated from $1.5$ to $10$. The ground truth of $10$ objects' trajectories over $200$ time steps is shown in Fig. \ref{fig:rate trajectories}, and the ground truth of each object's Poisson rate is shown in Fig. \ref{fig:rateall}. The clutter density is set to $10^{-5}$ per unit area, and the clutter rate is $5240$ for the considered tasks. The measurements from all time steps are shown as gray dots in Fig. \ref{fig:rate trajectories}, appearing as a grey background due to their excessively large number. 

\begin{figure}[tp!]
\centerline{\includegraphics[width=5cm]{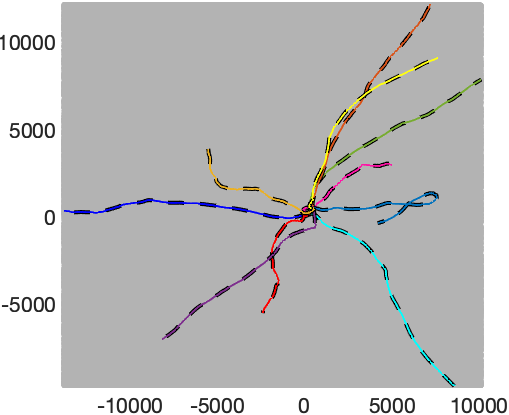}}
\caption{Estimated trajectories (colored lines) of VB-AbNHPP (Algorithm \ref{Algo:tracker}) for the tracking and rate estimation task in Section \ref{sec: result rate estimation}. The black dashed lines are the ground truth. }
\label{fig:rate trajectories}
\vspace{-1.5em}
\end{figure}

We use the VB-AbNHPP in Algorithm \ref{Algo:tracker} to simultaneously track objects and estimate object and clutter rates. In particular, a Gamma distribution with shape parameter $1$ and scale parameter $5$ is set as the initial Poisson rate prior for all $10$ objects and clutter. The $\gamma_n$ in \eqref{eq:rate predictive prior computation} is chosen as $\gamma_n=1-0.1\times(\text{max}\{1,n-10\})^{-0.9}$, i.e., $\gamma_n$ is fixed as $0.9$ for the first $11$ time steps, then strictly increases and approaches $1$ when $n$ is large. Such a construction of $\{\gamma\}_n$ corresponds to the monotonically increasing property instructed in Section \ref{sec:prior choice for para learn}. As with our previous experiments, the prior of initial object state is set as the ground truth, and $I=100, \epsilon=0.01$ are used to end the iterative updates. 

A successful tracking result of VB-AbNHPP is shown in Fig. \ref{fig:rate trajectories}, where the estimated trajectories are overlapped with the ground truth. The rate estimation results are shown in Fig. \ref{fig:rateall}. As time goes on, the converged variational distribution $q_n^*(\Lambda)$ has a smaller variance and its mean successfully converges to the ground truth rate for all objects and clutter. This agrees with our anticipation in Section \ref{sec:prior choice for para learn} that $q_n^*(\Lambda)$ would eventually converge to ground truth $\Lambda$ as a point estimator. Note that these accurate estimation results for different ground truths of rates are obtained with the identical initial prior and the decreasing sequence $\{\gamma\}_n$. In particular, the ground truth of clutter rate is far from the specified initial prior. This demonstrates the robustness of the rate estimation of VB-AbNHPP under an inaccurate initial prior. 

\begin{figure}[tp!]
\centerline{\includegraphics[width=8cm]{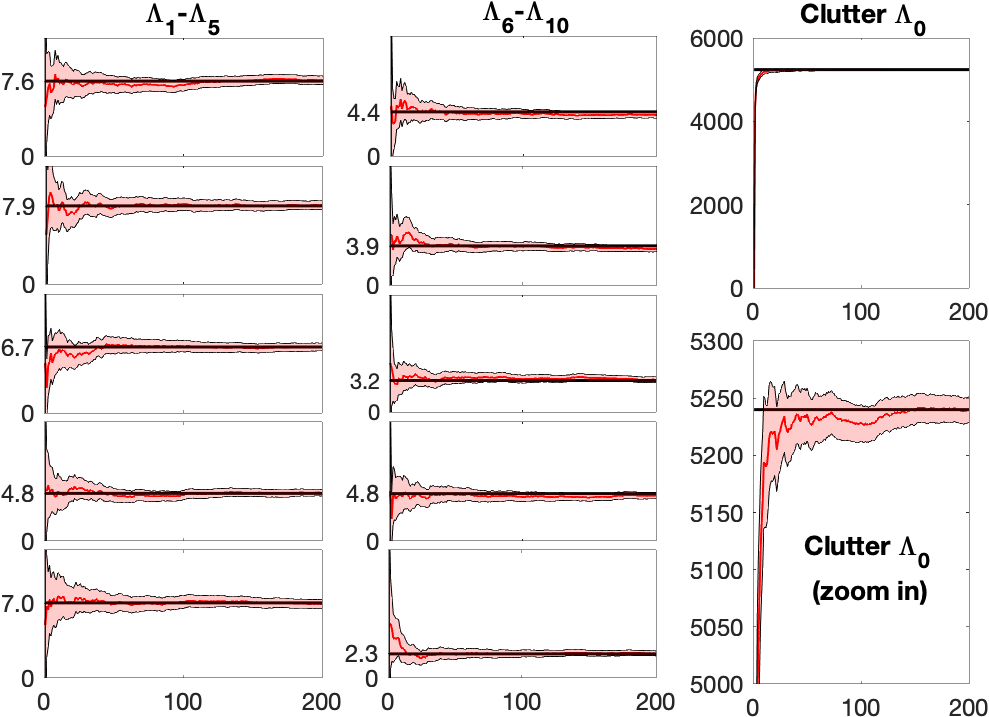}}
\caption{VB-AbNHPP (Algorithm \ref{Algo:tracker}) rate estimation over 200 time steps for the tracking and rate estimation task in Section \ref{sec: result rate estimation}. The ground truth of the Poisson rates of objects and clutter is shown by the black horizontal line in each subfigure. The red curve is the mean of $q_n^*(\Lambda)$, and the light red shade is the confidence region of $\pm2$ standard deviation of $q_n^*(\Lambda)$.
}
\label{fig:rateall}
\vspace{-1.5em}
\end{figure}

\section{Conclusion} \label{sec:conclusion}
In this paper, we extend the standard VB-AbNHPP tracker in our previous work \cite{gan2022variational} to a more adaptive VB-AbNHPP tracker that can simultaneously perform the tracking and parameter learning of the object and clutter Poisson rates, based on a general coordinate ascent variational filtering framework developed here. Another advanced extension is a VB-AbNHPP-RELO tracker that can robustly track closely-spaced objects in heavy clutter by automatically detecting and relocating the missed objects. This is based on the proposed novel variational localisation strategy, which can efficiently localise the object from a large surveillance area in heavy clutter. The developed trackers offer fast and parallelisable implementations with superior tracking and estimation performance, making them promising candidates in large scale object tracking scenarios. In particular, when tracking a large number of objects under heavy clutter, the VB-AbNHPP-RELO demonstrates significantly better tracking performance over existing trackers in terms of both accuracy and efficiency.

While we only derive the VB-AbNHPP tracker with rate estimation, other parameters such as measurement covariances/taregt extents can be learnt in a similar manner. The current dynamic model in the DBN in Section \ref{sec:PoissonModel} can also be replaced by a more complicated intent-driven model (e.g. \cite{gan2020modeling,gan2021levy}) or a leader-follower model \cite{li2021sequential}, such that the same framework can still be applied to infer the objects' destination/leader in clutter. Additionally, the variational tracker may be extended for more sophisticated NHPP measurement models that incorporate non-linear measurement functions, like those used in range-bearing measurements, and account for varying detection ranges across different sensors. Furthermore, the proposed variational localisation strategy may be extended to incorporate measurements from multiple time steps to offer a more reliable estimation. This localisation strategy can also facilitate joint detection and tracking task, and a more general VB-AbNHPP-RELO tracker that can handle other unknown variables, such as Poisson rates, measurement covariance, and the object number. This will be discussed in future work.

\appendices

\vspace{-0.5em}
\section{} 
\label{apx:remark for initialisation}
This appendix provides additional details about Remark \ref{remark: initial} in Section \ref{sec: single object vari loca strate}. Although difficult to analyse theoretically, Remark \ref{remark: initial} can be informally justified by looking into the iterative updates under a construction of the initialisation in \eqref{eq: Loca demo initial}. First, we describe some behaviours of these iterative updates, which will be used to illustrate Remark \ref{remark: initial} in detail.

With the initial $q^{(0)}(\theta_n)$ in \eqref{eq: Loca demo initial}, only measurements that lie in the local region, i.e. the $95\%$ confidence ellipse of $\mathcal{N}(m_s,C)$,
have a considerable probability to associate with the object. Conditional on this $q^{(0)}(\theta_n)$, it then produces a pseudo measurement $\overbar{Y}_n^{(0)}$ (evaluated according to \eqref{eq:pseudomeas Y}) that locates at the weighted average of those noticeable measurements with a covariance $\overbar{R}_n^{(0)}$ (i.e. evaluated according to \eqref{eq:pseudomeas R}).
For a local region that contains at least one measurement, the denominator in \eqref{eq:pseudomeas R} is considerable, leading to a much smaller covariance $\overbar{R}_n^{(0)}$ compared to the large positional covariance of the uncertain prior $p(X_{n,1})$.
Consequently, according to \eqref{eq: X update}, the position in the updated $q^{(1)}(X_{n,1})$ at the first iteration has a mean that is approximately equal to $\overbar{Y}_n^{(0)}$, and a covariance close to $\overbar{R}_n^{(0)}$.
Empirically, the next updated $q^{(2)}(X_{n,1})$ will typically cover the measurement nearest to $\overbar{Y}_n^{(0)}$. Then, \\
1) If there are dense measurements near or within the high confidence region of $q^{(2)}(X_{n,1})$: The subsequent iterative coordinate ascent updates will refine the $q(X_{n,1})$ at a smaller scale and converge to the $q^*(X_{n,1})$ that encompasses those dense measurements in the end. However, \\
2) If measurements are distributed more evenly or scattered around the high confidence region of $q^{(2)}(X_{n,1})$: The subsequent iterative coordinate ascent updates will enlarge the covariance of $q(X_{n,1})$ to cover as many of those measurements as possible, and finally converge to the $q^*(X_{n,1})$ that equals the uncertain prior  $p(X_{n,1})$, i.e. $q^*(X_{n,1})=p(X_{n,1})$.

Remark \ref{remark: initial} holds for a properly chosen $C$ since in this case, we expect $\overbar{Y}_n^{(0)}$, which is approximately the weighted average of the measurements covered by the $95\%$ confidence ellipse of $\mathcal{N}(m_s,C)$, is close to the location with the greatest density of measurements in this local area. Then the converged $q^*(X_{n,1})$ is very likely to capture this most probable object location. However, if $C$ is too high, Remark \ref{remark: initial} hardly stands. Since in this case, the $95\%$ confidence ellipse of $\mathcal{N}(m_s,C)$ ($s=1,2,...,N$) will include too many uniformly distributed clutter, resulting in their weighted average being close to the center of this local area, i.e., $\overbar{Y}_n^{(0)}\approx m_s$. Subsequently, depending on how the measurement(s) within the high confidence region of $q^{(2)}(X_{n,1})$ are distributed, the CAVI either converges to $q^*(X_{n,1})$ that encompasses the dense measurements nearest to $m_s$, or traps in the local optimum with $q^*(X_{n,1})=p(X_{n,1})$. 


\vspace{-0.5em}
\section{} \label{apx:parameter selection}
In this appendix, we discuss the choice of the track loss detection parameters $\tau_k,M_k^{los}$, relocation threshold $M_{k}^{reloc}$, and eligible initialisation threshold $M_{k}^{init}$ ($k=1,2,...,K$) for the developed full VB-AbNHPP-RELO tracker. We assume the object and clutter rate $\Lambda$ and measurement covariance $R$ are known to us. We will show that $\tau_k,M_k^{los}$, $M_{k}^{reloc},M_{k}^{init}$ can be automatically selected by specifying probabilities $P_{thres}^{los}\in(0,1)$ and $P_{thres}^{reloc}\in(0,1)$ (both defined below) respectively.

\vspace{-0.7em}
\subsection{Guide for the choice of $\tau_k$ and $M^{los}_k$} \label{apx:choice of tau and M}
When the event $E$ in \eqref{eq:lost track event} is used as the track loss criterion, the higher the $M_k^{los}$, the less likely that event $E$ will happen, and the more sensitive the algorithm is to the potential track loss. A lower $\tau_k$ leads to an event $E$ that considers the $\hat{M}_{n,k}$ from more recent time steps, and hence better reflecting the timely tracking performance.
However, if $\tau_k$ and $\Lambda_k$ are both very low (e.g. $\tau_k=1$ and $\Lambda_k=1$), it is very likely that no measurement is generated from the object $k$ during these $\tau_k$ time steps, which makes the event $E$ likely to occur even for a successful tracker. Therefore, a trade-off has to be made when choosing $\tau_k$ and $M^{los}_k$.

This subsection describes a strategy to automatically select suitable $\tau_k$ and $M^{los}_k$ such that the event $E$ in \eqref{eq:lost track event} rarely happens for a successful tracker.
In particular, we define $P_{thres}^{los}\in(0,1)$ as the probability of the rare event $E$ occurring, given that object $k$ has been successfully tracked. Subsequently, the rationale behind our track loss criterion can be explained with the idea of hypothesis testing. Specifically, if such a rare event $E$ (in terms of a successful tracker) occurs, it implies that the object may not be properly tracked. Furthermore, since the event $E$ is more likely to happen when the algorithm loses track of the object (as discussed in Section \ref{sec: lost track detect}), we conclude that the track of the object $k$ has been lost.

To select $\tau_k$ and $M^{los}_k$, we first specify the value of $P_{thres}^{los}$. Then for each $k=1,2,...,K$, $\tau_k$ and $M^{los}_k$ are evaluated based on the specified $P_{thres}^{los}$ according to three steps:
\begin{enumerate}[Step 1:]
\item Set $\tau_k=\ceil*{\frac{1}{\Lambda_k}\log\frac{1}{P^{thres}}}$, where $\ceil{\cdot}$ is the ceil function.
\item Obtain a continuous increasing function $\Tilde{F}_{\tau_k\Lambda_k}: [0,\infty)\rightarrow[F_{\tau_k\Lambda_k}(0),\infty)$ by interpolating (e.g. Spline interpolation) $(x,\Tilde{F}_{\tau_k\Lambda_k}(x))$ from the following series of coordinates: $(0,F_{\tau_k\Lambda_k}(0)),$
$(1,F_{\tau_k\Lambda_k}(1)),$
$(2,F_{\tau_k\Lambda_k}(2)),...$, where $F_{\tau_k\Lambda_k}$ is the exact Poisson cumulative distribution function (CDF) with rate parameter $\tau_k\Lambda_k$.
\item Find the $M^{los}_k$ such that $\Tilde{F}_{\tau_k\Lambda_k}(M^{los}_k)=P_{thres}^{los}$.
\end{enumerate}


We now explain the rationale behind these steps. Note that the exact CDF for the decimal variable $\sum_{t=n-\tau_k+1}^n \hat{M}_{t,k}$ (i.e. the estimated cumulative measurement numbers in \eqref{eq:lost track event}) is intractable. However, according to the additive property of Poisson distribution, the integer variable $\sum_{t=n-\tau_k+1}^n M_{t,k}$ (i.e. the exact cumulative measurement numbers) has a Poisson CDF $F_{\tau_k\Lambda_k}$. With the intuition that a successful tracker should not produce an estimate $\hat{M}_{t,k}$ that greatly deviates from the exact $M_{t,k}$, we use $\tilde{F}_{\tau_k\Lambda_k}$, i.e. a continuous increasing approximation of $F_{\tau_k\Lambda_k}$ constructed in step 2, as an approximate CDF for the estimate $\sum_{t=n-\tau_k+1}^n \hat{M}_{t,k}$ generated from a successful tracker. A noteworthy property of $\tilde{F}_{\tau_k\Lambda_k}$ is that it is continuous increasing, which makes it better to describe the decimal variable $\sum_{t=n-\tau_k+1}^n \hat{M}_{t,k}$ than the step-wise function $F_{\tau_k\Lambda_k}$.

Subsequently, the probability of the event $E$ in \eqref{eq:lost track event} conditional on the object $k$ having been successfully tracked can be approximated by $\tilde{F}_{\tau_k\Lambda_k}(M_k^{los})$, i.e., $P_{thres}^{los}\approx\tilde{F}_{\tau_k\Lambda_k}(M_k^{los})$. Therefore, once $\tau_k$ and $P_{thres}^{los}$ are specified, we can find $M_k^{los}$ by Step 3. Note that the uniqueness of $M_k^{los}$ found in Step 3 is guaranteed by the increasing property of $\tilde{F}_{\tau_k\Lambda_k}$, and the existence of $M_k^{los}$ is guaranteed by the continuous property of $\tilde{F}_{\tau_k\Lambda_k}$ and the fact that $\tilde{F}_{\tau_k\Lambda_k}(0)=F_{\tau_k\Lambda_k}(0)\leq P^{los}_{thres}$. This inequality is always satisfied owing to $\tau_k$ constructed in Step 1. Specifically, $\tau_k$ in Step 1 is essentially the minimal integer that satisfies $F_{\tau_k\Lambda_k}(0)\leq P_{thres}^{los}$. Such a small $\tau_k$ ensures our track loss criterion $E$ reflects timely tracking performance.

\vspace{-1.2em}
\subsection{Guide for the choice of $M_{k}^{reloc}$ and $M_{k}^{init}$} \label{apx:choice of Mreloc and Minit}
It is practically useful to choose $M_{k}^{reloc}$ based on a specified probability $P_{thres}^{reloc}\in(0,1)$, where $P_{thres}^{reloc}$ is the probability of the effective relocation criterion in \eqref{eq: effective relocation criterion} being satisfied, conditional on a successful relocation. On the one hand, we typically require $P_{thres}^{reloc}>0.2$ to ensure the effective relocation criterion is easy to meet, otherwise the algorithm may take too many steps to find a convincing position for missed object. On the other hand, we also need to ensure $P_{thres}^{reloc}$ is not too high, since a high $P_{thres}^{reloc}$ often leads to frequent relocations and bears a high chance of a false relocation that may deteriorate the tracking performance.

The relocation threshold $M_{k}^{reloc}$ can be selected according to $P_{thres}^{reloc}$ in a similar fashion as in Appendix \ref{apx:choice of tau and M}. Specifically, we first use a continuous increasing function $\Tilde{F}_{\Lambda_k}: [0,\infty)\rightarrow[F_{\Lambda_k}(0),\infty)$ to approximate the CDF for the variable $\sum_{j=0}^{M_n} q^w_{n}(\theta_{n,j}=k)$ in \eqref{eq: effective relocation criterion} conditional on a successful relocation with the $s$-th initialisation. This approximate CDF $\Tilde{F}_{\Lambda_k}$ can be obtained by carrying out Step 2 described in Appendix \ref{apx:choice of tau and M} with $\tau_k$ being $1$. Then find $M_{k}^{reloc}$ such that $\Tilde{F}_{\Lambda_k}(M_{k}^{reloc})=1-P_{thres}^{reloc}$. 


Finally, we discuss the choice of the eligible initialisation threshold $M_{k}^{init}$, which is employed in Algorithm \ref{Algo:relocation} to skip unnecessary initiasations to accelerate the relocation procedure. This eligible initialisation threshold $M_{k}^{init}$ can simply be set as $M_{k}^{reloc}-1$ or $M_{k}^{reloc}-2$. Such a design of $M_{k}^{init}$ satisfies $M_{k}^{init}<M_{k}^{reloc}$. Therefore the algorithm never skip exploring any local region that encompasses enough measurements to meet the effective relocation criterion in \eqref{eq: effective relocation criterion}. Moreover, by setting a lower $M_{k}^{init}$ (e.g., $M_{k}^{init}=M_{k}^{reloc}-2$), the algorithm also explores local regions that include fewer measurements, which sometimes is necessary. For example, when the true missed object is at the edge of a local region, only a few object-oriented measurements may lie in this local region. In this case, there may be less than $M_{k}^{init}$ measurements in this local region. However, with the corresponding initialisation, it is possible for CAVI to localise the missed object and the effective relocation criterion is still met. 
\vspace{-1em}

\bibliographystyle{IEEEtran}
\bibliography{./IEEEabrv,IEEEexample}

\begin{thebibliography}{10}
\providecommand{\url}[1]{#1}
\csname url@samestyle\endcsname
\providecommand{\newblock}{\relax}
\providecommand{\bibinfo}[2]{#2}
\providecommand{\BIBentrySTDinterwordspacing}{\spaceskip=0pt\relax}
\providecommand{\BIBentryALTinterwordstretchfactor}{4}
\providecommand{\BIBentryALTinterwordspacing}{\spaceskip=\fontdimen2\font plus
\BIBentryALTinterwordstretchfactor\fontdimen3\font minus \fontdimen4\font\relax}
\providecommand{\BIBforeignlanguage}[2]{{%
\expandafter\ifx\csname l@#1\endcsname\relax
\typeout{** WARNING: IEEEtran.bst: No hyphenation pattern has been}%
\typeout{** loaded for the language `#1'. Using the pattern for}%
\typeout{** the default language instead.}%
\else
\language=\csname l@#1\endcsname
\fi
#2}}
\providecommand{\BIBdecl}{\relax}
\BIBdecl

\bibitem{bar1995multitarget}
Y.~Bar-Shalom and X.-R. Li, \emph{Multitarget-multisensor tracking: principles and techniques}.\hskip 1em plus 0.5em minus 0.4em\relax YBs Storrs, CT, 1995, vol.~19.

\bibitem{granstrom2019poisson}
K.~Granstr{\"o}m, M.~Fatemi, and L.~Svensson, ``Poisson multi-{B}ernoulli mixture conjugate prior for multiple extended target filtering,'' \emph{IEEE Transactions on Aerospace and Electronic Systems}, vol.~56, no.~1, pp. 208--225, 2019.

\bibitem{granstrom2017likelihood}
K.~Granstr{\"o}m, L.~Svensson, S.~Reuter, Y.~Xia, and M.~Fatemi, ``Likelihood-based data association for extended object tracking using sampling methods,'' \emph{IEEE Transactions on intelligent vehicles}, vol.~3, no.~1, pp. 30--45, 2017.

\bibitem{meyer2018message}
F.~Meyer, T.~Kropfreiter, J.~L. Williams, R.~Lau, F.~Hlawatsch, P.~Braca, and M.~Z. Win, ``Message passing algorithms for scalable multitarget tracking,'' \emph{Proceedings of the IEEE}, vol. 106, no.~2, pp. 221--259, 2018.

\bibitem{meyer2021scalable}
F.~Meyer and J.~L. Williams, ``Scalable detection and tracking of geometric extended objects,'' \emph{IEEE Transactions on Signal Processing}, vol.~69, pp. 6283--6298, 2021.

\bibitem{gilholm2005poisson}
K.~Gilholm, S.~Godsill, S.~Maskell, and D.~Salmond, ``Poisson models for extended target and group tracking,'' in \emph{Signal and Data Processing of Small Targets 2005}, vol. 5913.\hskip 1em plus 0.5em minus 0.4em\relax International Society for Optics and Photonics, 2005, p. 59130R.

\bibitem{yang2018linear}
S.~Yang, K.~Thormann, and M.~Baum, ``Linear-time joint probabilistic data association for multiple extended object tracking,'' in \emph{2018 IEEE 10th Sensor Array and Multichannel Signal Processing Workshop (SAM)}.\hskip 1em plus 0.5em minus 0.4em\relax IEEE, 2018, pp. 6--10.

\bibitem{gan2022variational}
R.~Gan, Q.~Li, and S.~Godsill, ``A variational {B}ayes association-based multi-object tracker under the non-homogeneous {P}oisson measurement process,'' in \emph{2022 25th International Conference on Information Fusion (FUSION)}.\hskip 1em plus 0.5em minus 0.4em\relax IEEE, 2022, pp. 1--8.

\bibitem{li2022scalable}
Q.~Li, J.~Liang, and S.~Godsill, ``Scalable data association and multi-target tracking under a {P}oisson mixture measurement process,'' in \emph{ICASSP 2022-2022 IEEE International Conference on Acoustics, Speech and Signal Processing (ICASSP)}.\hskip 1em plus 0.5em minus 0.4em\relax IEEE, 2022, pp. 5503--5507.

\bibitem{li2023adaptive}
Q.~Li, R.~Gan, J.~Liang, and S.~J. Godsill, ``An adaptive and scalable multi-object tracker based on the non-homogeneous {P}oisson process,'' \emph{IEEE Transactions on Signal Processing}, 2023.

\bibitem{mahler2003multitarget}
R.~Mahler, ``Multitarget {B}ayes filtering via first-order multitarget moments,'' \emph{IEEE Transactions on Aerospace and Electronic systems}, vol.~39, no.~4, pp. 1152--1178, 2003.

\bibitem{bishop:2006:PRML}
C.~M. Bishop, \emph{Pattern Recognition and Machine Learning}.\hskip 1em plus 0.5em minus 0.4em\relax Springer, 2006.

\bibitem{turner2014complete}
R.~D. Turner, S.~Bottone, and B.~Avasarala, ``A complete variational tracker,'' \emph{Advances in Neural Information Processing Systems}, vol.~27, 2014.

\bibitem{lau2016structured}
R.~A. Lau and J.~L. Williams, ``A structured mean field approach for existence-based multiple target tracking,'' in \emph{2016 19th International Conference on Information Fusion (FUSION)}.\hskip 1em plus 0.5em minus 0.4em\relax IEEE, 2016, pp. 1111--1118.

\bibitem{blei2017variational}
D.~M. Blei, A.~Kucukelbir, and J.~D. McAuliffe, ``Variational inference: A review for statisticians,'' \emph{Journal of the American statistical Association}, vol. 112, no. 518, pp. 859--877, 2017.

\bibitem{4526445}
S.~K. Pang, J.~Li, and S.~J. Godsill, ``Models and algorithms for detection and tracking of coordinated groups,'' in \emph{2008 IEEE Aerospace Conference}, 2008, pp. 1--17.

\bibitem{li2023scalable}
Q.~Li, R.~Gan, and S.~Godsill, ``A scalable {Rao-Blackwellised Sequential MCMC} sampler for joint detection and tracking in clutter,'' in \emph{2023 26th International Conference on Information Fusion (FUSION)}.\hskip 1em plus 0.5em minus 0.4em\relax IEEE, 2023, pp. 1--8.

\bibitem{sarkka2013non}
S.~S{\"a}rkk{\"a} and J.~Hartikainen, ``Non-linear noise adaptive {K}alman filtering via variational {B}ayes,'' in \emph{2013 IEEE International Workshop on Machine Learning for Signal Processing (MLSP)}.\hskip 1em plus 0.5em minus 0.4em\relax IEEE, 2013, pp. 1--6.

\bibitem{huang2017novel}
Y.~Huang, Y.~Zhang, Z.~Wu, N.~Li, and J.~Chambers, ``A novel adaptive kalman filter with inaccurate process and measurement noise covariance matrices,'' \emph{IEEE transactions on Automatic Control}, vol.~63, no.~2, pp. 594--601, 2017.

\bibitem{vsmidl2006variational}
V.~{\v{S}}m{\'\i}dl and A.~Quinn, \emph{The variational Bayes method in signal processing}.\hskip 1em plus 0.5em minus 0.4em\relax Springer Science \& Business Media, 2006.

\bibitem{vermaak2003variational}
J.~Vermaak, N.~D. Lawrence, and P.~Perez, ``Variational inference for visual tracking,'' in \emph{2003 IEEE Computer Society Conference on Computer Vision and Pattern Recognition, 2003. Proceedings.}, vol.~1.\hskip 1em plus 0.5em minus 0.4em\relax IEEE, 2003, pp. I--I.

\bibitem{sarkka2013bayesian}
S.~S{\"a}rkk{\"a}, \emph{Bayesian filtering and smoothing}.\hskip 1em plus 0.5em minus 0.4em\relax Cambridge university press, 2013, no.~3.

\bibitem{ye2020monte}
L.~Ye, A.~Beskos, M.~De~Iorio, and J.~Hao, ``{Monte Carlo} co-ordinate ascent variational inference,'' \emph{Statistics and Computing}, pp. 1--19, 2020.

\bibitem{Gan2022}
R.~Gan and S.~Godsill, ``Conditionally factorized variational {B}ayes with importance sampling,'' in \emph{ICASSP 2022-2022 IEEE International Conference on Acoustics, Speech and Signal Processing (ICASSP)}.\hskip 1em plus 0.5em minus 0.4em\relax IEEE, 2022, pp. 5777--5781.

\bibitem{beal2003}
M.~J. Beal, ``Variational algorithms for approximate {B}ayesian inference,'' Ph.D. dissertation, Gatsby Computational Neuroscience Unit, University College London, 2003.

\bibitem{meyer2020scalable}
F.~Meyer and M.~Z. Win, ``Scalable data association for extended object tracking,'' \emph{IEEE Transactions on Signal and Information Processing over Networks}, vol.~6, pp. 491--507, 2020.

\bibitem{schuhmacher2008consistent}
D.~Schuhmacher, B.-T. Vo, and B.-N. Vo, ``A consistent metric for performance evaluation of multi-object filters,'' \emph{IEEE transactions on signal processing}, vol.~56, no.~8, pp. 3447--3457, 2008.

\bibitem{gan2020modeling}
R.~Gan, J.~Liang, B.~I. Ahmad, and S.~Godsill, ``Modeling intent and destination prediction within a {B}ayesian framework: Predictive touch as a usecase,'' \emph{Data-Centric Engineering}, vol.~1, 2020.

\bibitem{gan2021levy}
R.~Gan, B.~I. Ahmad, and S.~J. Godsill, ``L{\'e}vy state-space models for tracking and intent prediction of highly maneuverable objects,'' \emph{IEEE Transactions on Aerospace and Electronic Systems}, vol.~57, no.~4, 2021.

\bibitem{li2021sequential}
Q.~Li, B.~I. Ahmad, and S.~J. Godsill, ``Sequential dynamic leadership inference using {Bayesian Monte C}arlo methods,'' \emph{IEEE Transactions on Aerospace and Electronic Systems}, vol.~57, no.~4, pp. 2039--2052, 2021.

\end{thebibliography}
\onecolumn
\begin{center}
{
\LARGE\bfseries Variational Tracking and Redetection for Closely-spaced Objects in Heavy Clutter: \  Supplementary Materials
}
\end{center}

\vspace{1cm}

\section{}\label{apx:association update deriv}
This appendix derives the association update in \eqref{eq:update theta}. First, using definition in \eqref{eq:single assoc prior}, we have
\begin{align} \notag
    \log \left(p(\theta_{n,j}|\Lambda) \ell(Y_{n,j}|X_{n,\theta_{n,j}})\right)=&\log \frac{\sum_{k=0}^K\Lambda_k\delta[\theta_{n,j}=k]}{\Lambda_{\text{sum}}}\ell(Y_{n,j}|X_{n,\theta_{n,j}})=\log \frac{\sum_{k=0}^K\Lambda_k\delta[\theta_{n,j}=k]\ell(Y_{n,j}|X_{n,k})}{\Lambda_{\text{sum}}}\\ \notag
    =& -\log \Lambda_{\text{sum}} +\log \left(\sum_{k=0}^K\Lambda_k\ell(Y_{n,j}|X_{n,k})\delta[\theta_{n,j}=k]\right)\\ \notag
    =& -\log \Lambda_{\text{sum}} +\sum_{k=0}^K \delta[\theta_{n,j}=k]\log \left(\Lambda_k\ell(Y_{n,j}|X_{n,k})\right),\\ \notag
    =&-\log \Lambda_{\text{sum}} +\left(\log \frac{\Lambda_0}{V}\right)\delta[\theta_{n,j}=0]+\sum_{k=1}^K\left(\log \Lambda_k\mathcal{N}(Y_{n,j};HX_{n,k},R_k)\right)\delta[\theta_{n,j}=k]
\end{align}
where the second last line follows from the fact that only one of $\delta[\theta_{n,j}=k]$ for $k=0,1,...,K$ equals $1$, with the rest being zero. The final line results from substituting \eqref{measurement model}. Now, recalling that $q_n(\Lambda_k)=\mathcal{G}(\Lambda_k;\eta_{n|n}^k,\rho_{n|n}^k)$, taking the expectation yields 
\begin{align} \notag
    &\E_{q_n(\Lambda)q_n(X_n)}\log \left(p(\theta_{n,j}|\Lambda) \ell(Y_{n,j}|X_{n,\theta_{n,j}})\right)\\ \notag
    =&-\log \Lambda_{\text{sum}} +\left(\E_{q_n(\Lambda_0)}\log \frac{\Lambda_0}{V}\right)\delta[\theta_{n,j}=0]+\sum_{k=1}^K\left(\E_{q_n(\Lambda_k)}\log \Lambda_k+\E_{q_n(X_{n,k})}\log\mathcal{N}(Y_{n,j};HX_{n,k},R_k)\right)\delta[\theta_{n,j}=k]\\ \label{eq: exp assoclikeli}
    =&-\log \Lambda_{\text{sum}} +\left(\psi(\eta_{n|n}^0)+\log(\rho_{n|n}^0)+\log\frac{1}{V}\right)\delta[\theta_{n,j}=0]\\ \notag
    &+\sum_{k=1}^K\left(\psi(\eta_{n|n}^k)+\log(\rho_{n|n}^k)  - \frac{1}{2}\E_{q_n(X_{n,k})}(Y_{n,j}-HX_{n,k})^\top R_k^{-1}(Y_{n,j}-HX_{n,k})+\log \frac{1}{\sqrt{(2\pi)^D{\det R_k}}}\right)\delta[\theta_{n,j}=k],
\end{align}
where $\psi(\cdot)$ is the digamma function. Note that we have the following result by using the formula for the expectation of a quadratic form and recalling that $q_n(X_{n,k})=\mathcal{N}(X_{n,k};\mu_{n|n}^k,\Sigma_{n|n}^k)$:
\begin{align} \notag
    &- \frac{1}{2}\E_{q_n(X_{n,k})}(Y_{n,j}-HX_{n,k})^\top R_k^{-1}(Y_{n,j}-HX_{n,k})+\log \frac{1}{\sqrt{(2\pi)^D{\det R_k}}}\\ \notag
    =&- \frac{1}{2}\left((Y_{n,j}-H\mu_{n|n}^k)^\top R_k^{-1}(Y_{n,j}-H\mu_{n|n}^k)+\Tr(R_k^{-1}H\Sigma_{n|n}^kH^\top) \right) +\log \frac{1}{\sqrt{(2\pi)^D{\det R_k}}}\\ \label{eq: exp log normal}
    =& \log\mathcal{N}(Y_{n,j};H\mu_{n|n}^k,R_k)-\frac{1}{2}\Tr(R_k^{-1}H\Sigma_{n|n}^kH^\top).
\end{align}
Now, substituting \eqref{eq: exp log normal} back to \eqref{eq: exp assoclikeli} yields
\begin{align} \notag
    &\E_{q_n(\Lambda)q_n(X_n)}\log \left(p(\theta_{n,j}|\Lambda) \ell(Y_{n,j}|X_{n,\theta_{n,j}})\right)\\ \notag
    =&-\log \Lambda_{\text{sum}} +\left(\psi(\eta_{n|n}^0)+\log(\rho_{n|n}^0)+\log\frac{1}{V}\right)\delta[\theta_{n,j}=0]\\ \notag
    &+\sum_{k=1}^K\left(\psi(\eta_{n|n}^k)+\log(\rho_{n|n}^k) + \log\mathcal{N}(Y_{n,j};H\mu_{n|n}^k,R_k)-\frac{1}{2}\Tr(R_k^{-1}H\Sigma_{n|n}^kH^\top)\right)\delta[\theta_{n,j}=k].
\end{align}
Finally, by reapplying the exponential and utilising the rule $\exp(\sum_{k=0}^Ka_k\delta[\theta_{n,j}=k])=\sum_{k=0}^K\exp(a_k)\delta[\theta_{n,j}=k]$, which is derived from the fact that only one $\delta[\theta_{n,j}=k]$ from $k=0,1,...,K$ equals $1$ with all others being zero, we arrive at the following result:
\begin{align} \notag
    &\exp\left(\E_{q_n(\Lambda)q_n(X_n)}\log \left(p(\theta_{n,j}|\Lambda) \ell(Y_{n,j}|X_{n,\theta_{n,j}})\right)\right)\\ \notag
    =&\frac{1}{\Lambda_{\text{sum}}}\left(\frac{\rho_{n|n}^0\exp(\psi(\eta_{n|n}^0))}{V}\delta[\theta_{n,j}=0]+\sum_{k=1}^K \rho_{n|n}^k\exp(\psi(\eta_{n|n}^k))\mathcal{N}(Y_{n,j};H\mu_{n|n}^k,R_k)\exp(-\frac{1}{2}\Tr(R_k^{-1}H\Sigma_{n|n}^kH^\top))\delta[\theta_{n,j}=k]\right)\\
    \propto & \frac{\overbar{\Lambda}_0}{V}\delta[\theta_{n,j}=0]+\sum_{k=1}^K\overbar{\Lambda}_k l_k\delta[\theta_{n,j}=k],
\end{align}
with $\overbar{\Lambda}_k$ and $l_k$ defined in \eqref{eq:update theta}. This completes the derivation of \eqref{eq:update theta}.
\section{}\label{apx:initialisation law}
Here we provide a detailed derivation of the initial variational distribution $\hat{p}_n(\theta_{n,j}|Y_{n,j},\Lambda=\hat{\Lambda})$ in \eqref{eq:init individul theta}. According to the definition of $\hat{p}_n$ in \eqref{eq: prob law}, we have 
\begin{align} \label{eq: complete prob law}
    \hat{p}_n(X_n,\theta_n,\Lambda,Y_n)=& p(Y_n|\theta_n,X_n)p(\theta_n|M_n,\Lambda)\hat{p}_{n|n-1}(X_n) p(M_n|\Lambda)\hat{p}_{n|n-1}(\Lambda)/c,
\end{align}
where $c$ is a normalisation constant. Then by marginalising $Y_n,X_n,\theta_n$ out of \eqref{eq: complete prob law}, we have
\begin{align} \label{eq: complete lambda}
    \hat{p}_n(\Lambda)=\sum_{\theta_n}\int \int\hat{p}_n(X_n,\theta_n,\Lambda,Y_n) dY_ndX_n = p(M_n|\Lambda)\hat{p}_{n|n-1}(\Lambda)/c,
\end{align}
where $c$ is the same constant in \eqref{eq: complete prob law}. Divide \eqref{eq: complete prob law} by \eqref{eq: complete lambda}:
\begin{align} \notag
\hat{p}_n(X_n,&\theta_n,Y_n|\Lambda)=p(Y_n|\theta_n,X_n)p(\theta_n|M_n,\Lambda)\hat{p}_{n|n-1}(X_n)\\ \label{eq: conditional law}
&=\hat{p}_{n|n-1}(X_n)\prod_{j=1}^{M_n} \ell(Y_{n,j}|X_{n,\theta_{n,j}})p(\theta_{n,j}|\Lambda),
\end{align}
where the last equality follows from \eqref{eq:obs prior} and \eqref{eq:assoc prior}. Denote $Y_{n,j-}$ as all measurements in $Y_n$ except for the $Y_{n,j}$, and $\theta_{n,j-}$ is defined similarly. For any $j=1,2,...,M_n$, by marginalising $Y_{n,j-},X_n,\theta_{n,j-}$ out of \eqref{eq: conditional law}:
\begin{align} \notag  \hat{p}&_n(\theta_{n,j},Y_{n,j}|\Lambda)=\sum_{\theta_{n,j-}}\int\int\hat{p}_n(X_n,\theta_n,Y_n|\Lambda)dX_ndY_{n,j-}\\ \notag
=& \int \hat{p}_{n|n-1}(X_n)\ell(Y_{n,j}|X_{n,\theta_{n,j}})dX_n \ p(\theta_{n,j}|\Lambda)\prod_{i=1:M_n,i\neq j} \int \ell(Y_{n,i}|X_{n,\theta_{n,i}})dY_{n,i}\sum_{\theta_{n,i}}p(\theta_{n,i}|\Lambda)\\ \notag
    =&p(\theta_{n,j}|\Lambda)\int \hat{p}_{n|n-1}(X_{n})\ell(Y_{n,j}|X_{n,\theta_{n,j}})dX_{n}\\ \notag
    =&p(\theta_{n,j}|\Lambda)\biggl(\frac{1}{V}\delta[\theta_{n,j}=0]+\sum_{k=1}^K l_k^{0}\delta[\theta_{n,j}=k]\biggr)\\ \label{eq: conidtional theta_j y_j}
    =&\frac{1}{\Lambda_{\text{sum}}}\biggl(\frac{\Lambda_0}{V}\delta[\theta_{n,j}=0]+\sum_{k=1}^K \Lambda_kl_k^{0}\delta[\theta_{n,j}=k]\biggr).
\end{align}
where $l_k^0$ is defined in \eqref{eq:init individul theta}. Note that the integral in the third last line in \eqref{eq: conidtional theta_j y_j} is either a Gaussian marginal likelihood $l_k^0$ if $\theta_{n,j}\neq0$, or the $\ell(Y_{n,j}|X_{n,0})$ in \eqref{measurement model} otherwise (recall that $X_n$ does not include $X_{n,0}$ as assumed Section \ref{sec:PoissonModel}). The last line in \eqref{eq: conidtional theta_j y_j} follows from \eqref{eq:single assoc prior} and the nature of Kronecker delta function. Finally, \eqref{eq:init individul theta} is obtained by omitting constant $\Lambda_{\text{sum}}$ in \eqref{eq: conidtional theta_j y_j} and substituting $\hat{\Lambda}_k=\E_{q_{n-1}^*(\Lambda)}\Lambda_k$ to each $\Lambda_k$.
\section{} \label{apx: elbo derivation}
This appendix derives the ELBO in \eqref{eq: efficient ELBO}. First, we introduce the following Lemma for calculating the summation of quadratic forms:
\begin{lemma} \label{lemma: quadratric sum lemma}
For symmetric and invertible matrix $C_i\in \mathbb{R}^{D\times D}$, and vectors $x,m_i\in \mathbb{R}^{D\times1}$ ($i=1,2,...,N$), we have 
\begin{align} \label{eq:quadratric sum theorem}
\begin{aligned}
    \sum_{i=1}^N-\frac{1}{2}(x-m_i)^\top& C_i^{-1}(x-m_i)=-\frac{1}{2}(x-\mu)^\top \Sigma^{-1}(x-\mu)+\frac{1}{2}\mu^\top\Sigma^{-1}\mu-\frac{1}{2}\sum_{i=1}^Nm_i^\top C_i^{-1} m_i,\\
    \Sigma&=\left(\sum_{i=1}^N C_i^{-1}\right)^{-1}, \ \ \ \ \ \ \mu=\left(\sum_{i=1}^N C_i^{-1}\right)^{-1}\sum_{i=1}^N C_i^{-1} m_i.
\end{aligned}
\end{align}
A special case often encountered
is that the $N$ symmetric matrices $C_i$ differ only in scale, i.e., we have $C_i=C/\omega_i$ for $i=1,2,...,N$. In this case, it is straightforward to derive the following from \eqref{eq:quadratric sum theorem}:
\begin{align} \label{eq:quadratric sum theorem special case}
\begin{aligned}
     \sum_{i=1}^N-\frac{\omega_i}{2}(x-m_i)^\top C^{-1}(x-m_i)&=-\frac{1}{2}(x-\mu)^\top \Sigma^{-1}(x-\mu)+\frac{1}{2}\mu^\top\Sigma^{-1}\mu-\frac{1}{2}\sum_{i=1}^N\omega_im_i^\top C_i^{-1} m_i,\\
    \Sigma&=\frac{C}{\sum_{i=1}^N \omega_i}, \ \ \ \ \ \ \mu=\frac{\sum_{i=1}^N \omega_im_i}{\sum_{i=1}^N \omega_i}.
\end{aligned}
\end{align}
\end{lemma}

\begin{proof}
First note that a single quadratic form $(x-\mu)^\top \Sigma^{-1}(x-\mu)$ can always be rewritten as 
\begin{align} \label{eq: single quadratric form}
    (x-\mu)^\top \Sigma^{-1}(x-\mu)=x^\top\Sigma^{-1} x -x^\top \Sigma^{-1}\mu-(\Sigma^{-1}\mu)^\top x + \mu^\top \Sigma^{-1}\mu.
\end{align}
Then 
\begin{align} \notag
    &\sum_{i=1}^N-\frac{1}{2}(x-m_i)^\top C_i^{-1}(x-m_i)=\sum_{i=1}^N-\frac{1}{2}\left[x^\top C_i^{-1} x - x^\top C_i^{-1}m_i - \left( C_i^{-1}m_i\right)^\top x + m_i^\top C_i^{-1} m_i\right]\\ \label{eq: sum quadra first part}
    =&-\frac{1}{2}\left[x^\top \sum_{i=1}^N C_i^{-1}x  -x^\top\sum_{i=1}^N C_i^{-1} m_i  - \left(\sum_{i=1}^N  C_i^{-1}m_i \right)^\top x\right]-\frac{1}{2}\sum_{i=1}^Nm_i^\top C_i^{-1} m_i
\end{align}
Comparing the terms in the bracket with \eqref{eq: single quadratric form}, we can find that it matches the quadratic form $-\frac{1}{2}(x-\mu)^\top \Sigma^{-1}(x-\mu)$ in \eqref{eq: single quadratric form} if 
\begin{align} \label{eq:summ mean var convenient}
    \Sigma^{-1}=\sum_{i=1}^N C_i^{-1}, \ \ \ \Sigma^{-1}\mu=\sum_{i=1}^N C_i^{-1} m_i,
\end{align}
which suggests that $\Sigma,\mu$ being
\begin{align} \label{eq:summation mean var}
    \Sigma=\left(\sum_{i=1}^N C_i^{-1}\right)^{-1}, \ \ \ \ \ \mu=\left(\sum_{i=1}^N C_i^{-1}\right)^{-1}\sum_{i=1}^N C_i^{-1} m_i.
\end{align}
Now substitute \eqref{eq:summation mean var}, or equivalently \eqref{eq:summ mean var convenient}, back to \eqref{eq: sum quadra first part}, we have
\begin{align} \notag
    &\sum_{i=1}^N-\frac{1}{2}(x-m_i)^\top C_i^{-1}(x-m_i)=-\frac{1}{2}\left[x^\top \sum_{i=1}^N C_i^{-1}x  -x^\top\sum_{i=1}^N C_i^{-1} m_i  - \left(\sum_{i=1}^N  C_i^{-1}m_i \right)^\top x\right]-\frac{1}{2}\sum_{i=1}^Nm_i^\top C_i^{-1} m_i\\ \notag
    =&-\frac{1}{2}\left[x^\top \Sigma^{-1}x  -x^\top\Sigma^{-1}\mu  - \left(\Sigma^{-1}\mu\right)^\top x\right]-\frac{1}{2}\sum_{i=1}^Nm_i^\top C_i^{-1} m_i\\
    =&-\frac{1}{2}(x-\mu)^\top \Sigma^{-1}(x-\mu)+\frac{1}{2}\mu^\top\Sigma^{-1}\mu-\frac{1}{2}\sum_{i=1}^Nm_i^\top C_i^{-1} m_i,
\end{align}
\end{proof}

We now begin deriving the ELBO in \eqref{eq: efficient ELBO}. Using \eqref{eq: approx prior factorize and predictive prior}, we first rewrite the ELBO in \eqref{eq: with para ELBO} as follows 
\begin{align} \notag
    \mathcal{F}(q_n)=&\E_{q_n(\theta_n)q_n(X_n)}\log\frac{\hat{p}_{n|n-1}(X_n)p(Y_n|\theta_n,X_n)}{q_n(X_n)}\\ \notag&+\E_{q_n(\Lambda)}\E_{q_n(\theta_n)}\left[\log p(\theta_n|M_n,\Lambda)p(M_n|\Lambda)-\log q_n(\theta_n)\right]\\ \label{eq: ELBO rewrite}
    &+\E_{q_n(\Lambda)}\left[\log \hat{p}_{n|n-1}(\Lambda)-\log q_n(\Lambda)\right].
\end{align}
Recognise that the last line in \eqref{eq: ELBO rewrite} is $-\text{KL}(q_n(\Lambda)||\hat{p}_{n|n-1}(\Lambda))$ which can be computed as in \eqref{eq: kl for ELBO}. We now compute the second line in \eqref{eq: ELBO rewrite}.
First using \eqref{eq:assoc prior},\eqref{eq:single assoc prior}, \eqref{eq:measurement number pdf}, and \eqref{eq:vari theta independent}, we have
\begin{align} \notag
    &\E_{q_n(\theta_n)}\left[\log p(\theta_n|M_n,\Lambda)p(M_n|\Lambda)-\log q_n(\theta_n)\right]\\ \notag
    =& \log p(M_n|\Lambda) +\E_{q_n(\theta_n)}\sum_{j=1}^{M_n}\bigl[\log p(\theta_{n,j}|\Lambda)-\log q_n(\theta_{n,j})\bigr]\\ \notag
    =& \log p(M_n|\Lambda) +\sum_{j=1}^{M_n}\sum_{k=0}^K q_n(\theta_{n,j}=k)\log \frac{p(\theta_{n,j}=k|\Lambda)}{ q_n(\theta_{n,j}=k)}\\ \notag
    =&\sum_{j=1}^{M_n}\sum_{k=0}^K q_n(\theta_{n,j}=k)\log \frac{\Lambda_k}{ q_n(\theta_{n,j}=k)}-\Lambda_{\text{sum}}-\log(M_n!).
\end{align}
Proceeding to compute its expectation with respect to $q_n(\Lambda)$ in \eqref{eq:independent rate update}, we obtain the second line in \eqref{eq: ELBO rewrite} as follows
\begin{align} \notag
    &\E_{q_n(\Lambda)}\E_{q_n(\theta_n)}\left[\log p(\theta_n|M_n,\Lambda)p(M_n|\Lambda)-\log q_n(\theta_n)\right]\\ \label{eq:second line ELBO}
    =&\sum_{j=1}^{M_n}\sum_{k=0}^K q_n(\theta_{n,j}=k)\biggl(\log \frac{\rho_{n|n}^k}{ q_n(\theta_{n,j}=k)}+\psi(\eta_{n|n}^k) \biggr)-\sum_{k=0}^K \eta_{n|n}^k\rho_{n|n}^k-\log(M_n!).
\end{align}
Finally, we move on to the first line in \eqref{eq: ELBO rewrite}. We rewrite it as 
\begin{align} \notag &\E_{q_n(\theta_n)q_n(X_n)}\log\frac{\hat{p}_{n|n-1}(X_n)p(Y_n|\theta_n,X_n)}{q_n(X_n)}\\ \label{eq: elbo first line first part} 
=&\E_{q_n(X_n)}\bigl[-\log q_n(X_n)+\log \hat{p}_{n|n-1}(X_n)+\E_{q_n(\theta_n)}\log p(Y_n|\theta_n,X_n)\bigr].
\end{align}
Recall that when $q(X_n)$ is just updated according to \eqref{eq: X update}, we have
\begin{align} \label{eq: Xn update exact}
    q_n(X_n)=\frac{ \hat{p}_{n|n-1}(X_n)\text{exp}\left(\E_{q_n(\theta_n)}\log p(Y_n|\theta_n,X_n) \right)}{\int \hat{p}_{n|n-1}(X_n)\text{exp}\left(\E_{q_n(\theta_n)}\log p(Y_n|\theta_n,X_n)\right)dX_n },
\end{align}
where the denominator is the normalisation constant that does not depend on $X_n$. Substituting \eqref{eq: Xn update exact} into \eqref{eq: elbo first line first part} yields
\begin{align} \notag
    &\E_{q_n(\theta_n)q_n(X_n)}\log\frac{\hat{p}_{n|n-1}(X_n)p(Y_n|\theta_n,X_n)}{q_n(X_n)}\\ \notag
    =&\E_{q_n(X_n)}\bigl[ \log \int \hat{p}_{n|n-1}(X_n)\text{exp}\left(\E_{q_n(\theta_n)}\log p(Y_n|\theta_n,X_n)\right)dX_n   \bigr]\\ \label{eq: elbo first line third part}
    =& \log \int \hat{p}_{n|n-1}(X_n)\text{exp}\left(\E_{q_n(\theta_n)}\log p(Y_n|\theta_n,X_n)\right)dX_n ,
\end{align}
where the last line follows because the term in the bracket is a constant that does not depend on $X_n$.  
Now, to continue to evaluate \eqref{eq: elbo first line third part}, we first compute $\E_{q_n(\theta_n)}\log p(Y_n|\theta_n,X_n)$, which can be expressed as follows using \eqref{eq:obs prior} and \eqref{measurement model}
\begin{align} \notag
    &\E_{q_n(\theta_n)}\log p(Y_n|\theta_n,X_n)=\E_{q_n(\theta_n)} \sum_{j=1}^{M_n}\log \ell(Y_{n,j}|X_{n,\theta_{n,j}})=\sum_{j=1}^{M_n}\sum_{k=0}^{K}q_n(\theta_{n,j}=k)\log \ell(Y_{n,j}|X_{n,k})\\ \notag
    =&\sum_{j=1}^{M_n}\sum_{k=1}^{K}q_n(\theta_{n,j}=k)\biggl[  -\frac{1}{2}(Y_{n,j}-HX_{n,k})^\top R_k^{-1}(Y_{n,j}-HX_{n,k})-\frac{D}{2}\log 2\pi -\frac{1}{2}\log\det R_k \biggr]+ \sum_{j=1}^{M_n}q_n(\theta_{n,j}=0) \log \frac{1}{V}\\ \notag
    =&\sum_{k=1}^{K}\sum_{j=1}^{M_n}-\frac{1}{2}(Y_{n,j}-HX_{n,k})^\top \left(\frac{R_k}{q_n(\theta_{n,j}=k)}\right)^{-1}(Y_{n,j}-HX_{n,k})\\ \label{eq: elbo first line third part inner}
    &-\frac{1}{2}\sum_{j=1}^{M_n}\sum_{k=1}^{K}q_n(\theta_{n,j}=k)\log\det R_k + \sum_{j=1}^{M_n}q_n(\theta_{n,j}=0) \log \frac{1}{V}-\frac{D}{2}\log 2\pi\sum_{j=1}^{M_n}\sum_{k=1}^{K}q_n(\theta_{n,j}=k).
\end{align}
Note that we can rewrite the last term in \eqref{eq: elbo first line third part inner} as
\begin{align} \notag
    -\frac{D}{2}\log 2\pi&\sum_{j=1}^{M_n}\sum_{k=1}^{K}q_n(\theta_{n,j}=k)=-\frac{D}{2}\log 2\pi\sum_{j=1}^{M_n}(1-q_n(\theta_{n,j}=0))\\ \label{eq: elbo first line third part inner last term rewrite}
    &=\frac{D}{2}\log 2\pi\sum_{j=1}^{M_n}q_n(\theta_{n,j}=0)-\frac{DM_n}{2}\log2\pi.
\end{align}
Also, by using the formula for summation of quadratic terms in \eqref{eq:quadratric sum theorem special case} from Lemma \ref{lemma: quadratric sum lemma}, we have
\begin{align} \notag
    &\sum_{j=1}^{M_n}-\frac{1}{2}(Y_{n,j}-HX_{n,k})^\top \left(\frac{R_k}{q_n(\theta_{n,j}=k)}\right)^{-1}(Y_{n,j}-HX_{n,k})\\ \label{eq: sum quadratic apply}
    =&-\frac{1}{2}(\overbar{Y}_n^k-HX_{n,k})^\top{\overbar{R}_n^k}^{-1}(\overbar{Y}_n^k-HX_{n,k})   +\frac{1}{2}{\overbar{Y}_n^k}^\top{\overbar{R}_n^k}^{-1} \overbar{Y}_n^k-\frac{1}{2}\sum_{j=1}^{M_n}q_n(\theta_{n,j}=k)Y_{n,j}^\top R_k^{-1}  Y_{n,j}
\end{align}
with $\overbar{Y}_n^k,\overbar{R}_n^k$ defined in \eqref{eq:pseudomeas Y},\eqref{eq:pseudomeas R}. Now, by substituting both \eqref{eq: sum quadratic apply} and \eqref{eq: elbo first line third part inner last term rewrite} back to \eqref{eq: elbo first line third part inner}, we have
\begin{align}\notag
    &\E_{q_n(\theta_n)}\log p(Y_n|\theta_n,X_n)\\ \notag
    =&\sum_{k=1}^K\left[-\frac{1}{2}(\overbar{Y}_n^k-HX_{n,k})^\top{\overbar{R}_n^k}^{-1}(\overbar{Y}_n^k-HX_{n,k})   +\frac{1}{2}{\overbar{Y}_n^k}^\top{\overbar{R}_n^k}^{-1} \overbar{Y}_n^k-\frac{1}{2}\sum_{j=1}^{M_n}q_n(\theta_{n,j}=k)Y_{n,j}^\top R_k^{-1}  Y_{n,j}\right]\\ \notag
    &-\frac{1}{2}\sum_{j=1}^{M_n}\sum_{k=1}^{K}q_n(\theta_{n,j}=k)\log\det R_k + \sum_{j=1}^{M_n}q_n(\theta_{n,j}=0) \log \frac{1}{V}+\frac{D}{2}\log 2\pi\sum_{j=1}^{M_n}q_n(\theta_{n,j}=0)-\frac{DM_n}{2}\log2\pi\\ \notag
    =&\sum_{k=1}^K\left[-\frac{1}{2}(\overbar{Y}_n^k-HX_{n,k})^\top{\overbar{R}_n^k}^{-1}(\overbar{Y}_n^k-HX_{n,k}) -\frac{D}{2}\log2\pi-\frac{1}{2}\log\det \overbar{R}_n^k+\frac{D}{2}\log2\pi+\frac{1}{2}\log\det \overbar{R}_n^k \right]\\ \notag
    &+\frac{1}{2}\sum_{k=1}^K{\overbar{Y}_n^k}^\top{\overbar{R}_n^k}^{-1} \overbar{Y}_n^k-\frac{1}{2}\sum_{j=1}^{M_n}\sum_{k=1}^Kq_n(\theta_{n,j}=k)\left(Y_{n,j}^\top R_{k}^{-1}Y_{n,j}+\log\det R_k\right)\\ \notag
    &+\left(\frac{D}{2}\log 2\pi+\log\frac{1}{V}\right)\sum_{j=1}^{M_n}q_n(\theta_{n,j}=0)-\frac{DM_n}{2}\log 2\pi\\ \label{eq: elbo first line third part inner final} 
    =&\sum_{k=1}^K\log \mathcal{N}(\overbar{Y}_n^k;HX_{n,k},\overbar{R}_n^k) +C_x,
\end{align}
where $C_x$ is a constant that does not depend on $X_{n,k}$, which is defined as
\begin{align} \label{eq:elbo first line third part inner final constant} 
\begin{aligned}
    C_x=&\frac{DK}{2}\log2\pi+\frac{1}{2}\sum_{k=1}^K\left(\log\det \overbar{R}_n^k+ {\overbar{Y}_n^k}^\top{\overbar{R}_n^k}^{-1} \overbar{Y}_n^k\right)-\frac{1}{2}\sum_{j=1}^{M_n}\sum_{k=1}^Kq_n(\theta_{n,j}=k)\left(Y_{n,j}^\top R_{k}^{-1}Y_{n,j}+\log\det R_k\right)\\ 
    &+ \left(\frac{D}{2}\log 2\pi+\log\frac{1}{V}\right)\sum_{j=1}^{M_n}q_n(\theta_{n,j}=0)-\frac{DM_n}{2}\log 2\pi
\end{aligned}
\end{align}
Now substitute \eqref{eq: elbo first line third part inner final} to \eqref{eq: elbo first line third part}, we have
\begin{align} \notag
    &\log \int \hat{p}_{n|n-1}(X_n) \exp\E_{q_n(\theta_n)}\log p(Y_n|\theta_n,X_n)dX_n\\ \notag
    =&\log \int \hat{p}_{n|n-1}(X_n) \exp\left(\sum_{k=1}^K\log \mathcal{N}(\overbar{Y}_n^k;HX_{n,k},\overbar{R}_n^k) +C_x \right)dX_n\\ \notag
    =&\log\left(\exp C_x\int \prod_{i=1}^K \hat{p}_{n|n-1}(X_{n,i}) \prod_{k=1}^K \mathcal{N}(\overbar{Y}_n^k;HX_{n,k},\overbar{R}_n^k)dX_n\right)\\ \label{eq: elbo first line fourth part}
    =&C_x+ \sum_{k=1}^K \log \int \hat{p}_{n|n-1}(X_{n,k})\mathcal{N}(\overbar{Y}_n^k;HX_{n,k},\overbar{R}_n^k)dX_{n,k},
\end{align}
where the second equality in \eqref{eq: elbo first line fourth part} is obtained by using the fact that we always have $\hat{p}_{n|n-1}(X_n)=\prod_{k=1}^K\hat{p}_{n|n-1}(X_{n,k})$ for our assumed mean-field family. Note that with a Gaussian $\hat{p}_{n|n-1}(X_{n,k})$ in \eqref{eq:state predictive prior computation}, the marginal likelihood in \eqref{eq: elbo first line fourth part} can be easily computed in a Gaussian form, i.e.
\begin{align} \notag
    \int \hat{p}_{n|n-1}(X_{n,k})\mathcal{N}&\left(\overbar{Y}_n^k;HX_{n,k},\overbar{R}_n^k\right)dX_{n,k}=\mathcal{N}\left(\overbar{Y}_n^k;H\mu^{k*}_{n|n-1},H\Sigma^{k*}_{n|n-1}H^\top+\overbar{R}_n^k \right)\\ \label{eq: fourth part marginal likelihood} 
    &=\exp\left[ -\frac{1}{2}{T_n^k}^\top{S_n^k}^{-1}T_n^k-\frac{D}{2}\log 2\pi -\frac{1}{2}\log\det S_n^k \right],
\end{align}
where $\mu^{k*}_{n|n-1},\Sigma^{k*}_{n|n-1}$ are given in \eqref{eq:state predictive prior computation}. $T_n^k,S_n^k$ are defined in \eqref{eq:KF update} as by-products of Kalman filter updates. Now substitute \eqref{eq: fourth part marginal likelihood}, and \eqref{eq:elbo first line third part inner final constant} to the \eqref{eq: elbo first line fourth part}, we have
\begin{align} \notag
    &C_x+ \sum_{k=1}^K \log \int \hat{p}_{n|n-1}(X_{n,k})\mathcal{N}(\overbar{Y}_n^k;HX_{n,k},\overbar{R}_n^k)dX_{n,k}\\ \notag
    =&\sum_{k=1}^K \left[ -\frac{1}{2}{T_n^k}^\top{S_n^k}^{-1}T_n^k-\frac{D}{2}\log 2\pi -\frac{1}{2}\log\det S_n^k \right]+\frac{DK}{2}\log2\pi+\frac{1}{2}\sum_{k=1}^K\left(\log\det \overbar{R}_n^k+ {\overbar{Y}_n^k}^\top{\overbar{R}_n^k}^{-1} \overbar{Y}_n^k\right)\\ \notag
    &-\frac{1}{2}\sum_{j=1}^{M_n}\sum_{k=1}^Kq_n(\theta_{n,j}=k)\left(Y_{n,j}^\top R_{k}^{-1}Y_{n,j}+\log\det R_k\right)+ \left(\frac{D}{2}\log 2\pi+\log\frac{1}{V}\right)\sum_{j=1}^{M_n}q_n(\theta_{n,j}=0)-\frac{DM_n}{2}\log 2\pi\\ \notag
    =&\frac{1}{2}\sum_{k=1}^K\left({\overbar{Y}_n^k}^\top{\overbar{R}_n^k}^{-1}\overbar{Y}_n^k-{T_n^k}^\top{S_n^k}^{-1}T_n^k+\log \frac{\det \overbar{R}_n^k}{\det S_n^k} \right)-\frac{1}{2}\sum_{j=1}^{M_n}\sum_{k=1}^Kq_n(\theta_{n,j}=k)\left(Y_{n,j}^\top R_{k}^{-1}Y_{n,j}+\log\det R_k\right)\\ \label{eq:elbo first line result}
    &+ \left(\frac{D}{2}\log 2\pi+\log\frac{1}{V}\right)\sum_{j=1}^{M_n}q_n(\theta_{n,j}=0)-\frac{DM_n}{2}\log 2\pi
\end{align}
where \eqref{eq:elbo first line result} gives the final result of the first line of the ELBO in \eqref{eq: ELBO rewrite}, i.e. equation \eqref{eq: elbo first line third part} that computes $\E_{q_n(\theta_n)q_n(X_n)}\log\frac{\hat{p}_{n|n-1}(X_n)p(Y_n|\theta_n,X_n)}{q_n(X_n)}$, in which $q(X_n)$ is the latest updated variational distribution.  

In this way, the complete ELBO in \eqref{eq: ELBO rewrite} is derived with final result given in \eqref{eq: efficient ELBO}, where the first line of equation \eqref{eq: ELBO rewrite} is given in \eqref{eq:elbo first line result}, the second line of equation \eqref{eq: ELBO rewrite} corresponds to the solution in \eqref{eq:second line ELBO}, and the last line of equation \eqref{eq: ELBO rewrite}  is calculated in \eqref{eq: kl for ELBO}.
\section{} \label{apx: reloc elbo derivation}
In this appendix, we derive the explicit form of the ELBO in \eqref{eq: reloc ELBO} up to an additive constant for implementing the relocation strategy. In particular, we compute $\mathcal{F}_{n,h}(q_n(\theta_n),q_n(X_{n,h}))-c$, where $c$ is the constant from \eqref{eq: reloc ELBO}. This can be expressed as follows, based on \eqref{eq: reloc ELBO}:
\begin{align}  \notag
    \mathcal{F}_{n,h}(q_n(\theta_n),q_n(X_{n,h}))-c=& -\text{KL}(q_n(\theta_n)||p(\theta_n|M_n,\Lambda))-\text{KL}(q_n(X_{n,h})||\tilde{p}_n(X_{n,h}))\\ \label{eq: reloc elbo minus const}
    &+\E_{q_n(\theta_n)q_n(X_{n,h})q^*_n(X_{n,h-})}\log p(Y_n|\theta_n,X_n).
\end{align}
The first negative KL divergence is 
\begin{align} \notag
    -\text{KL}(q_n(\theta_n)||p(\theta_n|M_n,\Lambda))=&\sum_{j=1}^{M_n}\sum_{k=0}^K q_n(\theta_{n,j}=k)\log\frac{p(\theta_{n,j}=k|M_n,\Lambda)}{q_n(\theta_{n,j}=k)} \\ \label{eq:elbo reloc theta kl}
    =&\sum_{j=1}^{M_n}\sum_{k=0}^K q_n(\theta_{n,j}=k)\log\frac{\Lambda_k}{\Lambda_{\text{sum}}q_n(\theta_{n,j}=k)}.
\end{align}
Denote the mean and covariance of $\tilde{p}_n(X_{n,h})$ as $\tilde{\mu}_{n,h}$ and $\tilde\Sigma_{n,h}$; and recall that $q_n(X_{n,h})=\mathcal{N}(X_{n,h};\mu_{n|n}^h,\Sigma_{n|n}^h)$.Then the second negative KL divergence between two multivariate normal in \eqref{eq: reloc elbo minus const} is
\begin{align} \label{eq:elbo reloc gaussian kl}
    -\text{KL}(q_n(X_{n,h})||\tilde{p}_n(X_{n,h}))=-\frac{1}{2}\left[\Tr({\tilde{\Sigma}_{n,h}}^{-1}\Sigma_{n|n}^{h})+(\tilde{\mu}_{n,h}-\mu_{n|n}^h)^\top 
{\tilde{\Sigma}_{n,h}}^{-1} (\tilde{\mu}_{n,h}-\mu_{n|n}^h)-D_x+\log\frac{\det\tilde{\Sigma}_{n,h} }{\det \Sigma_{n|n}^{h}} \right],
\end{align}
where $D_x$ is the dimension of $X_{n,h}$. Note evaluating this divergence requires inverting $\tilde{\Sigma}_{n,h}$. Since $\tilde{\Sigma}_{n,h}$ is an uninformative prior assigned for the missed object, we can always let it be a diagonal matrix, which leads to a simple and efficient inversion.

By using the formula for the expectation of quadratic forms, the last term in \eqref{eq: reloc elbo minus const} is
\begin{align} \notag
    &\E_{q_n(\theta_n)q_n(X_{n,h})q^*_n(X_{n,h-})}\log p(Y_n|\theta_n,X_n)=\sum_{j=1}^{M_n}q_n(\theta_{n,j}=0)\log\frac{1}{V}\\ \notag
    &\ \ \ \ \ +\sum_{j=1}^{M_n}\sum_{k=0}^K q_n(\theta_{n,j}=k)\left[-\frac{1}{2}\E_{q_n(X_{n,h})q^*_n(X_{n,h-})}(Y_{n,j}-HX_{n,k})^\top R_k^{-1} (Y_{n,j}-HX_{n,k})-\frac{D}{2}\log 2\pi -\frac{1}{2}\log\det R_k\right]\\ \notag
    =&-\frac{1}{2}\sum_{j=1}^{M_n}\sum_{k=0,k\neq h}^K q_n(\theta_{n,j}=k)U_n^{k*}\\ \notag
    &+\sum_{j=1}^{M_n} q_n(\theta_{n,j}=h)\left[-\frac{1}{2}(Y_{n,j}-H\mu_{n|n}^{h})^\top R_k^{-1} (Y_{n,j}-H\mu_{n|n}^{h})-\frac{1}{2}\Tr(R_k^{-1}H\Sigma_{n|n}^{h}H^\top)-\frac{1}{2}\log\det R_k\right]\\ \label{eq:elbo reloc last term}
    &+\left(\frac{D}{2}\log2\pi+\log\frac{1}{V}\right)\sum_{j=1}^{M_n}q_n(\theta_{n,j}=0)-\frac{DM_n}{2}\log2\pi,
\end{align}
where we use \eqref{eq: elbo first line third part inner last term rewrite} to obtain the last line, and $U_n^{k*}$ is defined as
\begin{align} \label{eq: Un define}
    U_n^{k*}=(Y_{n,j}-H\mu_{n|n}^{k*})^\top R_k^{-1} (Y_{n,j}-H\mu_{n|n}^{k*})+\Tr(R_k^{-1}H\Sigma_{n|n}^{k*}H^\top)+\log\det R_k.
\end{align}
Finally the ELBO (with an additive constant) in \eqref{eq: reloc elbo minus const} is the summation of \eqref{eq:elbo reloc theta kl}, \eqref{eq:elbo reloc gaussian kl}, and \eqref{eq:elbo reloc last term}, i.e.
\begin{align}\notag
    \mathcal{F}_{n,h}(q_n(\theta_n),q_n(X_{n,h}&))-c= \sum_{j=1}^{M_n}\sum_{k=0}^K q_n(\theta_{n,j}=k)\log\frac{\Lambda_k}{\Lambda_{\text{sum}}q_n(\theta_{n,j}=k)}-\frac{1}{2}\sum_{j=1}^{M_n}\sum_{k=0,k\neq h}^K q_n(\theta_{n,j}=k)U_n^{k*}\\ \notag
    &-\frac{1}{2}\sum_{j=1}^{M_n} q_n(\theta_{n,j}=h)\left[(Y_{n,j}-H\mu_{n|n}^{h})^\top R_k^{-1} (Y_{n,j}-H\mu_{n|n}^{h})+\Tr(R_k^{-1}H\Sigma_{n|n}^{h}H^\top)+\log\det R_k\right]\\ \notag
    &-\frac{1}{2}\left[\Tr({\tilde{\Sigma}_{n,h}}^{-1}\Sigma_{n|n}^{h})+(\tilde{\mu}_{n,h}-\mu_{n|n}^h)^\top 
    {\tilde{\Sigma}_{n,h}}^{-1} (\tilde{\mu}_{n,h}-\mu_{n|n}^h)+\log\frac{\det\tilde{\Sigma}_{n,h} }{\det \Sigma_{n|n}^{h}} \right]\\
    &+\left(\frac{D}{2}\log2\pi+\log\frac{1}{V}\right)\sum_{j=1}^{M_n}q_n(\theta_{n,j}=0)+\frac{D_x}{2}-\frac{DM_n}{2}\log2\pi,
\end{align}
where $U_n^{k*}$ is given in \eqref{eq: Un define}, and the last two terms in the final line are constants that can be omitted when computing the ELBO. Additionally, $U_n^{k*}(k=1,2,...,K)$ is also a constant when implementing the relocation strategy, and it can be precomputed using the output from the standard tracker.
\end{document}